\documentclass[aps,10pt,groupedaddress,prx,superscriptaddress,twocolumn,balancelastpage, tightenlines]{revtex4-2}

\usepackage{env} 
\usepackage{soul}
\usepackage{subcaption}
\usepackage{tocloft}
\usepackage{thmtools} 
\usepackage{thm-restate}

\usepackage{titletoc}
\newcommand\DoToC{%
  \startcontents
  \printcontents{}{1}{\textbf{Contents}\vskip3pt\hrule\vskip5pt}
  \vskip3pt\hrule\vskip5pt
}

\theoremstyle{plain}  % Bold label, upright body (for definitions, examples)
\newtheorem{simtech}{Simulation Technique}

\crefname{simtech}{Simulation Technique}{Simulation Techniques}

\hypersetup{
    colorlinks=true,
    linkcolor=navy,
    citecolor=navy,
    urlcolor=navy,
    breaklinks=true,
}

\definecolor{HeaderRowColor}{RGB}{220, 234, 247}
\definecolor{FirstColumnColor}{RGB}{238, 243, 248}
\definecolor{EmphasisColorGreen}{HTML}{DCEFD8}
\definecolor{EmphasisColorRed}{HTML}{F6D6D1}
\definecolor{EmphasisColorPurple}{HTML}{E1D5F0}
\definecolor{tabblue}{HTML}{1f77b4}
\definecolor{tabred}{HTML}{d62728}
\definecolor{inpA}{rgb}{0.785, 0.544, 0.585}
\definecolor{inpB}{rgb}{0.770, 0.550, 0.596}
\definecolor{inpC}{rgb}{0.740, 0.563, 0.619}
\definecolor{inpD}{rgb}{0.706, 0.578, 0.645}
\definecolor{inpE}{rgb}{0.670, 0.594, 0.673}
\definecolor{inpF}{rgb}{0.633, 0.610, 0.702}
\definecolor{inpG}{rgb}{0.593, 0.628, 0.732}
\definecolor{inpH}{rgb}{0.556, 0.644, 0.760}
\definecolor{inpI}{rgb}{0.512, 0.663, 0.795}
\definecolor{inpJ}{rgb}{0.473, 0.680, 0.824}
\definecolor{inpK}{rgb}{0.658, 0.599, 0.682}

\newcommand{\diagfillcustom}[5]{%
\begin{tikzpicture}[baseline=(n.base)]
\node[
    anchor=center,
    minimum width=7em
] (n) {\strut #1};
\fill[#2] (n.south west)--(n.north west)--(n.south east)--cycle;
\fill[#3] (n.north west)--(n.north east)--(n.south east)--cycle;
\node[anchor=south west, inner sep=2pt] at (n.south west) {\small #4};
\node[anchor=north east, inner sep=2pt] at (n.north east) {\small #5};
\node[anchor=center] {#1};
\end{tikzpicture}%
}

\newcommand{\solidcelllabelled}[3]{%
\begin{tikzpicture}[baseline=(n.center)]
  \node[anchor=center, inner sep=0pt,
        minimum width=7em, minimum height=5.7em] (n) {};
  \fill[#1] (n.south west) rectangle (n.north east);   % solid background first
  \node[anchor=south west, inner sep=2pt] at (n.south west) {\small #2};
  \node[anchor=north east, inner sep=2pt] at (n.north east) {\small #3};
\end{tikzpicture}%
}

\newcommand{\diagfillcustomhighlighted}[5]{%
\begin{tikzpicture}[baseline=(n.base)]
\node[
    anchor=center,
    minimum width=7em
] (n) {\strut #1};
\fill[#2] (n.south west)--(n.north west)--(n.south east)--cycle;
\fill[#3] (n.north west)--(n.north east)--(n.south east)--cycle;
\draw[line width=1.5pt, orange!90!black]
    ($(n.south west)+(0.5pt,0.5pt)$) rectangle ($(n.north east)+(-0.5pt,-0.5pt)$);
\node[anchor=south west, inner sep=2pt] at (n.south west) {\small #4};
\node[anchor=north east, inner sep=2pt] at (n.north east) {\small #5};
\node[anchor=center] {#1};
\pgfresetboundingbox
\path[use as bounding box] (n.south west) rectangle (n.north east);
\end{tikzpicture}%
}

\makeatletter
\renewcommand\@fnsymbol[1]{\ensuremath{\ifcase#1\or *\or\dagger\or\ddagger\or\mathsection\or\mathparagraph\else\fi}}
\makeatother

\begin{document}

\title{Classical simulation and model concentration in passive linear optics}

\author{Léo Monbroussou}\email{leo.monbroussou@gmail.com}
\affiliation{School of Informatics, University of Edinburgh, Edinburgh, United Kingdom}
\affiliation{Terra Quantum AG, Kornhausstrasse 25, St.~Gallen, 9000, Switzerland}

\author{Hugo Thomas}
\thanks{These authors contributed equally to this work.}
\affiliation{Laboratoire d’Informatique de Paris 6, CNRS, Sorbonne Université, Paris, France}
\affiliation{Quandela, 7 rue Léonard de Vinci, Massy, France}
\affiliation{DIENS, \'Ecole Normale Supérieure, PSL University, CNRS, INRIA, Paris, France}

\author{Hela Mhiri}
\thanks{These authors contributed equally to this work.}
\affiliation{Laboratoire d’Informatique de Paris 6, CNRS, Sorbonne Université, Paris, France}
\affiliation{Institute of Physics, Ecole Polytechnique Fédérale de Lausanne (EPFL), Lausanne, Switzerland}

\author{Zoë Holmes}
\affiliation{Institute of Physics, Ecole Polytechnique Fédérale de Lausanne (EPFL), Lausanne, Switzerland}
\affiliation{Centre for Quantum Science and Engineering, École Polytechnique Fédérale de Lausanne (EPFL), Lausanne, Switzerland}

\author{Elham Kashefi}
\affiliation{School of Informatics, University of Edinburgh, Edinburgh, United Kingdom}
\affiliation{Laboratoire d’Informatique de Paris 6, CNRS, Sorbonne Université, Paris, France}

\begin{abstract}
    Passive linear optics is a restricted model of quantum computation, with complexity-theoretic evidence of quantum advantage for sampling tasks and low losses that make it attractive for near-term algorithms. In qubit architectures, a growing body of work has revealed a close connection between barren plateaus and classical simulability. Whether an analogous tradeoff exists for bosonic systems remains largely unexplored. Building on a recently developed representation-theoretic framework for moments of random passive linear-optical circuits, we characterize the concentration of expectation values for relevant families of particle-number-preserving observables by explicitly evaluating their projections into irreducible representations of the unitary group and analyzing their asymptotic scaling. We show that concentration is governed by the misalignment of the projections into irreducible representations of the input state and the observable, providing a unified representation-theoretic interpretation of generalized entanglement and locality in the bosonic setting. We further relate these concentration properties to existing classical simulation techniques, identifying broad classes of trainable observables that admit efficient classical simulation. In parallel, we identify Fock-state inputs and observables that evade exponential concentration. 
    We identify regimes that appear to evade exponential concentration while retaining a polynomially large signal component that is not accessible to known efficient classical simulation methods.
    However, the separation is only partial: most of the signal remains classically tractable, and the residual part, while not exponentially suppressed, is small enough that a truncation serves as a classical surrogate with polynomially small error. Our framework nonetheless provides a systematic route for searching for regimes that unambiguously combine the absence of exponential concentration and lies beyond known efficient classical simulation methods. %Although we do not identify a clear example that both avoids exponential concentration and unambiguously lies beyond known efficient classical simulation methods, our framework provides a systematic route for searching for such regimes.
\end{abstract}

\maketitle

\begin{figure*}[t!]
    \centering
    \includegraphics[width=1.0\linewidth]{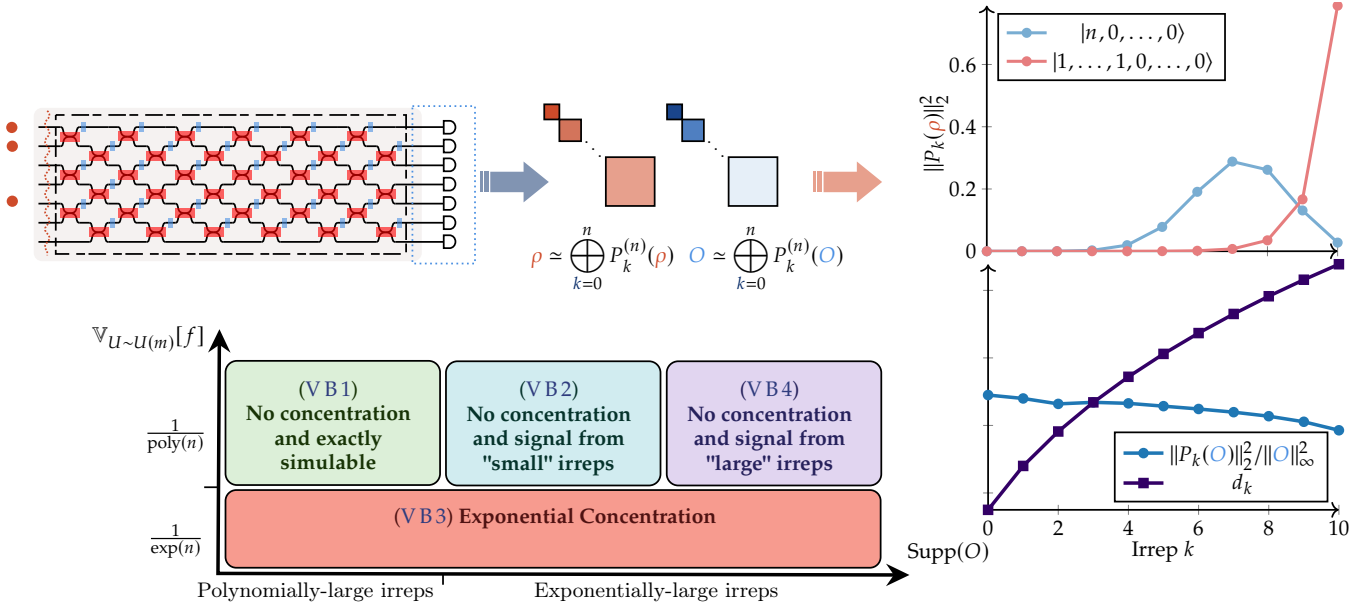}
    \caption{\justifying \textbf{Summary of the main results}: we study passive linear-optical models through the lens of the decomposition of its operator into irreducible components. By adapting the iterative procedure from \protect\cite{mhiri_boson_2026}, we study the irreducible component purities of the input state $\rho$ and observable $O$ in order to study the concentration of the model. We present the example of Fock-state purities and the number-phase observable introduced in \cref{subsec:Analytical_Study_Operators} compared with the irreducible dimension, for $n=10$ photons and $m=20$ modes. Finally, we compare our results with simulation techniques, and we identify four regimes of interest.}
    \label{fig:Introduction_Figure}
\end{figure*}

\section{Introduction}\label{sec:Introduction}

Passive linear optics has emerged as a promising model of computation for near term applications of quantum computing. It is a restricted but experimentally mature model of quantum computation: photons evolve through networks of beam splitters and phase shifters with low losses, room-temperature operations, and demonstrated large-scale integration~\cite{wang_integrated_2020, taballione_universal_2021, bourassa_blueprint_2021, aghaeerad_scaling_2025, alexander_manufacturable_2025, maring_versatile_2024}.

Variational quantum algorithms, in which a parametrized quantum circuit is trained by a classical optimizer, are a flexible paradigm with countless proposed applications~\cite{preskill_quantum_2018, cerezo_variational_2021, biamonte_quantum_2017}. However, their viability is threatened by barren plateaus (BPs), where loss functions concentrate exponentially with system size, causing gradients to vanish and training with polynomially many measurements to fail~\cite{larocca_barren_2025, mcclean_barren_2018, arrasmith_equivalence_2022}.
BPs find their origins in several interrelated sources, including circuit expressibility~\cite{holmes_connecting_2022}, global cost functions~\cite{cerezo_cost_2021}, entanglement of the input state~\cite{ortiz_marrero_entanglement-induced_2021}, and hardware noise~\cite{wang_noise-induced_2021}. 
For circuits whose generators close on a dynamical Lie algebra, these diagnoses have recently been unified into an exact Lie-algebraic theory of loss concentration~\cite{ragone_lie_2024, fontana_characterizing_2024}.

In qubit architectures, it has recently been argued that the structures used to establish trainability typically also enable classical simulation~\cite{toth_permutationally_2010, uvarov_barren_2021,cerezo_cost_2021,zhao_analyzing_2021,liu_presence_2022,larocca_diagnosing_2022,fontana_characterizing_2024,zhang_absence_2024,ragone_lie_2024,monbroussou_trainability_2025,cerezo_does_2025,angrisani2025simulatingquantumcircuitsarbitrary}. The central intuition is that BPs reflect a curse of dimensionality: exploring an exponentially large space causes loss functions to concentrate~\cite{cerezo_does_2025}. Methods that avoid this concentration generally restrict the dynamics to a smaller \emph{effective} subspace, which in turn makes the model classically simulable. These results have so far focused on qubit and fermionic architectures, leaving open whether the same link between trainability and simulability persists in passive linear optics. 
 
In this work, we build on recent results from~\cite{mhiri_boson_2026} to analyse closed-form expressions for the variance of expectation values of particle-number-preserving observables under random linear optical interferometers. We 
first recall how to express and evaluate the variance in terms of the projection norm of the input state and observable onto irreducible representations of the unitary group on $m$ modes, which we denote by irrep purities. We then analyze their dependence on photon bunching for Fock state inputs, and develop construction rules connecting observable structure to irrep support .
We interpret the state irrep purities in terms of \emph{generalized entanglement} and the observable irrep purities as a representation-theoretic notion of \emph{generalized locality}, and show how their asymptotic scaling governs concentration~\cite{larocca_diagnosing_2022,diaz2023showcasingbarrenplateautheory}. Intriguingly, in contrast to the usual ``curse of dimension'' story in the qubit case~\cite{cerezo_does_2025}, it is seemingly possible to identify observables where there is a non-vanishing signal from large-dimensional irreps.

We then compare these concentration regimes with existing average-case and worst-case, exact and approximate classical simulation methods, distinguishing representation-theoretic techniques from algorithms that exploit the algebraic structure of passive linear optics.

Although we identify settings that evade exponential concentration as well as specific simulation mechanisms, we do not find a clear example that simultaneously avoids exponential concentration and unambiguously lies beyond known efficient classical simulation methods. In examples we obtain using number-phase operators, most of the signal remains classically tractable, and the residual part, while not exponentially suppressed, is small enough that a truncation serves as a classical surrogate with polynomially small error. Our analysis nonetheless provides a systematic route for searching for such regimes, and highlights potential targets for new approximate simulation algorithms.

\section{Passive Linear Optics}\label{subsubsec:PLO}

The Hilbert space of states of indistinguishable photons in $m$ modes is the bosonic Fock space $\mathcal{F} = \bigoplus_{n \geq 0} \mathcal{F}_n$, where each subspace $\mathcal{F}_n = \mathrm{Sym}^n(\mathbb{C}^m)$ carries states of $n$ identical particles. This symmetric structure is realized using bosonic creation and annihilation operators, that are defined on $\mathcal{F}_n$ by
\begin{equation}
    \begin{split}
        a_i^\dagger \ket{s_1, \dots, s_j, \dots, s_m} &= \sqrt{s_j +1} \ket{s_1, \dots, s_j + 1, \dots, s_m}\,, \\
        a_i \ket{s_1, \dots, s_j, \dots, s_m} &= \sqrt{s_j} \ket{s_1, \dots, s_j - 1, \dots, s_m}\,,
    \end{split} 
\end{equation}
and satisfy the commutation relations $[a_i,a_j^\dagger]= \delta_{i,j}\Id$, and $[a_i,a_j]=[a_i^\dagger,a_j^\dagger]=0$. An orthonormal basis of $\mathcal{F}_n$ is given by occupation-number states, or Fock states
\begin{equation}
    \ket{S} = \prod_{i=1}^m \frac{(a_i^\dagger)^{s_i}}{\sqrt{s_i!}}\ket{\boldsymbol{0}}, 
    \quad \sum_{i=1}^m s_i = n,
\end{equation}
where $\ket{\boldsymbol{0}}$ corresponds to the $m$-mode vacuum state. These occupation-number states are
labelled by configurations $S \in \Phi_m^n$ defined as
\begin{equation}
    \Phi_m^n = \left\{(s_1,\dots,s_m)\ |\ \sum_{i=1}^m s_i=n\right\} \;,
\end{equation}
with $|\Phi_m^n| =  \binom{n+m-1}{n}$. Passive linear-optical transformations are generated by quadratic Hamiltonians of the form 
\begin{equation}
    \hat H = \sum_{i,j = 1}^m h_{i,j}\hat a_i^\dag \hat a_j,
\end{equation}
where $h_{i,j} = h_{j,i}^*$ since $\hat H$ is Hermitian. They can be physically implemented with beam splitters and phase shifters~\cite{makarov_theory_2022,pan_multiphoton_2012, clements_optimal_2016}. These operators form the Lie algebra $\mathfrak{u}(m)$ under commutation with their Hermitian conjugate. Exponentiation yields the unitary group $U(m)$, which characterizes all passive linear-optical transformations on $m$ modes~\cite{garcia-escartin_multiple_2019}.

Concretely, a linear interferometer is fully characterized by its action on the single-particle space $\mathcal{F}_1 = \mathbb{C}^m$, given by an $m \times m$ unitary matrix. This action lifts to the full Fock space via a unitary representation, known as the \emph{photonic homomorphism}~\cite{aaronson_computational_2011},
\begin{equation}\label{eq:photonic_homomorphism}
    \varphi: U(m) \rightarrow U(\mathcal{F})\;,
\end{equation}
with $\mathcal{F} = \bigoplus_{n\geq0} \mathcal{F}_n$ the bosonic Fock space, where each $\mathcal{F}_n = \mathrm{Sym}^n(\mathbb{C}^m)$ describes configurations of $n$ identical particles and $U(\mathcal{F})\subseteq U\left(|\Phi_m^n|\right)$ is the set of unitary transformations on the Fock space $\mathcal F$. Particle-number conservation implies that $\varphi$ can be decomposed into orthogonal, finite-dimensional unitary representations, where each block represents a fixed number of particles $n$:
\begin{equation}\label{eq:phi_n}
    \varphi = \bigoplus_{n \geq 0} \varphi_n\;,   
\end{equation}
where $\varphi_n$
is the maximally symmetric  irreducible representation of the group $U(m)$ in the $n$-particle Fock space $\mathcal{F}_n$~\cite{aniello_exploring_2006}.

In this work, we study expectation values of \emph{particle-number-preserving observables} on Fock states evolved under passive linear-optical transformations.
We  recall that the algebra of bosonic particle-number-preserving operators is spanned by normally ordered monomials containing equal numbers of creation and annihilation operators,
\begin{equation}\label{eq:subalgebra_commute}
    \mathcal{W} = {\rm Span} \{  a_{j_1}^\dagger \dots a_{j_d}^\dagger a_{i_1} \dots a_{i_d} \;|\;  d\geq 0\}\;.
 \end{equation}
 By construction, every operators in 
$\mathcal{W}$ commutes with the total number operator and therefore preserves each fixed-particle-number sector of the Fock space. Consequently,
\begin{equation}\label{eq:decomp_big_W}
    \mathcal{W} \simeq \bigoplus_{n \geq 0} W_n\;,
\end{equation}
where $W_n = \mathcal{F}_n \otimes \mathcal{F}_n^*$ denotes the  space of operators acting on the $n$--photon sector  $\mathcal{F}_n$. 

Throughout this work, we consider an initial state \(\rho\) with a fixed
photon number \(n\), supported on \(\mathcal F_n\). Hence, for
\(O\in\mathcal W\), its expectation value after a passive linear-optical
transformation \(U\) is
\begin{equation}\label{eq:linear_form_n}
    f_U(\rho,O)
    =
    \Tr\!\left[
        \rho\,
        \varphi_n(U)^\dagger
        O^{(n)}
        \varphi_n(U)
    \right],
\end{equation}
where \(O^{(n)}\) denotes the restriction of \(O\) to the \(n\)-photon
sector \(\mathcal F_n\). Moreover, we fix $m = \lceil2.1n\rceil$, in light of recent hardness results of Boson Sampling in this regime \cite{bouland_complexitytheoretic_2023}.

Although the dynamics within each fixed-particle-number sector are described by the representation $\varphi_n(U)$ of the low dimensional Lie group $U(m)$, evaluating expectation values of the form in \cref{eq:linear_form_n} can nevertheless be computationally intractable when the input state and/or the observable are not generated by passive linear optics. A canonical example is provided by Boson Sampling, where both the input state and the measurement are projectors onto Fock states. In this case, the output probabilities are given by
\begin{equation}
    \Tr[\ketbra{T}{T} \varphi_n(U)^\dagger \ketbra{S}{S}  \varphi_n(U)] = \frac{|\per{U_{S, T}}|^2}{\prod_{i=1}^{m}s_i! \prod_{j=1}^{m}t_j!},
\end{equation}
where $U_{S, T}$ is obtained from $U$ by taking $s_i$ copies of the $i$-th row
and $t_j$ copies of the $j$-th column~\cite{aaronson_computational_2011}, and
${\text{Per}(A) = \sum_{\sigma \in S_n} \prod_{i=1}^n A_{i,\sigma(i)}}$ is the matrix permanent function for an $n \times n$ matrix~\cite{marcus_permanents_1965}. Computing matrix permanents is $\#P-$hard in general, making these transition probabilities computationally intractable in the worst case. This hardness underpins the computational complexity of Boson Sampling and motivates the study of the more general class of expectation values in \cref{eq:linear_form_n}, of which Boson Sampling probabilities are a particular instance.

\section{Concentration of observable expectation values}

\subsection{Framework}
Since passive linear-optical transformations are described by  the compact Lie group $U(m)$, the canonical notion of sampling a random interferometer corresponds to drawing a unitary matrix $U$ from the Haar measure, which we denote by $U \sim U(m)$. The Haar measure, defined as the unique probability measure on $U(m)$ that is invariant under both left and right group multiplication \cite{collins_integration_2006}, provides the natural notion of a uniformly random interferometer used throughout Boson Sampling and related photonic quantum advantage proposals \cite{aaronson_computational_2011,hamilton_gaussian_2017,oszmaniec_fermion_2022}.

Linear-optical interferometers are specified by $m^2$ continuous parameters through an appropriate parametrization of the group, for example via generalized Euler-angle or Reck--Clements decompositions~\cite{ carolan_universal_2015, moralis-pegios_perfect_2024, clements_optimal_2016}. Consequently, sampling these parameters according to the probability distribution induced by the Haar measure is equivalent to sampling a Haar-random interferometer and can be done efficiently \cite{mezzadri_how_2007}. This correspondence naturally links our analysis to parametrized photonic circuits arising in variational quantum algorithms and quantum machine learning, where average-case properties are commonly studied through Haar-random ensembles.

In this work, we are interested in the concentration of expectation values of bosonic observables around their  mean, on average over the Haar measure. We say that an expectation value exhibits \emph{exponential concentration} if the probability that it deviates from its mean by more than a constant $\varepsilon > 0$ decays exponentially with the system size. By Chebyshev's inequality, this condition reads
\begin{equation}
    \Pr[\left|f_U - \e[U \sim U(m)]{f_U}\right|\geq \varepsilon] \leq \frac{\var[U \sim U(m)]{f_U}}{\varepsilon^2},
\end{equation}
so establishing exponential concentration reduces to proving that the variance of $f_U$, and hence its second moment,
decays exponentially with the system size, i.e.
\begin{equation}
\e[U \sim U(m)]{f_U^2} -  \e[U \sim U(m)]{f_U}\ ^{\!\!2} \in \mathcal{O}(\exp(-n)),
\end{equation}
where we have used the notation of \cref{eq:linear_form_n}
for the expected value without specifying the state and observable.

A key observation to compute second moments over random interferometers is that the adjoint action of $U(m)$ partitions the fixed particle-number operator space $W_n$ into irreducible representations (irreps), i.e.,

\begin{align}
     & W_n \simeq \bigoplus_{k=0}^n \lambda_k^{(n)}, \label{eq:decomp_Wn}\\
     &\varphi_n(U) P_k^{(n)} (O) \varphi_n^\dagger(U)  \in \lambda_k^{(n)},\label{eq:decomp_adjoint}
 \end{align}
where $P_k^{(n)} (O)$ denotes the projection of the operator ${O \in W_n}$ onto the irrep $\lambda_k^{(n)}$ and $\simeq$ indicates a basis change (realized by the Clebsch-Gordon transform) \cite{arienzo_bosonic_2025} .
The dimension of these irreps is given by 
\begin{equation}\label{eq:irrep_dim}
    d_k^{(n)}
    =
    \dim\lambda_k^{(n)}
    =
    \frac{2k+m-1}{m-1}
    \binom{k+m-2}{k}^{\!2},
\end{equation}
where the  photo number $n$ determines the maximal module index $k \leq n$. 
One can easily check that the irrep dimensions sum up to the dimension of  the operator space, i.e. 
\begin{equation}
    \sum_{k=0}^{n}d_k^{(n)}
    =
    \binom{n+m-1}{n}^{\!2}
    =
    \dim W_n.
\end{equation}

Since the second moment commutes with this group action, Schur's lemma  implies that distinct irreducible components do not mix under averaging \cite{goodman_symmetry_2009, serre_linear_1977}. As a result, computing the second moment over the Haar measure simply consists in taking the the $2$-norm after projecting the observable and state separately onto the different irreps appearing in \cref{eq:decomp_Wn}. We formalize this result in \cref{thm:1st_2nd_Moment_Bosonic}.

\begin{theorem}[First and second moment]\label{thm:1st_2nd_Moment_Bosonic}
The mean and variance of the expectation value function $f_U(\rho,O)$  over Haar-random interferometers $U \sim U(m)$ are given respectively by
    \begin{align}
        \e[U \sim U(m)]{f_U(\rho,O)} & = \frac{\Tr[O] }{|\Phi_m^n|} \;,\label{eq:1st_Moment} \\ 
        \var[U \sim U(m)]{{f_U(\rho,O)}} & =  \sum_{k=1}^n \frac{\left\|P_k^{(n)}(\rho)\right\|_2^2 \left\|P_k^{(n)}(O)\right\|_2^2}{d_{k}^{(n)}} \,, \label{eq:sec_m_exp}
    \end{align}
    where $|\Phi_m^n|$ is the dimension of the $n$-photon Fock space $\mathcal{F}_n$, $P_k^{(n)}(X)$ is the projection of the operator $X$ onto the irrep $\lambda_k^{(n)}$  of dimension $d_{k}^{(n)}$ introduced in \cref{eq:irrep_dim}. 
\end{theorem}

\begin{figure*}[t!]
    \centering
    \includegraphics[width=1.0\linewidth]{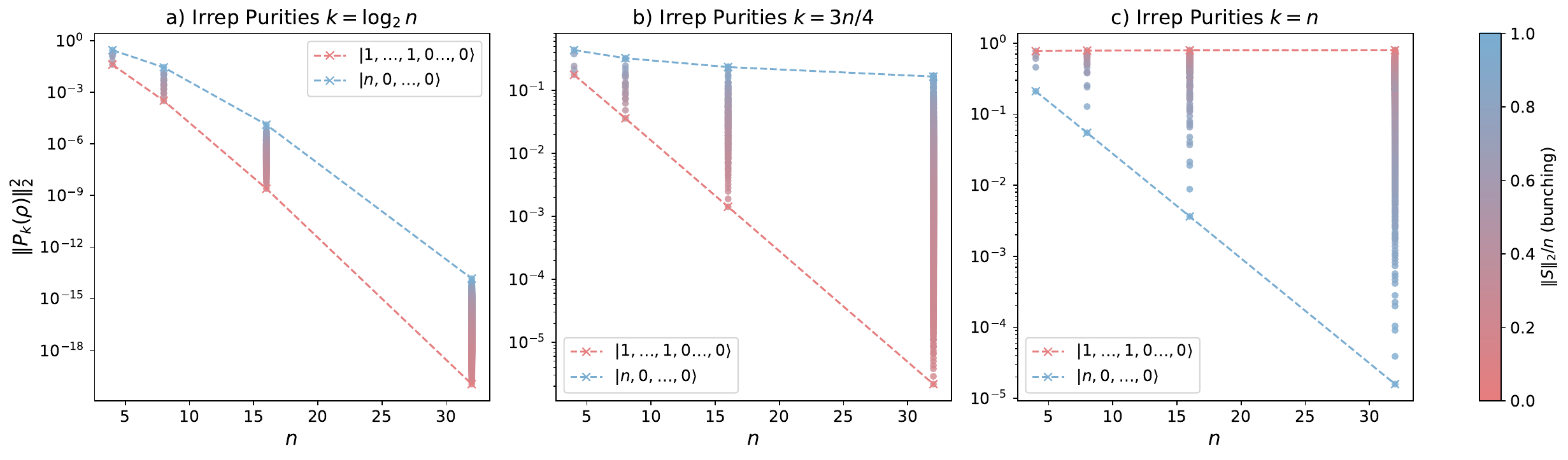}
    \caption{\justifying \textbf{Scaling of the Fock-state purities:} \textbf{a)} Fock-states purities for the irrep $k=\log_2 n$ ; \textbf{b)} Fock-states purities for the irrep $k = 3n/4$; \textbf{c)} Fock-states purities for the irrep $k = n$. The colour-scale corresponds to the $2$-norm of the occupation vector $(s_1, \cdots, s_m)$ for each Fock state.}
    \label{fig:Fock_States_3panels_scatter}
\end{figure*}

The first- and second-moment formulas were formally derived in Ref.~\cite[Proposition~2]{mhiri_boson_2026}. The variance follows by observing that the \(k=0\) contribution to the second moment is precisely the square of the first moment. Exponential concentration of \(f_U(\rho,O)\) occurs when the  contributions from different irreps in \cref{eq:sec_m_exp} are collectively exponentially suppressed. This behaviour is governed by the squared $2$-norms of the state and observable projections onto the irreps of \(W_n\), relative to the corresponding irrep dimensions. 

Henceforth, we refer to the quantities \(\|P_k^{(n)}(X)\|_2^2\) as the \emph{irrep purities} of \(X\). For  states, these purities quantify a generalized notion of entanglement, whereas for observables they quantify a generalized notion of locality, as discussed in Refs.~\cite{larocca_diagnosing_2022,diaz2023showcasingbarrenplateautheory}.
These interpretations are further developed  in \cref{subsec:state_concentration_effects} and \cref{subsec:Obs_Concentration_effects}.

While the subspaces $\lambda_k^{(n)}$ correspond to irreps of $U(m)$ and an orthonormal basis can be obtained through Clebsh-Gordan coefficients \cite{arienzo_bosonic_2025} and used to compute the irrep-purities, this method fails to scale and give closed form analytical expressions enabling to study the scaling of these quantities. In~\cite{mhiri_boson_2026}, we showed that these norms can be expressed as a linear combination of partial trace of the operator in question over particles.

\begin{theorem}[Informal, adapted from~\cite{mhiri_boson_2026}]\label{thm:Projection_Informal}
 Let $O \in W_n$ be an Hermitian operator acting on the $n$-particle Fock sector, and let $P_k^{(n)}(O)$ denote its projection onto the irreducible component $\lambda_k^{(n)}$. Its 2-norm can be expressed as  \begin{equation}\label{eq:norm_proj}
     \left\|P_k^{(n)}(O)\right\|_2^2
    =   \sum_{l=0}^{k}
    b_l \Tr\!\left[\Tr_{n-l}[O]^2\right],
\end{equation}
where $b_l$ are real valued  coefficients $b_l$  and $\Tr_{n-l}[O]$ denotes the partial trace over $n-l$ particles which consists in annihilating one photon in each mode, repeated $n-l$ times, i.e.  
\begin{equation}
    \Tr_{n-l}[O] = \sum_{s_1, \dots, s_{n-l}=1}^m a_{s_{n-l}}  \dots a_{s_1} O a_{s_1}^\dagger \dots a_{s_{n-l}}^\dagger\, .
\end{equation}
\end{theorem}

A formal version of \cref{thm:Projection_Informal} is provided in \cref{app:2ndMoment_Calculus_Iterative_Projection} for completeness, specifying the expression the the coefficients $b_l$. By combining \cref{eq:sec_m_exp,eq:norm_proj}, one obtains a closed-form expression of the variance of a quantum model over Haar random interferometers.

A similar variance expression to the one in \cref{eq:sec_m_exp}  has been obtained for qubit systems in the context of studying barren plateaus under different circuit architectures, pointing to the fact that exponential concentration can be avoided solely from the contribution of low dimensional subspaces such as low dimensional irreps,  making these expectation values prone to efficient classical simulation~\cite{fontana_characterizing_2024, ragone_lie_2024}. In this work, we explore this phenomena in the bosonic setting by further evaluating the closed form expression of the irrep-projections $2-$norms derived in \cref{eq:norm_proj}, for several observables that are relevant in the context of variational quantum algorithms and quantum machine learning.

Concretely, in \cref{subsec:state_concentration_effects} we first examine how the "generalized entanglement" of initial Fock state determines whether exponential concentration occurs. Building on this, in \cref{subsec:Obs_Concentration_effects} we discuss the role of "operator locality" and give guidelines for choosing observables that avoid exponential concentration, together with propositions that help identify them. We then turn to the other side of the problem: classical simulability. In \cref{sec:CS}, we review classical algorithms for  exact and approximate simulation of passive linear optics and identify the regimes in which they are efficient. Finally, in \cref{sec:Concentration_Simulability}, we bring these questions together to map out the different regimes of concentration and simulability.

\subsection{State-dependent concentration effects}\label{subsec:state_concentration_effects}

Throughout this work, we consider Fock-state inputs with a fixed total photon number \(n\). Since Fock states are non-Gaussian, they constitute resource states for passive linear optics, which preserves Gaussianity and therefore cannot generate them from Gaussian inputs~\cite{weedbrook_gaussian_2012,walschaers_non-gaussian_2021}. Their evolution through passive interferometers can consequently give rise to nontrivial computational dynamics, with Fock-state Boson Sampling providing the paradigmatic example~\cite{aaronson_computational_2011}.

For the concentration analysis, the relevant state-dependent quantities are the irrep purities, which we also refer to as \(k\)-purities,
\begin{equation}
    \left\|  P_k^{(n)}\!\left(\ketbra{S}\right)
    \right\|_2^2, \quad \text{with} \quad  \sum_{k=0}^n \left\|  P_k^{(n)}\!\left(\ketbra{S}\right)
    \right\|_2^2 =1.
\end{equation}
As follows from \cref{thm:1st_2nd_Moment_Bosonic}, these quantities determine
the state-dependent contribution to the variance of the expectation values
defined in \cref{eq:linear_form_n}. Characterizing how the purity of the input
state is distributed across the irreducible sectors of \(W_n\) therefore
provides the state-side irrep profile against which the corresponding
observable profile must be compared.

This irrep-based description also suggests a connection with recent group-theoretic formulations of quantum resource theories~\cite{diaz_unified_2025,bermejo_characterizing_2025,barnum_generalizations_2003}.
In these frameworks, the resourcefulness of a state relative to a group of free operations, here \(U(m)\) representing passive linear-optical transformations, can be characterized through its distribution of irrep purities.
While generalized entanglement is conventionally defined relative to the Lie-algebra component, here we adopt a broader notion relative to each irrep, with higher sectors probing increasingly high-order correlations~\cite{barnum_generalizations_2003,barnum_subsystemindependent_2004,ragone_lie_2024}.
Accordingly, a state with an exponentially vanishing purity on a low dimensional irreps indicates strong generalized entanglement relative to that subspace. From this perspective, the extent to which the irrep-purity profile is supported on higher sectors, and the rate at which this weight decays with system size, provide natural signatures of the resource structure of the state.

To investigate the relation between irrep purities and generalized entanglement of Fock states, we first identify the invariances of their irrep-purity profiles. We then study how these profiles vary with photon bunching and scale with system size.

Irrep purities are invariant under passive linear-optical transformations, including permutations of the optical modes. Recall that the \(k\)-th irrep purity of an operator \(X\in W_n\) is the squared $2-$norm of its projection onto \(\lambda_k^{(n)}\).

\begin{lemma}{(Irrep purity invariance )}\label{propop:irrep_purity_permutation}
Let $V \in U(m)$ denote a passive linear-optical transformation, and $\ket{S} \in \Phi_m^n$ a Fock state. By definition of irreducible representations, we have:
\begin{equation}
    \left\| P_k^{(n)} \!\left(\varphi_n(V) \ketbra{S}\varphi_n(V)^\dagger \right) \right\|_2^2 = \left\| P_k^{(n)}\!\left(\ketbra{S}\right) \right\|_2^2 .
\end{equation}

\end{lemma}
This invariance follows from the equivariance of \(P_k^{(n)}\) under the adjoint action of \(U(m)\) and the unitary invariance of the $2-$norm. Mode permutations form a special case of \cref{propop:irrep_purity_permutation}. Consequently, the irrep-purity profile of a Fock state depends only on its occupation partition $S$, rather than on the particular modes carrying the occupations. It is therefore sufficient to retain one representative from each equivalence class under mode permutations.

\begin{figure}[h!]
    \centering
    \includegraphics[width=1.0\linewidth]{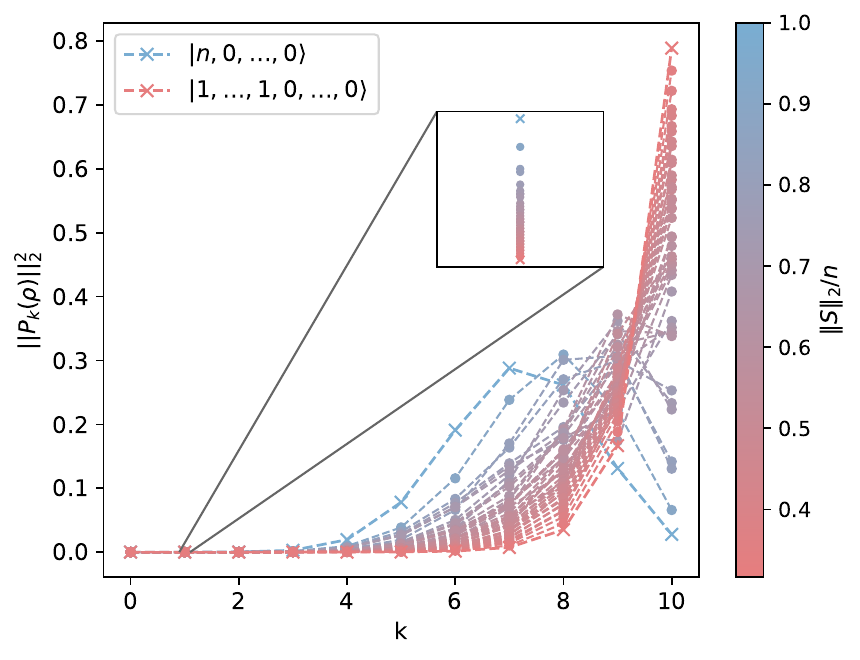}
    \caption{\justifying \textbf{Fock-State purity profile:} Fock state purity distribution for $n=10$ particles in $m=2n$ modes. We only consider the Fock states that are all different under any permutation of the occupation.}
    \label{fig:Fock_States_Purities_QRT}
\end{figure}

Beyond this permutation symmetry, we further observe that Fock states whose occupation vectors have the same \(2\)-norm exhibit identical irrep purities. Applying the general expression of \cref{thm:Projection_Informal} to a Fock basis state shows that its irrep purities are given by the following analytical expression.

\begin{lemma}{(Irrep purities  for Fock states, adapted from \cite{mhiri_boson_2026}).}\label{lemma:puR_fock}
    Consider a Fock state $S$ with a total number of $n$ photons, i.e $S \in \Phi_m^n$. For every $0\leq k \leq n$, its $k$-purity   defined in \cref{eq:norm_proj} can be expressed as 
    \begin{equation}
        \left\|P_k^{(n)}(\ketbra{S}{S})\right\|_2^2
    =    \sum_{l=0}^{k} ((n-l)!)^2 
    b_l \sum_{\substack{|B| = n-l\\0 \preccurlyeq B \preccurlyeq S}} \prod_{i=1}^m \binom{S_i}{B_i}^2\;.
    \end{equation}
    where $b_l$ are real valued coefficients depending on $m,n$ and the irrep index $k$.
\end{lemma} 

Consequently, one can easily show that the purity in the first non-trivial irrep, is determined by the \(2\)-norm of its occupation vector:
\begin{equation}
    \label{eq:supportLOneFockState}
    \left\| P_1^{(n)}\!\left(\ketbra{S}\right) \right\|_2^2 = \binom{n+m}{n-1}^{-1}\left(\|S\|_2^2-\frac{n^2}{m}\right).
\end{equation}
This dependence is illustrated in the inset of \cref{fig:Fock_States_Purities_QRT} and a derivation is provided in \cref{subsec:First_Irrep_Purity}.

For $k\geq 2$, we evaluate numerically the irrep purities using the analytical expansion derived in \cref{lemma:puR_fock} and plot the complete profiles depicted in \cref{fig:Fock_States_Purities_QRT}. We observe that the same invariance with respect to the occupation-vector norm persists across the full irrep decomposition. Moreover, the resulting profiles, i.e., the list of all irrep purities of a state, reveal a clear
dependence on photon bunching: as photons are distributed among more modes,
the dominant weight shifts toward higher-dimensional irreps. This behaviour is consistent
with interpreting high-irrep weight as a signature of generalized
entanglement and with the distinguished computational role of collision-free
inputs in Boson Sampling
\cite{aaronson_computational_2011,aaronson_generalizing_2014}.

Having established these structural properties, we now examine how the
irrep purities decay with system size. Our numerical results are organized
around the two extremal occupation patterns:
\begin{equation}
    \ket{1,\ldots,1,0,\ldots,0},
    \qquad
    \ket{n,0,\ldots,0},
\end{equation}
corresponding, respectively, to minimal and maximal photon bunching. As
illustrated in \cref{fig:Fock_States_3panels_scatter}, these configurations
provide useful limiting cases for characterizing the decay of the
state-side irrep purities and, consequently, their contribution to the
variance.

For both extremal states, the irrep purities decay exponentially with \(n\)
away from their respective dominant sectors. Numerically, the dominant
sector is located at \(k=n\) for the minimally bunched state and near
\(k=3n/4\) for the maximally bunched state. As a representative
intermediate regime, \cref{fig:Fock_States_3panels_scatter}(a) shows that,
for \(k=\log n\), the purities of all Fock states considered decay
exponentially, with a decay rate that increases as photon bunching
decreases. Additional choices of \(k\) are examined in
\cref{app:Initial_State_Purities}.

At \(k=3n/4\), \cref{fig:Fock_States_3panels_scatter}(b) indicates at most
polynomial decay for the maximally bunched state, while progressively less
bunched configurations recover exponential suppression. Conversely, at
\(k=n\), the minimally bunched state exhibits polynomially decaying purity,
whereas increasingly bunched states exhibit exponential decay.

Thus, except near a state-dependent dominant sector, Fock-state irrep
purities are generically exponentially suppressed. To avoid exponential
concentration of observable expectation values, this suppression must be compensated by sufficiently large
observable irrep purities, even when the relevant irrep has only polynomial
dimension. Among the states examined numerically, the maximally bunched
state carries the largest weight on the lowest irreducible sectors. However, one should bear in mind that this state is the "less resourcefull" among Fock states.

\subsection{Obervable-dependent concentration effects}\label{subsec:Obs_Concentration_effects}

Now that we have characterised the decay rates of  Fock basis state irrep purities, a natural question arises:
\begin{center}
    \emph{Given an initial Fock state, which design choices for the observable can avoid exponential concentration?}
\end{center} 
 Motivated by \cref{eq:sec_m_exp}, we characterize observables at different
levels of irrep resolution. At a coarse level, two relevant quantities are the set of nontrivial irreps on which the observable has support,
\begin{equation}
    \operatorname{Supp}(O) = \left\{ k\in\{0, \ldots,n\} \,\middle|\, P_k^{(n)}(O)\neq 0 \right\},
\end{equation}
and its total squared $2$-norm $\norm{O}_2^2$ on these sectors.
 While the support and $2$-norm of an observable may seem like independent quantities in the sense that an observable can have a large support size and a small $2$-norm, and vice versa, they are in fact linked through the more fine-grained irrep projection norms, i.e.

\begin{equation}
    \norm{O}_2^2 = \sum_{k \in \rm Supp(O)}  \norm{P_k^{(n)}(O)}_2^2
\end{equation}

These projection norms define a generalized notion of locality in irrep space: the support identifies the contributing sectors, while the irrep purities quantify their weights. Consequently, generalized locality  characterizes how this support and its weights scale with system size.
 
To answer the above question about observables design choices, we start by enumerating the scenarios under which the expectation value of a particle-number preserving observable may avoid exponential decay. \\

\paragraph*{\textbf{Signal from ``small'' irreps.}} The first possibility is that a non-negligible contribution originates from non exponentially large irreps. The most direct realization occurs when the observable is supported only on a fixed number of polynomial-dimensional irreps. More generally, the observable may have broad support while its high-irrep projections are negligible; exponential concentration is nevertheless avoided whenever its polynomially large irrep projections have sufficient overlap with those of the initial Fock state. As we discussed in the previous section, Fock states irrep profile is more skewed towards large irreps. This implies that in order to have a substantial signal from small irreps the observable 2-norm needs to be large enough to compensate for the initial state contribution. \\

\paragraph*{\textbf{Signal from ``large'' irreps.}} A qualitatively different possibility is that a non-negligible contribution survives on exponentially large irreps. Because the irrep dimension is then exponentially large, the corresponding observable projection must be sufficiently large, after normalization by $\|O\|_\infty$, to compensate this dimensional suppression. This is not sufficient by itself: the initial Fock state must also carry non-negligible purity on the same irreps. A high-irrep contribution can therefore disappear for either of two reasons. First, the observable may have a small global $2-$norm, or insufficient irrep projection norm, so that $\|P_k^{(n)}(O)\|_2^2$ fails to compensate $d_k^{(n)}$. Second, even when the observable projection compensates the irrep dimension, the Fock-state purity $\|P_k^{(n)}(\rho)\|_2^2$ may be exponentially small on that irrep, resulting in a state--observable misalignment. By contrast, if the observable has a sufficiently large normalized projection on an exponentially large irrep and this projection aligns with a non-negligible component of the Fock state, then signal from large irrep  can be non vanishing. In this case, exponential concentration is avoided even after removing all contributions from small irreps. \\

\paragraph*{\textbf{Irrep support constructions.}}

As discussed above, the support of the observable is an important property both for characterizing the variance decay.

Since every particle-number-preserving observable restricted to $W_n$ can be expanded in normally ordered monomials, their irrep decomposition provides a natural building block for understanding the structure of general observables. We therefore consider
\begin{equation}
   [M_{\boldsymbol{p},\boldsymbol{q}}]_d=  a_{p_1}^\dagger \dots a_{p_d}^\dagger a_{q_1} \dots a_{q_d}
\end{equation}
for $\boldsymbol{p},\boldsymbol{q} \in \{1, \dots, m\}^k\}$ and $d \geq 0$,
which we refer to as \textit{degree-$d$ monomials}.
We further introduce the notation $[M_{\boldsymbol{p},\boldsymbol{q}}]_d^{(n)}$ denoting the projection of the monomial on the fixed-particle number sector $\mathcal{F}_n$.

We proceed to give generic construction guidelines of such monomials which will characterize their irrep support. Proofs of all   following propositions is provided in \cref{app:irrep_supp}.

 Our first result establishes that the operator degree alone imposes an upper bound on the number of irreducible sectors that can contribute.

\begin{proposition}\label{propop:degree_cutoff}
    A monomial $[{M}_{\bm p, \bm q}]_d^{(n)}$ of degree $d \leq n$ acting on $n$ particles has support on at most the first $d+1$ irreps of $W_n$, namely $\lambda_0^{(n)} ,\dots, \lambda_d^{(n)}$. 
\end{proposition}
In other words, a degree-$d$ observable can never populate irreducible sectors beyond $\lambda_d^{(n)}$. Thus, the complexity of its decomposition is controlled solely by its operator degree rather than by the dimension of the Hilbert space.

This proposition leaves open the question of whether this bound is typically saturated. The next result shows that it is indeed the case under the following conditions.

\begin{proposition}\label{propop:monomial_single_irrep_support}
     A monomial $[{M}_{\bm p, \bm q}]_d^{(n)}$ of degree $d\leq n$ acting on $n$ particles such that $p_i \neq p_j\, \forall \, 1 \leq i,j \leq d$ has only support on the irrep $\lambda_d^{(n)}$.
\end{proposition}

Monomials with distinct creation and annihilation mode indices therefore occupy the entire range of irreducible sectors allowed by \cref{propop:degree_cutoff}. Consequently, additional structure is required to further reduce the irrep support. One important source of such additional structure is the overlap between the creation and annihilation indices. As the number of common indices increases, destructive cancellations progressively eliminate the highest irreducible sectors.
\begin{proposition}\label{propop:MonomialLastIrrepSupport}
    A monomial $[{M}_{\bm p, \bm q}]_d^{(n)}$ of degree $d\leq n$ acting on $n$ particles such that $|\p \cap \q| = k$, i.e., $\bm p$ and $\bm q$ have exactly $k-1$ common indices, has support only on the $K$ consecutive irreps ending at $\lambda_d^{(n)}$, namely $\lambda_{d-k+1}^{(n)},\dots,\lambda_d^{(n)}$.
\end{proposition}
This result quantifies how index overlap compresses the irrep decomposition. In the extreme case where all indices coincide, the monomial has support on all irreps up to $\lambda_d^{(n)}$ , whereas monomials with no repeated indices are fully supported on the largest irrep $\lambda_d^{(n)}$.

\medskip

\paragraph*{\textbf{Observable normalization and shot noise.}}
To ensure that avoiding concentration reflects a genuine improvement in trainability, we restrict to observables with bounded operator norm, $\lVert O\rVert_\infty=\mathcal{O}(1)$. Indeed, estimating $\langle O\rangle_\rho$ from $N_{\mathrm{shots}}$ projective measurements gives
\begin{equation}
  \operatorname{Var}(\widehat{\langle O\rangle}_\rho)
=\frac{\operatorname{Var}_\rho(O)}{N_{\mathrm{shots}}}
\leq \frac{\lVert O\rVert_\infty^2}{N_{\mathrm{shots}}},
\end{equation}
where $\widehat{\langle O\rangle}_\rho$ denotes the approximation of ${\langle O\rangle}_\rho$ obtained from sampling.
Without this normalization, one could artificially increase the variance across the loss landscape by rescaling $O\mapsto cO$. This, however, increases the shot-noise variance by the same factor $c^2$ and therefore provides no improvement in the number of measurements required to resolve the signal. Bounding $\lVert O\rVert_\infty$ rules out such trivial rescalings \cite{Aghaei_Saem_2026}.

\medskip

In \cref{sec:Concentration_Simulability}, we identify families of observables that do and do not concentrate, based on the principles and propositions developed above. Before doing so, motivated by the close connection between concentration and classical simulability highlighted in Refs.~\cite{ragone_lie_2024,cerezo_does_2025}, we first discuss the regimes in which passive linear-optical systems can be simulated classically. This will allow us to analyse concentration and simulability in parallel in \cref{sec:Concentration_Simulability} and identify the different characteristic regimes. Readers seeking an immediate overview may skip ahead to Table~\cref{tbl:BP_CS_table}.

\section{Classical Simulation Algorithms}\label{sec:CS}

In qubit and fermionic settings, the mechanisms responsible for the absence of barren plateaus have often been found to enable efficient classical simulation. Motivated by this connection, we briefly review the classical-estimation methods relevant to passive linear optics and relate them to the same properties that govern concentration. Although our main objective is to identify if/how it is possible to avoid exponential concentration and evade known simulation methods, classically estimable observables remain relevant to proposals in which a model is trained classically and subsequently deployed on quantum hardware.

The methods considered below exploit two broad types of structure: compression in irrep space and algebraic structure of linear-optical amplitudes. For our purposes, however, the most important distinction concerns the nature of their guarantees. Irrep truncation generally provides an average-case approximation over Haar-random interferometers, whereas exact irrep simulation and the algebraic methods discussed below apply to arbitrary interferometers.
\medskip

\paragraph*{\textbf{Representation-theoretic compression.}}

If the irrep support of the  observable is confined to efficiently representable irreps, the expectation value can be evaluated within this reduced space for every interferometer. This gives the following exact simulation regime.

\begin{simtech}[Exact $\mathfrak{g}$-sim, adapted
from~\cite{goh_liealgebraic_2023}]
\label{simtech:exact_gsim}
Consider a passive linear-optical model \(f_U(\rho,O)\) as defined in
\cref{eq:linear_form_n}. If the shared support of \(\rho\) and \(O\)
is restricted to irreducible components of total polynomial dimension,
and the corresponding projections can be constructed efficiently, then
\(f_U(\rho,O)\) can be computed in polynomial time for arbitrary
\(U\in U(m)\).
\end{simtech}

This criterion applies, in particular, to constant-degree observables,
whose support is restricted to the first few irreps by
\cref{propop:degree_cutoff}. The same construction also gives an
average-case approximation when the support is not exactly restricted
but the contribution of the large irreps is negligible. Define the
dimension-weighted tail
\begin{equation}
    \Delta_K(\rho,O)
    :=
    \sum_{k=K}^{n}
    \frac{
        \|P_k^{(n)}(\rho)\|_2^2
        \|P_k^{(n)}(O)\|_2^2
    }{
        d_k^{(n)}
    }.
    \label{eq:irrep_simulation_tail}
\end{equation}
Retaining only the first \(K\) irreps yields the surrogate
\begin{equation}
    f_C(\rho,O)
    :=
    \sum_{k=0}^{K-1}
    \Tr\!\left[
        P_k^{(n)}(\rho)\,
        \varphi_n(U)^\dagger
        P_k^{(n)}(O)\,
        \varphi_n(U)
    \right],
    \label{eq:surrogate_mode_gsim}
\end{equation}
whose average squared error is exactly
\begin{equation}
    \e[U\sim U(m)]{|f_U-f_C|^2}
    =
    \Delta_K(\rho,O).
    \label{eq:approx_gsim_average}
\end{equation}

\begin{simtech}[Approximate $\mathfrak{g}$-simulation, adapted
from~\cite{goh_liealgebraic_2023}]
\label{simtech:approx_gsim}
If the first \(K\) irreps have polynomial dimension and
\(\Delta_K(\rho,O)=\mathcal O(e^{-n})\), then \(f_C\) approximates
\(f_U\) with exponentially small mean-squared error over Haar-random
interferometers.
\end{simtech}

Thus, the same second-moment analysis used to diagnose concentration also directly constructs an average-case classical surrogate. The exponentially small error required in \cref{simtech:approx_gsim} should, however, be viewed as a strong sufficient condition rather than a necessary requirement for useful simulation. For many applications, an inverse-polynomial additive accuracy may suffice, as considered for example in~\cite{Angrisani_2025}.

Simply replacing the exponential condition on $\Delta_K$ by an unspecified polynomial one is nevertheless too weak to meaningfully distinguish simulable landscapes. Indeed, the trivial surrogate $ f_{\mathrm{triv}}:= \mathbb E_{U\sim U(m)}[f_U]$ has mean-squared error $\operatorname{Var}_{U\sim U(m)}[f_U]$. Consequently, whenever the landscape variance itself decays polynomially, the trivial estimator already achieves a polynomially small mean-squared error while capturing none of the parameter-dependent signal.

A fairer criterion might be to compare the truncation error with the scale of the non-constant landscape. Assuming $\operatorname{Var}_{U}[f_U]>0$, we can define the signal-relative error
\begin{equation}
    \varepsilon_{\mathrm{rel}}^2(K)
    :=
    \frac{\epsilon_{\rm sim}(K)}
    {\operatorname{Var}_{U\sim U(m)}[f_U(\rho,O)]}.
    \label{eq:relative_gsim_error}
\end{equation}
where $\epsilon_{\rm sim}(K) \in \Theta(\Delta_K(\rho,O))$ is the approximation of an order $K$ $\mathfrak{g}$-simulation.
By \cref{eq:approx_gsim_average}, $\varepsilon_{\mathrm{rel}}(K)$ is precisely the root-mean-squared error of the truncated surrogate relative to the standard deviation of the full landscape. Hence, for a specified relative accuracy $\varepsilon$, approximate $\mathfrak g$-simulation is useful when an efficient cutoff $K$ satisfies
\begin{equation}
    \Delta_K(\rho,O)
    \leq
    \varepsilon^2
    \operatorname{Var}_{U\sim U(m)}[f_U(\rho,O)].
\end{equation}
Consequently, evading approximate $\mathfrak g$-simulation on the scale of the landscape signal requires the tail beyond every efficient truncation must remain a non-negligible fraction of the full centered variance.  \\

\paragraph*{\textbf{Worst-case algebraic methods.}}

Passive linear optics also admits algorithms that exploit the structure of its transition amplitudes rather than a truncation in irrep space for classical simulation. 
Unlike the average-case approximation of \cref{simtech:approx_gsim}, these methods apply to every fixed interferometer and therefore provide worst-case (i.e. arbitrary interferometer) guarantees. They may also cover a broader class of observables than the low-irrep families admitting exact evaluation via \cref{simtech:exact_gsim}, at the cost of estimating expectation values up to additive error. 

One such approach estimates expectation values with running time scaling with the norm and locality properties of the observable. The algorithm of Ref.~\cite{lim_classical_2025} applies to product inputs and observables that act nontrivially only on a subsystem \(A\) of modes. Its runtime is polynomial in the \(\|O_A\|_2^2 \times \Tr(\rho_A^2)\), where \(\rho_A\) is the reduced output state on \(A\).

\begin{simtech}[Estimation of expectation values on product states, adapted
from~\cite{lim_classical_2025}]
\label{simtech:Expected_value_LO}
For a product input and an observable \(O_A\otimes\mathbb I_B\), the
expectation value of an arbitrary linear-optical circuit admits an
efficient worst-case additive-error approximation whenever the relevant
observable-norm and reduced-purity factors grow at most polynomially.
\end{simtech}

\begin{figure*}[t!]
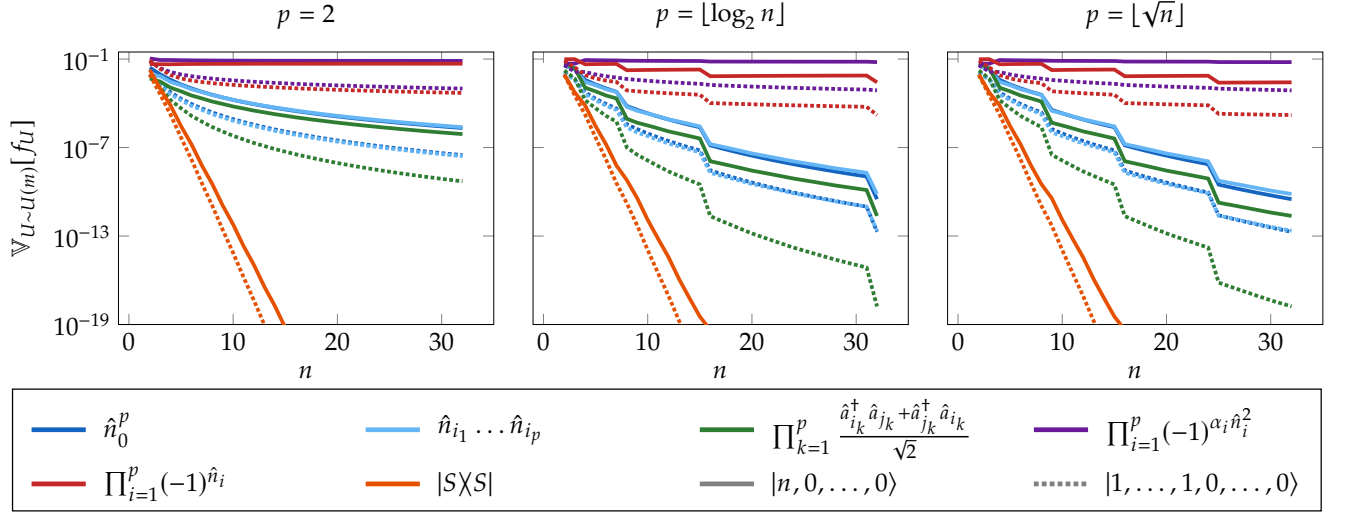

    \centering
    \includegraphics[width=1.0\linewidth]{plots/Concentration_Comparison.tex}
    \includegraphics[width=1.0\linewidth]{plots/Concentration_Comparison_legend.tex}
    \caption{\justifying \textbf{Passive linear optical model concentration:}  Expectation value variance of the different observable of interest presented in \cref{subsec:Analytical_Study_Operators}. For the Fock-state projector case, we consider a projection over the Fock state $\ket{S}$ such that $\occ{S}=p$. A comparison between the different scaling and simulation method is presented in \cref{tbl:BP_CS_table} and discussed in \cref{sec:con}.}
    \label{fig:Concentration_Comparison}
\end{figure*}

Unlike approximate \(\mathfrak g\)-sim, this result does not rely on Haar averaging. It may therefore provide a worst-case approximation for observables whose small norm already makes them easy on average. A second class of algorithms applies when the quantity of interest can be reduced to a linear-optical transition amplitude.

\begin{simtech}[Linear-optical transition-amplitude approximation on product states, adapted from~\cite{lim_classical_2025}]
\label{simtech:Transition_Approx_LO}
The transition amplitude
\(\langle\phi|\varphi_n(U)|\psi\rangle\) between product states of the form $\ket{\psi} = \otimes_{i=1}^m \ket{\psi_i}$ and $\ket*{\phi} = \otimes_{i=1}^m \ket*{\phi_i}$ admits a polynomial time algorithm for computing
an additive-error approximation for arbitrary
\(U\in U(m)\).
\end{simtech}

This result allows to estimate classically the probability of reaching an arbitrary state from any input state. It extends the results from Gurvits' algorithm~\cite{gurvits_complexity_2005, aaronson_computational_2011}, limited to collision-free Fock states, and its generalization for any output Fock state~\cite{aaronson_linearoptical_2011}. This method allow to retrieve the expectation values of passive phase shifters, as 

\begin{align}
        \expval*{e^{i\bm\theta \cdot \hat{\bm {n}}}}
    &=
    \bra{\psi}
        \varphi_n(U)^\dagger
        e^{
        i\sum_{j=1}^{m}\theta_j\hat n_j}
        \varphi_n(U)
    \ket{\psi}.
    \label{eq:linear_statistic_characteristic_function}
\end{align}
The operator in the middle is a layer of phase shifters and hence remains passive linear optical, thus this can be estimated within inverse-polynomial additive-error efficiently. Parity and number phase measurements, which we detail in \cref{subsec:Analytical_Study_Operators} are a particular instance of this construction. This approximation algorithm is central to classical algorithms for \emph{binned-} or \emph{grouped-} probability distributions, via a Fourier transformation of \cref{eq:linear_statistic_characteristic_function} \cite{oh_quantuminspired_2024,seron_efficient_2024}.

We now consider the case where the observable used is defined by power of each creation and annihilation operator. We use the notation ${E}_{\bm p, \bm q}$ to contrast with $M_{\bm p, \bm q}$:
\begin{equation}\label{eq:monomial}
    {E}_{\bm p, \bm q} = \hat{a}_{1}^{p_1\dagger} \cdots \hat{a}_{m}^{p_m\dagger} \hat{a}_1^{q_1} \cdots \hat{a}_m^{q_m}
\end{equation}
is a monomial of degree $k$,
with $\bm p, \bm q \in \Phi_m^k$.

We call $\text{occ}(S)$ the number of non-zero entries of $S$ (corresponding to the number of occupied modes of $\ket{S}$) and 
\begin{equation}
    \loc{{M}_{\bm p, \bm q}} = \max(\occ{\bm p}, \occ{\bm q})
\end{equation}
the \emph{locality} of the monomial operator. 

\begin{simtech}[Direct Expansion, adapted form \cite{thomas_shedding_2025}]\label{simtech:direct_expansion}
    Let $n \in \mathbb{N}^\ast, 1 \leq k \leq n$ and $m \geq n$, with $n$ and $m$ the number of photons and modes respectively, and $k$ the degree of the considered monomial ${M}_{\bm p, \bm q}$. Then, for an occupation $\bm n \in \Phi_m^n$, one can compute classically $f_U(\ketbra{S}{S},{M}_{\bm p, \bm q})$ exactly if either of the following holds: 
    \begin{enumerate}
        \item $k = \mathcal{O}(1)$,
        \item $\occ S = \mathcal{O}(1)$,
        \item $\loc{{M}_{\bm p, \bm q}} = \mathcal{O}(1)$,
        \item $k = \mathcal{O}(\log n)$ and $\occ S = \mathcal{O}(\log n)$.
    \end{enumerate}
\end{simtech}

We discuss this technique in \cref{app:Simulation_Techniques}. This result exploits two methods for computing the matrix permanent, namely
Ryser's algorithm~\cite{ryser_combinatorial_1973}, whose running time is $O(n2^n)$ for an $n\times n$ matrix and Barvinok's~\cite{barvinok_two_1996}, whose running time is $O(n^{\rank (A)})$.

As an illustrative example, we show how the expected value of $\hat n_i^k$ acting on the $m$ mode Hilbert space, can be computed exactly. First, combinatorial analysis of commutation relations gives the normal form
\begin{equation}
    \hat n_i^k = \sum_{j=1}^k S(k,j) \hat a_i^{\dagger j}\hat a_i^j,
\end{equation}
where $S(k,j)$ are Stirling numbers of the second kind~\cite{blasiak_combinatorics_2005}. The operator $a_i^{\dagger j}\hat a_i^j$ is associated with exponent vectors $\bm p = \bm q = (0, \dots, j, 0\dots)$ with a $j$ at the $i$-th position. We show via condition (iii) of \cref{simtech:direct_expansion} that it requires computing the permanent of  rank-$1$ matrices, for which Barvinok's algorithm runs in polynomial time. Thus, $\expval*{\hat n_i^k}$ can be computed in polynomial time for any $k$.

\section{Concentration and Simulability Regimes}\label{sec:Concentration_Simulability}

\begin{table*}[t!]
\centering
\begin{tabular}{|c|c|c|c|c|c|c|c|c|c|c|c|c|}
    \hline
    \multirow{2}{*}{\diagbox[width=1.6cm, height=1.15cm]{\small$p$}{\small$O$}}
    & \multicolumn{2}{c|}{$\hat{n}_i^p$} 
    & \multicolumn{2}{c|}{$\hat{n}_{i_1} \dots \hat{n}_{i_p}$} 
    & \multicolumn{2}{c|}{$\frac{\prod \hat{a}^\dagger_{i_k} \hat{a}_{j_k} + \hat{a}^\dagger_{j_k} \hat{a}_{i_k}}{\sqrt{2}}$} 
    & \multicolumn{2}{c|}{$\displaystyle\prod_{i=1}^p (-1)^{\alpha_i \hat{n}_{i}^2}$} 
    & \multicolumn{2}{c|}{$\displaystyle\prod_{i=1}^p (-1)^{\hat{n}_{i}}$} 
    & \multicolumn{2}{c|}{$\; \ketbra{S}{S} \;$} 
    \\ \hhline{|~|------------|}

    % --- Sub-header row: cols 1-2 empty (covered by multirow), cols 3-14 alternate ---
    & $\mathbb{V}[f_U]$ & Sim
    & $\mathbb{V}[f_U]$ & Sim
    & $\mathbb{V}[f_U]$ & Sim
    & $\mathbb{V}[f_U]$ & Sim
    & $\mathbb{V}[f_U]$ & Sim
    & $\mathbb{V}[f_U]$ & Sim
    \\ \hline

    % --- Row 1: Theta(1) ---
    $\Theta(1)$
    & \multicolumn{2}{@{}c@{}|}{\diagfillcustom{}{EmphasisColorGreen}{EmphasisColorRed}{\st{BP}}{\ref{simtech:exact_gsim},\ref{simtech:direct_expansion}}} 
    & \multicolumn{2}{@{}c@{}|}{\diagfillcustom{}{EmphasisColorGreen}{EmphasisColorRed}{\st{BP}}{\ref{simtech:exact_gsim},\ref{simtech:direct_expansion}}}
    & \multicolumn{2}{@{}c@{}|}{\diagfillcustom{}{EmphasisColorGreen}{EmphasisColorRed}{\st{BP}}{\ref{simtech:exact_gsim},\ref{simtech:direct_expansion}}} 
    & \multicolumn{2}{@{}c@{}|}{\diagfillcustom{}{EmphasisColorGreen}{EmphasisColorRed}{\st{BP}}{\ref{simtech:Expected_value_LO}}} 
    & \multicolumn{2}{@{}c@{}|}{\diagfillcustom{}{EmphasisColorGreen}{EmphasisColorRed}{\st{BP}}{\ref{simtech:Expected_value_LO}}} 
    & \multicolumn{2}{@{}c@{}|}{\multirow{3}{*}{\solidcelllabelled{EmphasisColorRed}{BP}{\ref{simtech:Transition_Approx_LO}}}}
    \\ \hhline{-----------|~~|}

    % --- Row 2: Theta(log n) ---
    $\Theta(\log_2 n)$
    & \multicolumn{2}{@{}c@{}|}{\diagfillcustom{}{orange!30}{EmphasisColorRed}{unclear}{\ref{simtech:exact_gsim},\ref{simtech:approx_gsim},\ref{simtech:direct_expansion}}} 
    & \multicolumn{2}{@{}c@{}|}{\diagfillcustom{}{orange!30}{EmphasisColorRed}{unclear}{\ref{simtech:exact_gsim},\ref{simtech:approx_gsim},\ref{simtech:direct_expansion}}} 
    & \multicolumn{2}{@{}c@{}|}{\diagfillcustom{}{orange!30}{EmphasisColorRed}{unclear}{\ref{simtech:exact_gsim},\ref{simtech:direct_expansion}}}
    & \multicolumn{2}{@{}c@{}|}{\diagfillcustom{}{EmphasisColorGreen}{EmphasisColorRed}{\st{BP}}{\ref{simtech:Expected_value_LO}}} 
    & \multicolumn{2}{@{}c@{}|}{\diagfillcustom{}{EmphasisColorGreen}{EmphasisColorRed}{\st{BP}}{\ref{simtech:Expected_value_LO}}} 
    & \multicolumn{2}{@{}c@{}|}{}
    \\ \hhline{-----------|~~|}

    % --- Row 3: Theta(sqrt(n)) ---
    $\Theta(\sqrt{n})$
    & \multicolumn{2}{@{}c@{}|}{\diagfillcustom{}{EmphasisColorRed}{EmphasisColorRed}{BP}{\ref{simtech:approx_gsim},\ref{simtech:direct_expansion}$^\dagger$}}
    & \multicolumn{2}{@{}c@{}|}{\diagfillcustom{}{EmphasisColorRed}{EmphasisColorRed}{BP}{\ref{simtech:approx_gsim},\ref{simtech:direct_expansion}$^\dagger$}}
    & \multicolumn{2}{@{}c@{}|}{\diagfillcustom{}{EmphasisColorRed}{EmphasisColorRed}{BP}{\ref{simtech:approx_gsim}}}
    & \multicolumn{2}{@{}c@{}|}{\diagfillcustomhighlighted{}{EmphasisColorGreen}{EmphasisColorPurple}{\st{BP}}{}}
    & \multicolumn{2}{@{}c@{}|}{\diagfillcustom{}{EmphasisColorGreen}{EmphasisColorRed}{\st{BP}}{\ref{simtech:Expected_value_LO}}} 
    & \multicolumn{2}{@{}c@{}|}{}
    \\ \hline
\end{tabular}
\begin{tikzpicture}[
    font=\small,
    sq/.style={draw=black, line width=0.4pt, minimum size=1.8ex, inner sep=0pt}
]
  \node[sq, fill=EmphasisColorRed] (r) at (0,0) {};
  \node[anchor=west] (rl) at ([xshift=1mm]r.east) {BP and/or classical simulation};

  \node[sq, fill=EmphasisColorPurple, anchor=west] (p) at ([xshift=7mm]rl.east) {};
  \node[anchor=west] (pl) at ([xshift=1mm]p.east) {Poly small tail contribution with unknown simulation algorithm};

  \node[sq, fill=EmphasisColorGreen, anchor=west] (g) at ([xshift=7mm]pl.east) {};
  \node[anchor=west] (gl) at ([xshift=1mm]g.east) {No BP};
\end{tikzpicture}
\caption{\justifying \textbf{Exponential concentration and classical simulation for different observables in passive linear optics.} The different regimes are described by $p$, as presented in \cref{subsec:Analytical_Study_Operators}. We normalize observables by their infinity norm so as to make them efficient to estimate on a quantum computer. We highlight in orange box a regime where a non-linear number-phase operator avoids barren plateau with a non-negligible signal from large irreps and no clear worst-case classical simulation, when $\alpha_i = 1/r_i$ with $r_i$ denotes the $i$-th prime number. The symbol $\dag$ indicates a dependency on the input state for classical simulation, see \cref{simtech:direct_expansion}.}
\label{tbl:BP_CS_table}
\end{table*}

In this section, we examine some concrete families of particle-number preserving bosonic observables  commonly used in the literature or motivated by the present study, which will be the focus of our classical simulability and concentration analysis.
In \cref{subsec:Analytical_Study_Operators}, for each family of considered observables, we first comment on its irrep support and relate it to the design principles introduced in \cref{subsec:Obs_Concentration_effects}. Then, we identify the appropriate worst-case classical simulation technique  among the ones presented in \cref{sec:CS}, whenever applicable. Finally, we comment on how irrep purities of such observables can be evaluated efficiently. 

In \cref{sec:con},  we evaluate the variance expression from \cref{thm:1st_2nd_Moment_Bosonic} for this list of observables with Fock state initial inputs and classify its decay behaviour into $4$ main categories , which we summarize in the phase diagram in  \cref{fig:Introduction_Figure}. 
Precise classical simulability and concentration results are summarized in \cref{tbl:BP_CS_table}.

\subsection{Bosonic observable families }\label{subsec:Analytical_Study_Operators}

\medskip

\noindent \textbf{Product of photon number operators}~\cite{killoran_continuous-variable_2019, killoran_strawberry_2019, bromley_applications_2020, gan_fock_2022, steinbrecher_quantum_2019}: 
\begin{equation}
    \prod_{i=1}^p \hat{n}_{i} = \prod_{i=1}^p \hat{a}_i^\dagger \hat{a}_i \, .
\end{equation}
A single photon number operator $\hat{n}_i$ on mode $i$ describes a photon-count, and can be realised with photon-number-resolving (PNR) measurement, a detection apparatus that determines how many photons arrive in a given mode. A product on different modes defines photon-number correlations between these modes~\cite{laiho_measuring_2022}. This observable, of degree $p \leq n$, only has support on the first $p$ irreps according to \cref{propop:degree_cutoff}, and can thus be simulated by exploiting its limited support~\cite{goh_lie-algebraic_2025} (see \cref{simtech:exact_gsim}) in polynomial time for constant $p$. In addition, existing techniques of simulation can be used when this observable can be written as a constant number of monomials (see \cref{simtech:direct_expansion}). Finally, we provide a more detailed analysis  of this observable, including how to derive its projection over the different irreps, and how to adapt the corresponding simulation techniques in \cref{app:Product_Number_Operators}. Broadly, this type of observables are diagonal in the Fock basis, hence their irrep purities will make use of the irrep purities derivations of Fock states. \\

\noindent \textbf{Single Support Observable}:
\begin{equation}\label{eq:obs_single_irrep}
    \prod_{k=1}^p \frac{\hat{a}_{i_k}^\dagger \hat{a}_{j_k} + \hat{a}_{j_k}^\dagger \hat{a}_{i_k}}{\sqrt{2}} \, ,
\end{equation}
with all the indexes $\{i_k, j_k\}_{k=1}^p$ different. This observable only has support on the last irrep $p \leq n$ (see \cref{propop:monomial_single_irrep_support}). It can be measured experimentally by using beam splitters between each pair of modes in $\{(i_k,j_k)\}_{k=1}^p$ and PNR measurement. Moreover, it can be simulated efficiently classically for $p = O(1)$ with respect to $n$ using its decomposition into a sum of monomial (see \cref{simtech:direct_expansion}) or its limited support on small polynomially large irrep (see \cref{simtech:exact_gsim} and \cref{propop:degree_cutoff}). A detailed study is presented in \cref{app:Single_Support_Obs}, where the projection over the single support irrep can simply be replaced by the $2$-norm of the observable. \\

\noindent \textbf{Number-phase Operators}:

We call $g_p(\{\hat{n}_i \}_{i \in I})$, a polynomial function of number operators of degree $p \leq n$, with $n$ the number of particles. We call \emph{Number-phase Operators}, operators of the form:
\begin{equation}
    O = (-1)^{g_p(\{\hat{n}_i \}_{i \in I})} = \sum_{R \in \Phi_m^n} (-1)^{g_p(\{R_i \}_{i \in I})} \ket{R}\bra{R} \,.
\end{equation}
The case where $p=1$ is commonly used in the context of generative modelling~\cite{kurkin_universality_2026, kolarovszki_generative_2026, gottlieb_efficient_2026}, inspired by the used of parity operator in qubit setting~\cite{rudolph_trainability_2024,recio-armengol_train_2026,herrero-gonzalez_born_2025}. A complete study of this operator can be found in \cref{app:Parity_Operators}. In particular, we remind in \cref{lem:Decomposition_Phase_Shifter} that the expectation value of such operators can be approximated within additive-error by decomposing it into a sum of  phase shifters for which a classical approximation algorithm exist~\cite{upreti_exponentially-improved_2026, jin-wei_superposition_1996}. It can thus be simulated when the number of terms is constant with respect to $n$. While this operator is not Hermitian in general, we  consider the associate Hermitian observable defined as $(O+ O^\dagger)/2$. This observable in also diagonal in the Fock basis by construction and hence irrep purities calculations follows from those of Fock states, as detailed in \cref{app:Parity_Operators}.\\

\noindent \textbf{Projector over a Fock state}~\cite{wang_quantum_2021, di_bartolo_time-series_2026, chabaud_quantum_2021}:
\begin{equation}
    \ketbra{S}{S}, \quad \text{with} \quad \text{occ}(S) = p \; ,
\end{equation}
with $S \in \Phi_m^n$ and $ \text{occ}(S)$ the number of occupation of the Fock state. It is well known that the projection over a Fock state, i.e., output probabilities of Boson Sampling, can be approximated classically efficienly ~\cite{gurvits_complexity_2005, aaronson_computational_2011,lim_classical_2025} (see \cref{simtech:Transition_Approx_LO}). As we showed in  \cref{subsec:state_concentration_effects} and according to \cref{propop:MonomialLastIrrepSupport}, this operator is supported by all irreps and its irrep purity decay is extensively studied in \cref{subsec:state_concentration_effects}.

\subsection{Concentration regimes and link to classical simulability }\label{sec:con}

In this section, we analyze and interpret the variance decay of the observable families introduced above, as summarized in \cref{tbl:BP_CS_table}. For each observable, by combining their irrep purities expansions derived in \cref{app:Product_Number_Operators}-\cref{app:Parity_Operators} with irrep purities of input Fock states established in \cref{lemma:puR_fock}, we numerically evaluate the  variance formula given in \cref{eq:sec_m_exp} as a function of system size to determine its scaling. Based on the resulting behaviour, we classify the observables into four categories according to which irreducible sectors yield non-exponentially suppressed contributions to the variance.

The absence of exponential variance decay also yields an average-case classical surrogate for the corresponding expectation-value model, as detailed in \cref{simtech:approx_gsim}. Whenever applicable, we compare this consequence with the worst-case simulation guarantees discussed in \cref{subsec:Analytical_Study_Operators} and further discuss when non of the aforementioned classical simulation techniques apply.

Let \(\widetilde O=O/\|O\|_\infty\) denote the observable of interest rescaled by its $\infty$-norm, and denote the contribution of the \(k\)-th irrep by
\begin{equation}
    \Gamma_k(\rho,O)
    =
    \frac{
        \left\|P_k^{(n)}(\rho)\right\|_2^2
        \left\|P_k^{(n)}(\widetilde O)\right\|_2^2
    }{
        d_k^{(n)}
    }.
    \label{eq:irrep_contribution_regimes}
\end{equation}
Throughout our analysis, we consider the two extremal Fock-state inputs introduced previously: the \emph{minimally bunched} Fock state, with one photon per occupied mode, and the \emph{maximally bunched} Fock state, with all photons occupying a single mode. Detailed variance scaling of the list of studied observables with respect to both minimally and maximally bunched input Fock states  are reported in \cref{app:Product_Number_Operators}-\cref{app:Parity_Operators}. Since the two inputs lead to the same qualitative scaling regimes for the observables considered, we do not distinguish between them in the summary table, although simulation techniques could exploit the input Fock state degree of bunching.

To identify the origin of a non-negligible signal, we introduce an irrep cutoff \(K\),  and define
\begin{equation}
    \mathcal V_{\leq K}
    =
    \sum_{k=1}^{K}\Gamma_k(\rho,O),
    \qquad
    \mathcal V_{>K}
    =
    \sum_{k=K+1}^{n}\Gamma_k(\rho,O).
    \label{eq:low_high_irrep_signal}
\end{equation}
For $K = \Theta(1)$, the cut-off separates (classically) efficiently accessible irreps, i.e., that of polynomial dimension, from the larger, exponential-dimensional ones. Henceforth, we refer to irreps of polynomial dimension by \emph{low} or \emph{small} irreps, and to irreps of at least super-polynomial dimension by \emph{large} irreps. 
In the Appendices corresponding to each observable family, we provide the corresponding numerical scaling of the full variance and of the truncated contributions for different choices of the irrep cut-off \(K\). To diagnose the observed decay regime, we compare polynomial and exponential fits and report their  \(R^2-\)score.  The resulting concentration and simulability regimes are summarized in \cref{tbl:BP_CS_table} and interpreted in the following sections.

Specifically, the concentration behavior can then be organized into the following four regimes.

\subsubsection{No concentration and exactly simulable}\label{subsec:NoConcentration_Exact_Sim}

The simplest way to avoid exponential concentration is to restrict the
observable exactly to polynomial-dimensional irreps. If
\begin{equation}
    \operatorname{Supp}(O)
    \subseteq
    \{0,\ldots,K\},
\end{equation}
then \(\mathcal V_{>K}=0\), and the second moment is entirely
determined by the small-irrep contribution. Provided that
\(\mathcal V_{\leq K}\) is not exponentially small, the model avoids
exponential concentration.

As depicted in \cref{tbl:BP_CS_table} (first row, first three columns), this regime includes constant-degree photon number product observables, whose irrep support
is restricted by \cref{propop:degree_cutoff}, as well as
observables explicitly constructed within a fixed low irrep, introduced in \cref{eq:obs_single_irrep}.

Here we note that while the dimension of the irreps for $p= \Theta(1)$ is at most polynomial in system size, the irrep purities of Fock states are nevertheless exponentially decaying. However, the non decay of the term $\Gamma_p(\rho,O)$ is due to the large 2-norm of these observables (up to normalization). The detailed variance scaling analysis is summarized in \cref{tbl:Decay_Product_Number_Operator} and \cref{tbl:Decay_Single_Irrep_Observable} for the photon number product observables and single-irrep observable respectively.

The same restriction that protects the model from concentration also enables exact \(\mathfrak g\)-simulation through \cref{simtech:exact_gsim}. This therefore realizes the familiar scenario in which trainability is obtained by confining the model to a small classically accessible subspace. Other simulation methods such as \cref{simtech:direct_expansion} can also be applied in this context.

We note that for the same class of observables (first three columns) with $p = \rm log(n)$, the variance decay scaling are not conclusive since the $R^2$ score for the polynomial and exponential fits were close.

\subsubsection{No concentration due to signal from "small" irreps}\label{subsec:NoConcentration_SmallIrreps}

An observable may have support on large irreps while receiving a non-negligible contribution only from the small ones. In this case,
\begin{equation}
    \mathcal V_{\leq K}
    \notin \mathcal O(\exp(-n)),
    \qquad
    \mathcal V_{>K}
    \in \mathcal O(\exp(-n)).
\end{equation}
The model avoids exponential concentration because of its low-irrep
signal, even though the observable is not exactly supported on that
sector.

This regime cannot be diagnosed from the support or total
Hilbert--Schmidt norm alone. An observable may have broad support and
a large \(2\)-norm while its irrep-purities on 
large irreps remain negligible. The relevant quantity is instead the
fine-grained tail \(\mathcal V_{>K}\). When this tail decays
exponentially, the model admits the average-case irrep surrogate of
\cref{simtech:approx_gsim}. Thus, broad irrep support does not by
itself prevent classical approximation: the model must carry a
substantial signal on the large irreps, rather than merely having
nonzero projections onto them.

\subsubsection{Exponential concentration}\label{subsec:Exp_Concentration}

Exponential concentration occurs when neither the low- nor the
high-irrep sectors provide a non-negligible contribution,
\begin{equation}
    \mathcal V_{\leq K}
    +
    \mathcal V_{>K}
    \in
    \mathcal O(\exp(-n)).
\end{equation}
This suppression can arise through two complementary mechanisms.

First, the normalized observable projections may be too small to
compensate the irrep dimensions. This is particularly clear for Fock
projectors, whose Hilbert--Schmidt norm is equal to one. Their
irrep-purities therefore cannot compensate the dimensions of the large
irreps and  are exponentially suppressed for low dimensional irreps as discussed in \cref{subsec:state_concentration_effects} ,  leading to an overall exponential concentration.

Second, even when
\(\|P_k^{(n)}(\widetilde O)\|_2^2\) compensates
\(d_k^{(n)}\), the state purity
\(\|P_k^{(n)}(\rho)\|_2^2\) may decay exponentially faster than the ratio $\|P_k^{(n)}(\widetilde O)\|_2^2 / d_k^{(n)}$. This is reminiscent of  a misalignment between the state and observable
profiles. Fock-state inputs place their dominant purity on relatively
large irreps, with the collision-free state shifted furthest toward
the end of the decomposition. Consequently, observables whose weight
is concentrated on substantially smaller irreps or simply not properly aligned with maximal state weight sector can exhibit
exponential concentration even when their total \(2\)-norm is large.

This mechanism appears, for instance, for products of photon-number
operators when their degree grows with the photon number. Their
irrep-purities need not be exponentially suppressed by
themselves, but their overlap with the Fock-state profile may still
decay exponentially. In this regime, the Haar mean already provides
an average-case additive surrogate. Depending on their degree,
locality, and the bunching of the input state, some of these
observables may additionally admit stronger worst-case
permanent-based approximations.

\subsubsection{No concentration and signal from "large" irreps}\label{subsec:NoConcentration_Signal_Large_Irreps}

The most interesting regime occurs when a non-negligible contribution
survives on irreps beyond every efficient low-irrep truncation:
\begin{equation}
    \mathcal V_{>K}
    \notin
    \mathcal O(\exp(-n)).
    \label{eq:substantial_high_irrep_signal}
\end{equation}

While the contribution term $\Gamma_k(\rho ,O)$ cannot avoid exponential concentration for $k>n/2$ by simple dimensionality arguments developed in \cref{app:pos}, signal from irrep indexes growing with $n$ can nevertheless escape exponential decay.

Specifically, since the dimensions of these irreps can grow exponentially, this
requires the normalized observable projections to compensate the
dimension factor while remaining aligned with non-negligible
components of the Fock-state profile. A large total \(2\)-norm is
therefore useful only if a sufficient fraction of it is distributed
over the appropriate high irreps.

Number-phase observables provide examples of this behavior. They
generally have broad irrep support and a large Hilbert--Schmidt norm,
allowing the variance to remain non-exponentially vanishing even after the
low-irrep contributions are removed.

Their absence of concentration therefore cannot be attributed simply to confinement within a small representation-theoretic subspace. However, for the minimally bunched input and  number-phase observable, our numerical fits indicate that the large-irrep tail is polynomially smaller than the full centered variance. Thus, for the cutoffs considered, approximate \(\mathfrak g\)-simulation would potentially provide a reasonably inverse-polynomial signal-relative approximation (as discussed in \cref{simtech:approx_gsim}). Whether this level of relative precision is sufficient for useful classical simulability remains open in practise.

Nevertheless, these average-case classical surrogate guarantees do not directly extend to worst-case approximation guarantees.
Specifically, 
the (non-)decay of high-irrep signal does not by itself characterize
classical hardness. While, the expectation value of standard parity observables (5th column in \cref{tbl:BP_CS_table}), consisting of products of
phase shifters, reduce to transition
amplitudes of passive linear-optical circuits and can be estimated using \cref{simtech:Transition_Approx_LO}, non-linear number-phase observables may instead be expanded into sums of linear phase shifts through a discrete Fourier decomposition, as discussed
in \cref{lem:Decomposition_Phase_Shifter}. The resulting classical
algorithm is efficient only when the number and total weight of the
Fourier components remain polynomially controlled.

In \cref{app:Parity_Operators}, we further argue that direct approximation methods such as importance sampling and truncation based on Fourier coefficients amplitudes would fail for such an observable whenever the initial state considered in the min bunching Fock states in the regime $m=cn$ where $c>1$ is a constant. Similarly all worst-case simulation techniques presented in \cref{sec:CS} do not apply. Hence, worst-case approximation for such an observable remains out of reach of currently available simulation methods.

\section{Discussion}

In this work, we presented an adaptation of a representation-theoretic framework to study the concentration of passive linear optical quantum models. We identified the main causes of exponential concentration, conditions for its absence, and examined how they relate to classical simulation techniques. While our study aims to be as exhaustive as possible, particularly regarding the review of simulation techniques, we believe that it could serve as the starting point for a broader discussion on the trainability and simulation of quantum models.

Recent works have shown that regimes used to prove the absence of barren plateaus often also admit classical simulation~\cite{cerezo_does_2025}. In the qubit setting, the variance can be related to the irrep-purities of the initial state and observable~\cite{fontana_characterizing_2024, ragone_lie_2024} relative to the adjoint representation of the Lie group corresponding to the circuit in question, and this connection has been attributed to a curse of dimensionality when loss functions compare objects in exponentially large spaces. Our results identify a possible route to circumventing this curse in bosonic systems (and maybe beyond). In the second-moment expression, the dimensional suppression of a large irrep can, in principle, be compensated by a bounded suitably aligned observable whose projection has an exponentially large squared $2$-norm. Non-concentration could therefore arise from signal retained in an exponentially large operator space, rather than from confinement to a classically accessible subspace. The iterative projection procedure adapted from~\cite{mhiri_boson_2026} makes this mechanism directly testable. This raises the question of whether analogous constructions exist in qubit settings, for example using Lie-algebraic ans\"atze~\cite{larocca_diagnosing_2022, kerenidis_quantum_2022, monbroussou_trainability_2025, wiersema_classification_2024, kokcu_fixed_2022}.

However, we have not obtained a definitive separation between non-concentration and classical simulability. For a minimally bunched input and nonlinear number-phase observable, our numerics indicate both an absence of exponential concentration and (potentially - the fits are inconclusive) inverse-polynomial contribution from high-dimensional irreps after removing lower sectors. Yet this high-irrep contribution is a large polynomial smaller than the full centered landscape variance, leaving open whether approximate $\mathfrak{g}$-simulation could reproduce the landscape to useful inverse-polynomial relative precision. The example therefore illustrates the proposed mechanism without providing a clear case that is both non-concentrated and beyond known efficient classical simulation methods.

More broadly, our framework provides a systematic pipeline: construct bounded physical observables with large high-irrep projections, align them with the input-state purity profile, quantify the signal remaining beyond every efficient irrep truncation relative to the full variance, and test the model against simulation methods exploiting additional algebraic structure. A more thorough search may reveal examples for which even signal-relative approximate $\mathfrak{g}$-simulation fails.

These observations raise the question of extending worst-case classical approximation methods to nonlinear number-phase observables and, more generally, to observables following the recipe described above. For every other observable studied here that admits a clearly efficient irrep truncation, we identified a corresponding worst-case simulation algorithm in the bosonic literature. Whether such an algorithm exists for the non-linear number-phase observable considered here remains unclear, as none of the worst-case techniques reviewed in this work applies directly to it.

Finally, in this work we considered only passive linear optical transformations. While this setting permits the use of the various representation-theoretic techniques presented here, extending the analysis to other frameworks such as adaptive \cite{rodari_beyond_2026, hoch_quantum_2025, monbroussou_toward_2025, monbroussou_photonic_2025} or non-linear \cite{spagnolo_non-linear_2023, scheel_measurement-induced_2003, costanzo_measurement-induced_2017} optics  would be a valuable direction for future work.

\medskip
\emph{Note added}. Upon completion of this work, we became aware of independent work by Makarovskiy et al. \cite{makarovskiy_trainability_2026} that studies the trainability of linear optical models, through the same second moment expression.\\

\section*{Acknowledgments}

The authors acknowledge Ulysse Chabaud, Verena Yacoub, Pierre-Emmanuel Emeriau, and Marco Cerezo for fruitful discussions. LM and EK acknowledge support from the EPSRC Quantum Advantage Pathfinder research (EP/X026167/1), and the Hub for Quantum Computing via Integrated and Interconnected Implementations (QCI3, EP/Z53318X/1) programs within the UK’s National Quantum Computing Center. HM acknowledges support from the grant ANR-22-PETQ-0007. HT acknowledges support from the European Commission as part of the EIC accelerator program under the grant agreement 190188855 for SEPOQC project, the Horizon-CL4 program under the grant agreement 101135288 for EPIQUE project, and the the CIFRE grant N. 2023/1746. ZH acknowledges support from the Sandoz Family Foundation-Monique de Meuron program for Academic Promotion.\\

\paragraph*{Code Availability.} The code, and the data used in this article are available on the following \href{https://github.com/quantumsoftwarelab/Classical-simulation-and-model-concentration-in-passive-linear-optics}{repository}.

\bibliographystyle{apsrev4-2}
%\bibliography{references}
\bibliography{references}

\clearpage

\appendix

\onecolumngrid
\newpage
% \section*{Appendix}

% \begingroup
% \makeatletter
% \par 
% \newcommand{\contentsname}{Appendix Contents}
% \@starttoc{atoc}
% \makeatother
% \endgroup

% ##### start TOC stuff #####
\DoToC
% ##### end TOC stuff #####
\section{Second Moment Calculus via the Iterative Projection Procedure}\label{app:2ndMoment_Calculus_Iterative_Projection}

In order to derive the observable and input state purities over the different irreducible representations, we use the iterative projection procedure developed in \cite{mhiri_boson_2026}. Methods such as the Clebsch--Gordan construction provides an explicit, basis-level realization of the irreducible decomposition, that is hard to realize on larger system size. The recursive projection method operates directly at the level of invariant maps, isolating irreducible contributions through the ladder structure without explicit computations of coefficients. This basis-independent formulation is particularly well suited for evaluating quantities such as Hilbert--Schmidt norms, i.e., irrep purities necessary to evaluate the second moment given in \cref{thm:1st_2nd_Moment_Bosonic}.

We introduce the \emph{lowering} and \emph{raising} maps, denoted respectively by $\Cm$ and $\Dm$,
\begin{align}
     \Cm (\cdot) & = \sum_{s=1}^m a_s (\cdot) a_s^\dagger \;,\label{eq:C_map}\\
     \Dm(\cdot)  & = \sum_{s=1}^m a_s^\dagger (\cdot) a_s.\label{eq:D_map}
\end{align}

These maps admit a simple interpretation in terms of photon number: when acting on operators supported on a fixed sector $W_n$, $\Cm$ decreases the photon number by one, while $\Dm$ increases it by one. They therefore provide a natural mechanism for relating operator spaces associated with different photon numbers. 

\begin{theorem}[Irrep norms, from Theorem 5 in \cite{mhiri_boson_2026}]\label{th:recursive_irrep_norm_app}
Let $O$ be an observable acting on the $n$-photon Fock sector, i.e., $O \in W_n$. 
For every $0 \le k \le n$, define
\begin{equation}\label{eq:g_k_th}
    g_k(O) = \Tr\!\left[\big(\Cm_{k+1 \rightarrow n}(O)\big)^2\right] := \Tr[\Tr_{n-k}[O]^2]\;,
\end{equation}
where $\Cm_{k+1 \rightarrow n}$ denotes the composition of lowering maps and corresponds equivalently to partial tracing $n-k$ particles. 
Let $O_k^{(n)}$ denote the component of $O$ supported on the irreducible subspace $\lambda_k^{(n)}$ (see \cref{eq:decomp_Wn}). 
Then the expression of Hilbert--Schmidt norm of $O_k^{(n)}$ is given by
\begin{equation}
   \Tr\left[(O_k^{(n)})^2\right] = \frac{1}{\alpha_{k,k+1,n}} \sum_{l=0}^k \frac{(-1)^{k-l}}{(k-l)! (k+l+ m-1)_{k-l}} g_l\;,
\end{equation}
where $(x)_p$ denotes the Pochhammer symbol defined as $(x)_p := x (x+1) \dots (x+p-1)$ and $\alpha_{k,k+1,n} = \frac{(n-k)! (m+n+k-1)!}{(m+2k-1)!}$.
\end{theorem}

\section{Generic design guidelines for bosonic observables}\label{app:irrep_supp}

In this section, we present useful characterizations of particle number preserving bosonic operators. These results are used  to construct observable that only have support on certain irreps. We consider the space of rank $d$-symmetric tensors $W_d:= \rm Sym^d(\mathbb{C}^m) \otimes \rm Sym^d(\mathbb{C}^m)^*$ which correspond to the space bosonic opera spanned by normally ordered monomials in creation annihilation operators on $m$ modes. 
Indeed any particle number preserving operator of degree $d$ can be written as 
\begin{equation}
    O = \sum_{i_1,\dots,i_d =1}^m \sum_{j_1,\dots,j_d=1}^m O_{i_1,\dots,i_d}^{j_1,\dots,j_d} a^\dagger_{i_1} \dots a^\dagger_{i_d} a_{j_1} \dots a_{j_d} = \sum_{I,J} O_{I}^J [M_{I,J}]_d^{(n)} \;.
\end{equation}
Here, one can easily see that any particle number preserving observable $O$ of degree $d$ can be indeed represented as a symmetric tensor of rank $d$ since $[a_i^\dagger , a_j^\dagger]=[a_i , a_j]=0$.

We introduce the following notation of a normally ordered monomial of degree $d$.
\begin{equation}
    [M_{I,J}]_d := a_{i_1}^\dagger \dots a_{i_d}^\dagger a_{j_1} \dots a_{j_d} 
\end{equation}

Another notation of the same monomial which will be convenient in some cases is the one where we track the number of creation and annihilation operators on each mode. Precisely,
\begin{equation}\label{eq:compact_monomial}
    [M_{I,J}]_d := a_{i_1}^\dagger \dots a_{i_d}^\dagger a_{j_1} \dots a_{j_d} = (a_1^\dagger)^{\alpha_1(I)}  \dots (a_m^\dagger)^{\alpha_m(I)} a_1^{\alpha_1(J)} \dots a_m^{\alpha_m(J)} = (a^\dagger)^{\alpha(I)} (a)^{\alpha(J)}
\end{equation}
Here $\alpha(I)$ and $\alpha(J)$ are multi-indices of size $m$ whose components take values between $0$ and $d$ such that $\sum_{s=1}^m \alpha_s(I) = \sum_{s=1}^m \alpha_s(J) = d$.

When acting on $n$ particles, a monomial of degree $d$ has non trivial action only for $d\leq n$. In this setting, we denote the monomial projected onto the $n$-photon sector by
\begin{align}
    [M_{I,J}]_d^{(n)} &:=P_n [M_{I,J}]_d P_n
\end{align}
where $I$ and $J$ are multi-indices of size $d\leq n$ with each component taking integer values from $1$ to $m$.

%%%%

In the remainder of this section, we derive the proofs of the irrep support \cref{propop:degree_cutoff}, \cref{propop:monomial_single_irrep_support} and \cref{propop:MonomialLastIrrepSupport}.

To do so we recall some useful results for this derivation which were established in \cite{mhiri_boson_2026}.

\begin{lemma}\label{lemma:com_property}
Under the bosonic commutations relations, we have
    \begin{equation}
        [a,[a^\dagger, (a^\dagger)^p (a)^q]] = -pq  (a^\dagger)^{p-1} (a)^{q-1} \;, \forall p,q \geq 1
    \end{equation}
\end{lemma}

\begin{proposition}{(Top irrep of the $k$-photon operator subspace)} \label{prop:kernel_top_irrep}
 The kernel of the lowering map $\Cm_k$ coincides with the top irrep in $W_k$, i.e.
    \begin{equation}\label{eq:top_irrep_ker}
   {\rm Ker}(\Cm_k) = \lambda_k^{(k)}\;.
    \end{equation}    
\end{proposition}

\begin{proposition}{(Irreducible component characterization).}\label{prop:lower_irreps}
For every fixed $r \geq 0$, the irreducible components $\lambda_r^{(i)}$ and $\lambda_r^{(j)}$ are isomorphic for all $i,j \geq r$. Moreover, for any $k \geq 0$ and every $r < k$, the irreducible component $\lambda_r^{(k)}$ is obtained from $\lambda_r^{(r)}$ by successive application of the raising maps, namely
\begin{equation}\label{eq:irrep_expression}
\lambda_r^{(k)} 
= \Dm_k \circ \Dm_{k-1} \circ \dots \circ \Dm_{r+1}\bigl(\lambda_r^{(r)}\bigr) \; .
\end{equation}
\end{proposition}

\begin{proposition}{(Eigenvalues of raising and lowering maps on irreducible components).}\label{prop:CD_DC}    
         The maps $\Cm_{k+1} \circ \Dm_{k+1} : W_{k} \rightarrow W_k$ and  $\Dm_k \circ \Cm_k : W_{k} \rightarrow W_k$ act as scalar multiples of the identity on each  irreducible component  $\lambda_r^{(k)}$ for $0 \leq r\leq k$, i.e.
        \begin{align}
            \Cm_{k+1} \circ \Dm_{k+1}|_{\lambda_r^{(k)}} &= \beta_{r,k} \Id_{\lambda_r^{(k)}} \\
            \Dm_k \circ \Cm_k|_{\lambda_r^{(k)}} &= \beta_{r,k-1} \Id_{\lambda_r^{(k)}}
        \end{align}
        where $\beta_{r,k} =  (k-r+1) (m+r+k)$.    \end{proposition}

\subsection{Irrep support construction proofs}

\begin{proposition}\label{propop:degree_cutoff_app}
    A monomial $[{M}_{\bm p, \bm q}]_d^{(n)}$ of degree $d \leq n$ acting on $n$ particles has support on at most the first $d+1$ irreps of $W_n$, namely $\lambda_0^{(n)} ,\dots, \lambda_d^{(n)}$. 
\end{proposition}
\begin{proof}
     It suffices to show that there exists $[X]^{(d)}_d \in W_d$ such that $[M_{I,J}]_d^{(n)}= R_n \dots R_{d+1}([X]^{(d)}_d)$, meaning that the monomial has zero support on the top $n-d$ irreps. Indeed, in what follows we show that by construction $[X]^{(d)}_d= \frac{1}{(n-d)!}[M_{I,J}]_d^{(d)}$.

    First we start by showing that the action of the raising map $R_k$ on a normally ordered monomials of degree $d \leq k$   reduces to
\begin{align}\label{eq:raising_trivial}
    R_k([M_{I,J}]_d^{(k-1)}) = (k-d) [M_{I,J}]_d^{(k)}
\end{align}
where we recall that $[M_{I,J}]_d^{(k-1)} := P_{k-1} [M_{I,J}]_d P_{k-1}$ and $[M_{I,J}]_d^{(k)} := P_k [M_{I,J}]_d P_k$.

Indeed, we can further expand the LHS in \cref{eq:raising_trivial} as follows
\begin{align}
    R_k([M_{I,J}]_d^{(k-1)}) &= P_k R([M_{I,J}]_d) P_k\\
    &= P_k \sum_{s=1}^m a_J^\dagger a_s^\dagger  a_s  a_I P_k \\
    &=  P_k  a_J^\dagger \sum_{s=1}^m(a_s^\dagger  a_s)  a_I P_k\\
    &= P_k  a_J^\dagger \hat{N}  a_I P_k\\
    &= a_J^\dagger P_{k-d} \hat{N} P_{k-d}  a_I \\
    &= (k-d) P_k a_J^\dagger a_I P_k := (k-d)[M_{I,J}]_d^{(k)}    
\end{align}
where in the last two equalities, we applied the bosonic commutation relations, $k-d$ times.

By using the recursive relation in \cref{eq:raising_trivial}, we get that 
\begin{align}
    R_{d+1}([M_{I,J}]_d^{(d)}) &= (d+1-d) [M_{I,J}]_d^{(d+1)}\\
     R_{d+2} R_{d+1}([M_{I,J}]_d^{(d)}) &= 1 \times (d+2-d)  [M_{I,J}]_d^{(d+2)}\\
     &\vdots\\
     R_n \dots R_{d+1}([M_{I,J}]_d^{(d)}) &= (n-d)![M_{I,J}]_d^{(n)}
\end{align}
which concludes the proof.
\end{proof}

\begin{proposition}\label{propop:monomial_single_irrep_support_app}
     A monomial $[{M}_{\bm p, \bm q}]_d^{(n)}$ of degree $d\leq n$ acting on $n$ particles such that every indexes in $\bm p$ and $\bm q$ are different has only support on the irrep $\lambda_d^{(n)}$.
\end{proposition}
\begin{proof}
 From Proposition \ref{propop:degree_cutoff_app}, we know that the monomial $[M_{I,J}]_d^{(n)}$ has support on at most the first $d+1$ irreps of $W_n$. Moreover, we have 
 \begin{equation}
      [M_{I,J}]_d^{(n)} = \frac{1}{(n-d)!} R_n \dots R_{d+1}([M_{I,J}]_d^{(d)})  
 \end{equation}

 To prove that $[M_{I,J}]_d^{(n)}$ is fully supported on the irrep $\lambda_d^{(n)}$, it suffices to show that  $[M_{I,J}]_d^{(d)} \in Ker(L_d)$, since $Ker(L_d) = \lambda_d^{(d)}$ as established in \cref{prop:kernel_top_irrep}.

The action of the lowering map on a monomial of degree $d\leq k$ give rise to a sum of terms of degree $d-1, d $ and $d+1$ projected on the  subspace $W_{k-1}$. Precisely, we have 
\begin{align}
    L([M_{I,J}]_d^{(k)}) &= \sum_{s=1}^m a_s P_k (a^\dagger)^{\alpha(J)} a^{\beta(I)} P_k a_s^\dagger\\
    &= P_{k-1} \left(\sum_{s=1}^m  \left(\prod_{l \neq s} (a^\dagger)^{\alpha_l(J)}   \right)   a_s  (a^\dagger)^{\alpha_s(J)}    a^{\beta_s(I)}    a_s^\dagger  \left(\prod_{l \neq s} a^{\beta_l(I)}   \right) \right)P_{k-1} \label{eq:contract_dev}
\end{align}
where in the first equality we used the monomial expression introduced in \cref{eq:compact_monomial} and in the second equality, we used the bosonic commutation relations to pull back the projectors.

Using the commutation identity from Lemma \ref{lemma:com_property}, the term $a_s  (a^\dagger)^{\alpha_s(J)}    a^{\beta_s(I)}    a_s^\dagger$ can be expressed as 
\begin{align}
    a  (a^\dagger)^{\alpha}    a^{\beta}    a^\dagger &= \left( (a^\dagger)^{\alpha}a + \alpha (a^\dagger)^{\alpha-1} [a,a^\dagger] \right) \left( a^\dagger a^\beta + \beta a^{\beta-1} [a,a^\dagger]\right)\\
    &= \alpha \beta (a^\dagger)^{\alpha-1}  a^{\beta-1} + (\alpha +\beta) (a^\dagger)^{\alpha}  a^{\beta} + (a^\dagger)^{\alpha} a a^\dagger a^\beta\\
    &= \alpha \beta (a^\dagger)^{\alpha-1}  a^{\beta-1} + (\alpha +\beta) (a^\dagger)^{\alpha}  a^{\beta} + (a^\dagger)^{\alpha} ( a^\dagger a + \Id) a^\beta \\
    &= \alpha \beta (a^\dagger)^{\alpha-1}  a^{\beta-1}  + (\alpha +\beta+1) (a^\dagger)^{\alpha}  a^{\beta} +  (a^\dagger)^{\alpha+1} a^{\beta+1} \label{eq:contract_mode}
\end{align}

By plugging the result in \cref{eq:contract_mode} back in \cref{eq:contract_dev}, we get that the contraction map give rise to three monomials of degree $k-1, k$ and $k+1$ respectively, all projected onto the $k-1$ particle subspace.

For an observable of degree exactly $d=k$, the only surviving term in the contraction map expansion is the one corresponding to a degree $k-1$ monomial,
\begin{align}\label{eq:contraction_degree}
    C([M_{I,J}]_k^{(k)}) = P_{k-1} \left(\sum_{s=1}^m \alpha_s(J) \beta_s(I) (a^\dagger)^{J-e_s} a^{I-e_s}   \right)  P_{k-1} 
\end{align}
where $e_s$ is the vector of length $m$ with all entries zeros except for 1 in position $s$. 

Therefore the contraction on the monomial $[M_{I,J}]_d^{(d)}$ with $\sigma(I) \neq \sigma(J) \;,\forall \sigma \in S_d$ yields zero since $\alpha_s(J)\beta_s(J) = 0 \; \forall 1 \leq s \leq m $. 
\end{proof}

\begin{proposition}\label{propop:MonomialLastIrrepSupport_app}
    A monomial $[{M}_{\bm p, \bm q}]_d^{(n)}$ of degree $d\leq n$ acting on $n$ particles such that $|\p \cap \q| = k$, i.e., $\bm p$ and $\bm q$ have exactly $k-1$ common indices, has support only on the $K$ consecutive irreps ending at $\lambda_d^{(n)}$, namely $\lambda_{d-k+1}^{(n)},\dots,\lambda_d^{(n)}$.
\end{proposition}
\begin{proof}
    Combining the result from Proposition \ref{propop:degree_cutoff_app} and the result from \cref{eq:irrep_expression} in Proposition \ref{prop:lower_irreps},
    we get that 
    \begin{align}
        [M_{I,J}]_d^{(n)} &= \frac{1}{(n-d)!} R_n \dots R_{d+1}([M_{I,J}]_d^{(d)}) \\
        [M_{I,J}]_d^{(d)} &= \sum_{j=0}^d R_d \dots R_{j+1}(H_j) \label{eq:monomial_decomp_irrep}
    \end{align}
    where $H_j \in \lambda_j^{(j)}$ depend on $[M_{I,J}]_d^{(d)}$.

    By applying the lowering map on $[M_{I,J}]_d^{(d)}$ $K$ times, we get
    \begin{align}
     L_{d-K+1}\dots L_{d-1} L_d([M_{I,J}]_d^{(d)}) &= \sum_{j=0}^d L_{d-K+1}\dots L_{d-1} L_dR_d \dots R_{j+1}(H_j)\\
      &= \sum_{j=0}^{d-K} L_{d-K+1}\dots L_{d-1} L_dR_d \dots R_{j+1}(H_j) + \sum_{j=d-K+1}^d  L_{d-K+1} \dots  L_j (L_{j+1}\dots L_{d-1} L_dR_d \dots R_{j+1})(H_j) \label{eq:contract1}
    \end{align}

    The second term in \cref{eq:contract1} turns out to be simply zero. Precisely,  by exploiting the fact that the map $L_{j+1}\dots L_{d-1} L_dR_d \dots R_{j+1}$ maps $\rm Ker(L_j)$ to $\rm Ker(L_j)$, we get that
    \begin{align}
       \sum_{j=d-K+1}^d  L_{d-K+1} \dots  L_j (L_{j+1}\dots L_{d-1} L_dR_d \dots R_{j+1})(H_j) =0 
    \end{align}

    Hence, the application of the contraction map $K$ times on $[M_{I,J}]_d^{(d)}$ cancels out its components  on the top $K$ irreps, i.e.  $\lambda_{d-K+1}^{(d)}, \dots, \lambda_d^{(d)}$ and \cref{eq:contract1} reduces to 
    \begin{align}
        L_{d-K+1}\dots L_{d-1} L_d([M_{I,J}]_d^{(d)}) &= \sum_{j=0}^{d-K} L_{d-K+1}\dots L_{d-1} L_dR_d \dots R_{j+1}(H_j)\\
        &=( L_{d-K+1}\dots L_{d-1} L_dR_d \dots R_{d-K+1}) \left( \sum_{j=0}^{d-K} R_{d-K} \dots R_{j+1}(H_j)\right) \label{eq:contract2}
    \end{align}
    where we recall that each element in the above sum of the form $R_{d-K} \dots R_{j+1}(H_j)$ belongs to the irrep $\lambda_j^{(d-K)}$ according to the result in \cref{eq:irrep_expression} from Proposition \ref{prop:lower_irreps}. Moreover the map $L_{d-K+1}\dots L_{d-1} L_dR_d \dots R_{d-K+1}$ acts as a scalar multiple of identity on such irreps as proved in Proposition \ref{prop:CD_DC}.

    Now, we show that applying the lowering map $K$ times on a monomial $[M_{I,J}]_d^{(d)}$  with $I$ and $J$ overlapping on $K$ modes yields zero. Combined with \cref{eq:contract2}, this implies that the monomial has zero support on the first $d+1-K$ irreps, thus proving the desired result.

    As shown in \cref{eq:contraction_degree}, we have that a single application of the lowering map results in a sum of degree $d-1$ monomials, i.e.
    \begin{align}
       C([M_{I,J}]_d^{(d)}) &=  \sum_{s=1}^m \alpha_s(J) \beta_s(I) \left( P_{k-1}(a^\dagger)^{J-e_s} a^{I-e_s}    P_{k-1} \right)\\
       &=\sum_{s=1}^m \alpha_s(J) \beta_s(I) [M_{J-e_s,I-e_s}]_{d-1}^{(d-1)}
    \end{align}

    Since each monomial of the form $[M_{J-e_s,I-e_s}]_{d-1}^{(d-1)}$ appearing in the above equation is of degree $d-1$ acting on $d-1$ particles, the  application of the lowering map again will further reduce its degree by one. Iteratively, one can see that after applying the lowering map $K$ times, we do indeed get zero, which concludes the proof.
\end{proof}

\newpage
\section{Simulation Techniques}\label{app:Simulation_Techniques}

Let ${M}_{\bm p, \bm q}$ be a monomial in creation annihilation operators of degree $k$ as defined in \cref{eq:monomial} 
\begin{equation*}
    {M}_{\bm p, \bm q} = a_{1}^{p_1\dagger} \dots a_{m}^{p_m\dagger} a_1^{q_1} \dots a_m^{q_m}
\end{equation*}

 with $\bm p, \bm q \in \Phi_m^k$ and consider the Fock state  $\ketbra{S}{S}$ with $n$ photons. 
In \cite{thomas_shedding_2025}, the authors show that the expectation value $f_U(\ketbra{S}{S},{M}_{\bm p, \bm q})$ of the form in \cref{eq:linear_form_n}  can be expressed as
\begin{equation}\label{eq:expValPerm}
f_U(\ketbra{S}{S},{M}_{\bm p, \bm q})  
    = \sum_{\substack{\bm \ell \in \Phi_m^k \\ \bm \ell \preccurlyeq \bm n}}\frac{\bm n!}{(\bm n - \bm \ell)!\bm\ell!^2} \per{U_{\bm p, \bm \ell}}^* \per{U_{\bm q, \bm \ell}},
\end{equation}
where $\bm \ell \preccurlyeq \bm n$ is the partial entry-wise order. From \cref{eq:expValPerm}, it is clear that several quantities can be leveraged to efficiently compute either exactly or approximately the expectation value, namely, the number of summands and the dimension and rank of $U_{\bm p, \bm \ell}$ and $U_{\bm q, \bm \ell}$. Finally,
we observe that $\bm p = \bm q$ allows us to recover the result of
\cite{mayer_counting_2011}.

    \subsection{Proof of \cref{simtech:direct_expansion}}

In order to prove \cref{simtech:direct_expansion}, we present different simulation techniques. First, based on \cref{eq:expValPerm}, we bound the number of summands required to approximate $f_U(\ketbra{S}{S},{M}_{\bm p, \bm q})$ via direct expansion.

\begin{proposition}[Cardinality of the summand]\label{prop:cardGoodSummands}
    Let $n \in \mathbb{N}, 1 \leq k \leq n$ and $m \geq n$. Then, for an occupation $\bm n \in \Phi_m^n$, it holds that
    \begin{equation}
        \left|\left\{\bm \ell \ | \ \bm \ell \in \Phi_m^k : \bm \ell \preccurlyeq \bm n\right\}\right| \leq \binom{ \min(k, n-k) + \occ S - 1}{\occ S-1}.
    \end{equation}
\end{proposition}
\begin{proof}
Indeed, there are at most $|\Phi_m^k| = \binom{m+k-1}{k}$ such matrix
permanents. However, this bound can be very loose, as for instance if $S = (n,
0, \dots, 0)$, only $\bm \ell = (k, 0, \dots, 0)$ satisfies $\bm \ell
\preccurlyeq S$. To improve on this, we count the number of $\bm \ell \in
\Phi_m^k$ satisfying $\bm \ell \preccurlyeq S$. The number of such $\bm\ell$
can be found from the generating function 
\begin{equation}\label{eq:Gn}
    \begin{aligned}
        g_{\bm n}(x)
        = \prod_{i=1}^m (1 + x + x^2 + \cdots + x^{n_i})
        =  \frac{1}{(1-x)^m} \prod_{i=1}^m (1- x^{n_i +1}).
    \end{aligned}
\end{equation} 
In particular, one is interested in the coefficient of $x^k$ in $g_{S}(x)$,
which we write $[x^k] g_{\bm n}(x)$. First, using the binomial theorem with
negative exponent, we find
\begin{equation}
    (1+x)^{-m} = \sum_{j \geq 0} \binom{m+j-1}{j}x^j.
\end{equation}
Second, the right-most product of \cref{eq:Gn} expands as 
\begin{equation}
    \prod_{i=1}^m (1- x^{n_i +1}) = \sum_{S \subseteq \{1, \cdots, m\}} (-1)^{|S|} x^{\sum_{i\in S} n_i +1}.
\end{equation}
Combining everything yields the number of $\bm \ell$ contributing to the sum of
\cref{eq:expValPerm}:
\begin{equation}\label{eq:coefficientOfTk}
    \begin{aligned}
        |\{\bm \ell\  |\ \bm \ell \in \Phi_m^k,\, \bm \ell \preccurlyeq S \} |
            & = [x^k]g_{S}(x)  \\
            & = [x^k] \sum_{j \geq 0} \binom{m+j-1}{j}x^j \sum_{S \subseteq \{1, \cdots, m\}} (-1)^{|S|} x^{\sum_{i\in S} x_i +1} \\
            & = \sum_{S \subseteq \{1, \cdots, m\}} (-1)^{|S|} \binom{k-\lpr\sum_{i\in S}n_i + 1\rpr + m - 1}{m-1}.    \end{aligned}
\end{equation}
Moreover, $[x^k]g_{S}(x)$ is symmetric with respect to $n$, namely
$[x^k]g_{S}(x) = [x^{n-k}]g_{S}(x)$ for $k>1$. To simplify the analysis, we
now upper bound \cref{eq:coefficientOfTk}. Let $\occ\s = |\{i \ | \ s_i >
0\}|$ be the number of non-empty modes in the occupancy vector $S$. As only
these modes give freedom in the choice of $\bm\ell$ (since $n_i = 0$ enforces
$l_i = 0$), letting $b = \min(k, n-k)$ we obtain:
\begin{equation}\label{eq:upperBoundTk}
    [x^k]g_{S}(x) \leq \binom{ b + \occ S - 1}{\occ S-1},
    % \sim \sqrt{\frac{b+s-1}{2\pi(s-1)b}} \exp{b\log(1 + \frac{s}{b}) + s\log(1 + \frac{b}{s})},
\end{equation}
with equality when $b = 1$. Interestingly, for
$k=n$, \cref{eq:coefficientOfTk} yields $[x^n]g(x) = 1$.
\end{proof}

We now present two simulation methods that can be used to compute exactly $f_U(\ketbra{S}{S},{M}_{\bm p, \bm q})$ in polynomial time:

\begin{simtech}[Ryser based]\label{prop:efficientCompExpValRyser}
    $f_U(\ketbra{S}{S},{M}_{\bm p, \bm q})$ can be computed exactly if $k= \mathcal{O}(1)$ or $\occ S = \mathcal{O}(\log n)$ and $k = \mathcal{O}(\log n)$ 
\end{simtech}

\begin{proof}
The matrices $U_{\bm p, \bm \ell}$ and $U_{\bm q, \bm \ell}$ are $k\times k$
matrices whose permanent can be computed exactly in time $\bigo{k2^k}$ using
Ryser's method \cite{ryser_combinatorial_1973}. The constraints we impose on $k$ ensures the runtime of Ryser's algorithm is polynomial. To conclude, in order for the bound of \cref{prop:cardGoodSummands} it be polynomial in $n$, resulting in polynomially-many matrix permanents to compute, one must ensure $k= \mathcal{O}(1)$ or or $\occ S = \mathcal{O}(\log n)$ and $k = \mathcal{O}(\log n)$ as claimed.
\end{proof}

\begin{simtech}[Barvinok based]\label{prop:efficientCompExpValBarvinok}
    $f_U(\ketbra{S}{S},{M}_{\bm p, \bm q})$ can be computed exactly if $k = \mathcal{O}(n)$ and $\occ S = \mathcal{O}(1)$ or $\loc{{M}_{\bm p, \bm q}} = \max(\occ{\bm p}, \occ{\bm p}) = \mathcal{O}(1)$.
\end{simtech}
\begin{proof}
Even for large $k$, it is still
possible to compute the permanent exactly if the matrix has small rank using
Barvinok's algorithm \cite{barvinok_two_1996}.

More precisely, Barvinok's algorithm allows one to compute the permanent of a $\rank(r)$ matrix $A \in \mathbb C^{n\times n}$ in time $n^{\mathcal{O}(r)}$. Now, observe that 
$\rank(U_{\bm p, \bm
    \ell}) = \min(\occ{\bm{p}}, \occ{\bm\ell}) \leq \min(\occ{\bm{p}}, \occ{\bm n}) $, as the number of nonzero entries in $\bm \ell$ is at
    most $\occ{\bm n}$,

\end{proof}

\newpage
\section{Initial state purities}\label{app:Initial_State_Purities}

    \subsection{Irrep purities of Fock states}\label{subsec:First_Irrep_Purity}

\begin{lemma}{(Irrep purities  for Fock states, adapted from \cite{mhiri_boson_2026}).}\label{lemma:puR_fock_2}
    Consider a Fock state $R$ with a total number of $n$ photons, i.e $R \in \Phi_m^n$. For every $0\leq k \leq n$, the term $g_k(\ketbra{R}{R})$ defined in \cref{eq:g_k_th} can be expressed as 
    \begin{equation}
        g_k(\ketbra{R}{R}) = ((n-k)!)^2 \sum_{\substack{|\boldsymbol{b}| = n-k\\0 \leq \boldsymbol{b}\leq R}} \prod_{i=1}^m \binom{R_i}{b_i}^2\;.
    \end{equation}

    The $k-$purity is then given by 
    \begin{equation}
    \left\| P_k^{(n)}\!\left(\ketbra{S}\right) \right\|_2^2 = \frac{1}{\alpha_{k,k+1,n}} \sum_{l=0}^k \frac{(-1)^{k-l}}{(k-l)! (k+l+ m-1)_{k-l}} g_l\;,
\end{equation}
where $(x)_p$ denotes the Pochhammer symbol defined as $(x)_p := x (x+1) \dots (x+p-1)$ and $\alpha_{k,k+1,n} = \frac{(n-k)! (m+n+k-1)!}{(m+2k-1)!}$.

In particular, 
\begin{equation*}
    \left\| P_1^{(n)}\!\left(\ketbra{S}\right) \right\|_2^2 = \binom{n+m}{n-1}^{-1} \left(\|S\|_2^2-\frac{n^2}{m} \right).
\end{equation*}
\end{lemma}

\begin{proof}
Setting $k=1$ in \cref{th:recursive_irrep_norm_app}, the inner sum has only two terms.
Evaluating the Pochhammer factors $(k+l+m-1)_{k-l}$ at $(k,l)=(1,0)$ and $(1,1)$ gives
$(m)_1=m$ and $(m+1)_0=1$, so
\begin{equation*}
    \sum_{l=0}^{1}\frac{(-1)^{1-l}}{(1-l)!\,(1+l+m-1)_{1-l}}\,g_l
    = g_1 - \frac{1}{m}\,g_0 .
\end{equation*}
With $\alpha_{1,2,n}=\dfrac{(n-1)!\,(m+n)!}{(m+1)!}$ this yields
\begin{equation}\label{eq:P1-intermediate}
    \left\| P_1^{(n)}\!\left(\ketbra{S}\right) \right\|_2^2
    = \frac{(m+1)!}{(n-1)!\,(m+n)!}\left(g_1-\frac{g_0}{m}\right).
\end{equation}

It remains to evaluate $g_0$ and $g_1$ for the Fock state $S\in\Phi_m^n$ via \cref{lemma:puR_fock_2}.
For $k=0$, the constraints $|\boldsymbol b|=n$ and $\boldsymbol 0\le\boldsymbol b\le S$ together with $|S|=n$
force $\boldsymbol b=S$, whence $\prod_i\binom{S_i}{S_i}^2=1$ and
\begin{equation*}
    g_0 = (n!)^2 .
\end{equation*}
For $k=1$, $|\boldsymbol b|=n-1$ with $\boldsymbol b\le S$ forces $\boldsymbol b=S-\boldsymbol e_j$ for some mode $j$
with $S_j\ge 1$; then $\prod_i\binom{S_i}{b_i}^2=\binom{S_j}{S_j-1}^2=S_j^2$, and summing over $j$
(empty modes contribute nothing),
\begin{equation*}
    g_1 = ((n-1)!)^2\sum_{j=1}^m S_j^2 = ((n-1)!)^2\,\|S\|_2^2 .
\end{equation*}
Using $(n!)^2=n^2\,((n-1)!)^2$,
\begin{equation*}
    g_1-\frac{g_0}{m} = ((n-1)!)^2\left(\|S\|_2^2-\frac{n^2}{m}\right),
\end{equation*}
and substituting into \eqref{eq:P1-intermediate},
\begin{equation*}
    \left\| P_1^{(n)}\!\left(\ketbra{S}\right) \right\|_2^2
    = \frac{(m+1)!\,(n-1)!}{(m+n)!}\left(\|S\|_2^2-\frac{n^2}{m}\right)
    = \binom{n+m}{n-1}^{-1}\!\left(\|S\|_2^2-\frac{n^2}{m}\right). \qedhere
\end{equation*}
\end{proof}

\begin{theorem}[Purity of the top irrep of a Fock state]\label{th:top_irrep_fock_bound}
Let $S \in \Phi_m^n$. Then
\begin{equation}\label{eq:top_irrep_fock_bound}
    1-\frac{\|S\|_2^2}{m+2n-2}
    \;\leq\;
    \left\| P_n^{(n)}\!\left(\ketbra{S}\right) \right\|_2^2
    \;\leq\;
    1-\frac{\|S\|_2^2}{n(m+n-1)} \;.
\end{equation}
\end{theorem}

\begin{proof}
Write $O := \ketbra{S} \in W_n$, so that $P_r^{(n)}(O) = O_r^{(n)}$ is the component of $O$ on
$\lambda_r^{(n)}$ and $\|P_r^{(n)}(O)\|_2^2 = \Tr[(O_r^{(n)})^2]$. The proof rests on the
observation that $g_{n-1}(O)$ is a positive combination of the irrep purities of $O$, whose
weights are the eigenvalues $\beta_{r,n-1}$ of \cref{prop:CD_DC}, and that these weights are
sharply bracketed.

It follows from \cref{eq:C_map,eq:D_map} and the bosonic commutation relations that $\Dm$ is the
adjoint of $\Cm$ with respect to the Hilbert--Schmidt inner product, since for $A \in W_{n-1}$
and $B \in W_n$
\begin{equation}\label{eq:CD_adjoint}
    \left\langle \Cm_n(B) , A \right\rangle
    = \sum_{s=1}^m \Tr\!\left[a_s^\dagger B^\dagger a_s A\right]
    = \left\langle B , \Dm_n(A)\right\rangle \;.
\end{equation}
Moreover the lowering map preserves the irrep label: for $r<n$ we have
$\lambda_r^{(n)} = \Dm_n(\lambda_r^{(n-1)})$ by \cref{eq:irrep_expression}, so that
$\Cm_n(\lambda_r^{(n)}) = \Cm_n \circ \Dm_n(\lambda_r^{(n-1)}) \subseteq \lambda_r^{(n-1)}$ by
\cref{prop:CD_DC}, while $\Cm_n$ annihilates $\lambda_n^{(n)} = \rm{Ker}(\Cm_n)$ according to
\cref{eq:top_irrep_ker}. Decomposing $O = \sum_{r=0}^n O_r^{(n)}$ along \cref{eq:decomp_Wn}, the
operators $\Cm_n(O_r^{(n)})$ therefore lie in pairwise distinct irreps of $W_{n-1}$ and are
consequently pairwise orthogonal, whence, using \cref{eq:CD_adjoint} with $B = O_r^{(n)}$ and
$A = \Cm_n(O_r^{(n)})$ together with
$\Dm_n \circ \Cm_n|_{\lambda_r^{(n)}} = \beta_{r,n-1}\Id$ from \cref{prop:CD_DC},
\begin{equation}\label{eq:gnm1_expansion}
    g_{n-1}(O)
    = \sum_{r=0}^{n} \left\| \Cm_n\!\left(O_r^{(n)}\right)\right\|_2^2
    = \sum_{r=0}^{n} \beta_{r,n-1} \left\| P_r^{(n)}(O)\right\|_2^2 \;.
\end{equation}

Both sides of \cref{eq:gnm1_expansion} are explicit for a Fock state. On the one hand, since
$a_s\ket{S} = \sqrt{S_s}\ket{S-e_s}$, the map \cref{eq:C_map} gives
$\Cm_n(O) = \sum_{s=1}^m S_s \ketbra{S-e_s}$, where the states $\ket{S-e_s}$ with $S_s \geq 1$ are
pairwise orthogonal, so that
\begin{equation}\label{eq:gnm1_fock}
    g_{n-1}(O) = \sum_{s,t=1}^m S_sS_t\left|\braket{S-e_s|S-e_t}\right|^2 = \|S\|_2^2 \;.
\end{equation}
On the other hand the decomposition \cref{eq:decomp_Wn} is orthogonal, so the irrep purities of
$O$ sum to $\Tr[O^2]=1$.

It remains to bracket the weights $\beta_{r,n-1} = (n-r)(m+n+r-1)$. They vanish at $r=n$ and are
strictly decreasing in $r$, since $\partial_r\beta_{r,n-1} = 1-m-2r<0$, so that
$m+2n-2 = \beta_{n-1,n-1} \leq \beta_{r,n-1} \leq \beta_{0,n-1} = n(m+n-1)$ for $0 \leq r \leq n-1$.
Inserting these two inequalities into \cref{eq:gnm1_expansion}, whose term $r=n$ vanishes, and
using \cref{eq:gnm1_fock} together with
$\sum_{r=0}^{n-1}\|P_r^{(n)}(O)\|_2^2 = 1 - \|P_n^{(n)}(O)\|_2^2$, we obtain
\begin{equation*}
    (m+2n-2)\left(1-\left\|P_n^{(n)}(O)\right\|_2^2\right)
    \;\leq\; \|S\|_2^2 \;\leq\;
    n(m+n-1)\left(1-\left\|P_n^{(n)}(O)\right\|_2^2\right) \;,
\end{equation*}
and dividing the left inequality by $m+2n-2$, the right one by $n(m+n-1)$, and solving both for
$\|P_n^{(n)}(O)\|_2^2$ yields \cref{eq:top_irrep_fock_bound}.
\end{proof}

    \subsection{Irrep profile discussion}

We propose additional simulations and comments for the purity distribution of initial state in passive linear optics. First, we observe that the distribution of Fock state $\ket{R}$ irrep purities is determined by the occupation integer string $R = (R_1, \dots, R_m)$ such that $R \in \Phi_{m,n}$:

\begin{figure}[h!]
    \centering
    \includegraphics[width=0.89\linewidth]{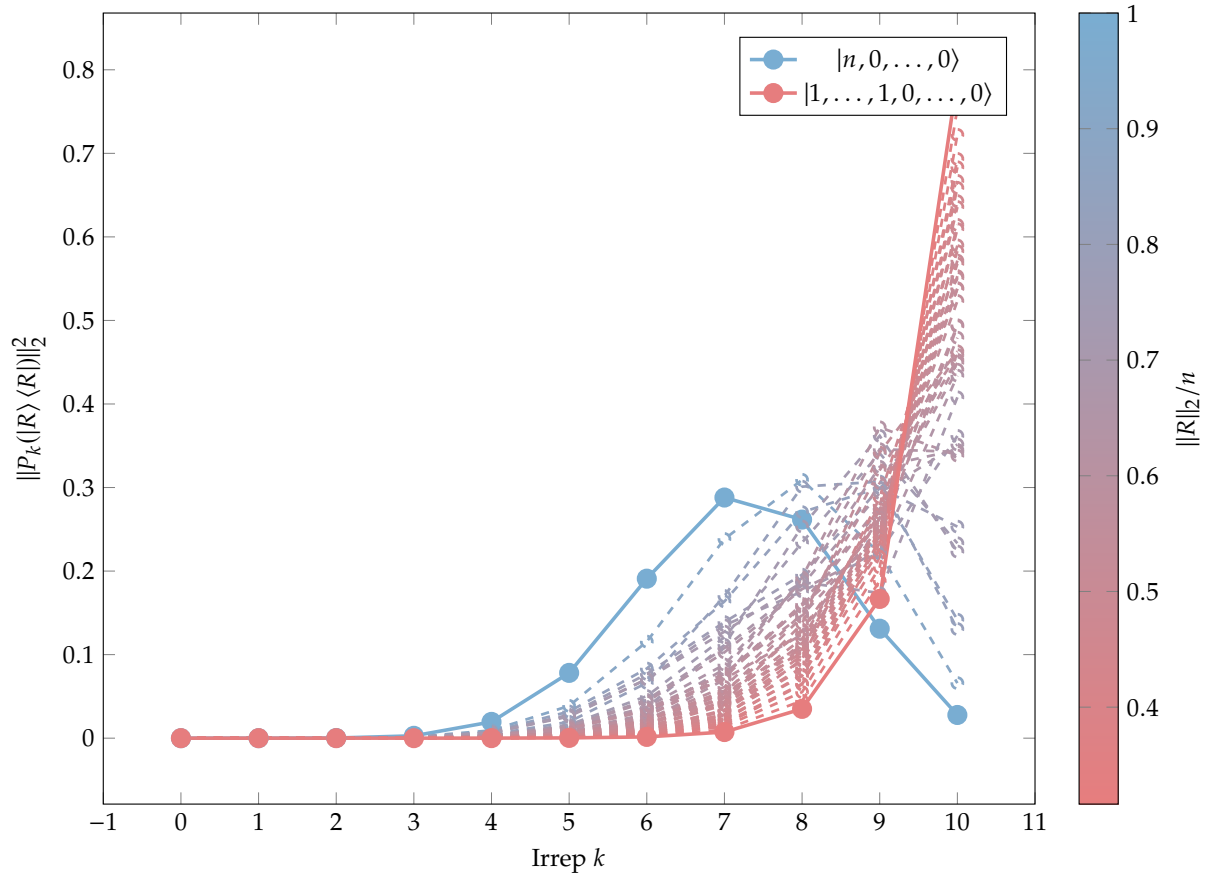}
    \caption{\justifying Fock state purity distribution for $n=10$ particles in $m=2n$ modes. We only consider the Fock states that are all different under any permutation of the occupation.}
    \label{fig:Fock_State_Irrep}
\end{figure}

To derive the projection of a Fock state $\ket{R}$ over the different irreps, one can use the technique presented in \cref{app:2ndMoment_Calculus_Iterative_Projection} by considering $O = \ket{R}\bra{R}$. Notice that any permutation of the modes do not affect the result of \cref{th:recursive_irrep_norm_app}, only $\occ{R}$ affect the decomposition for a fixed number of particles $n$. In this work, we restrict the study of the second moment to the consideration of input Fock state with minimum and maximum bunching. We observe that for small irreps, the purities are sorted according to $||R||_2$, as illustrated by \cref{fig:Fock_States_Scatter_k_const}. Fock state with more bunching, i.e., with particles in the same modes, the maximum value for the norm of the purity is achieved for smaller irrep. In \cref{fig:Fock_States_Scatter_k_const}, we show that the extreme case is given by the state with maximum bunching where all the particles are in the same mode, and is achieved for the irrep $k \simeq 3n/4$.

\newpage

\begin{figure}[t!]
    \centering
    \includegraphics[width=0.89\linewidth]{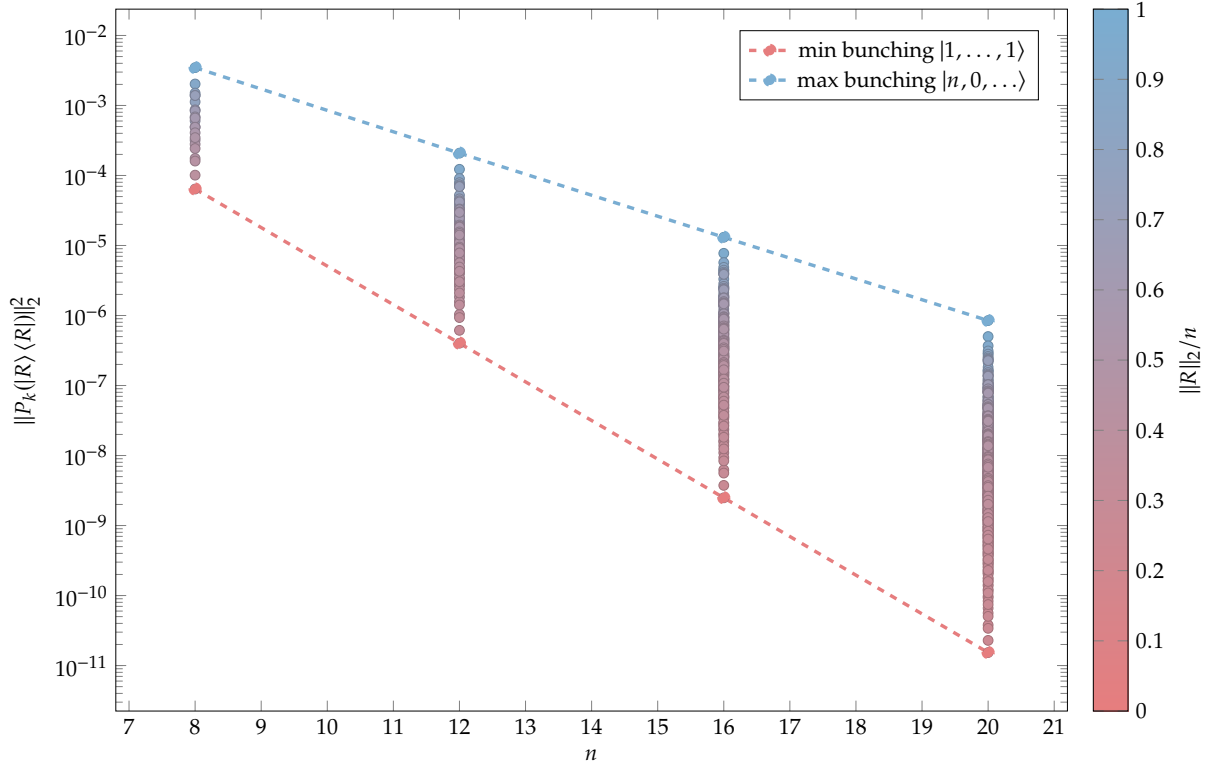}
    \caption{\justifying Fock state purity on the irreducible component $k=\lfloor n/4 \rfloor$ for different number of particles $n$, and $m=2n$ modes.}
    \label{fig:Fock_States_Scatter_k_const}
\end{figure}

However, as illustrated in \cref{fig:Fock_State_Irrep}, 
We observe that Fock state with bunching maximise their support on irreps  

\begin{figure}[h!]
    \centering
    \includegraphics[width=0.89\linewidth]{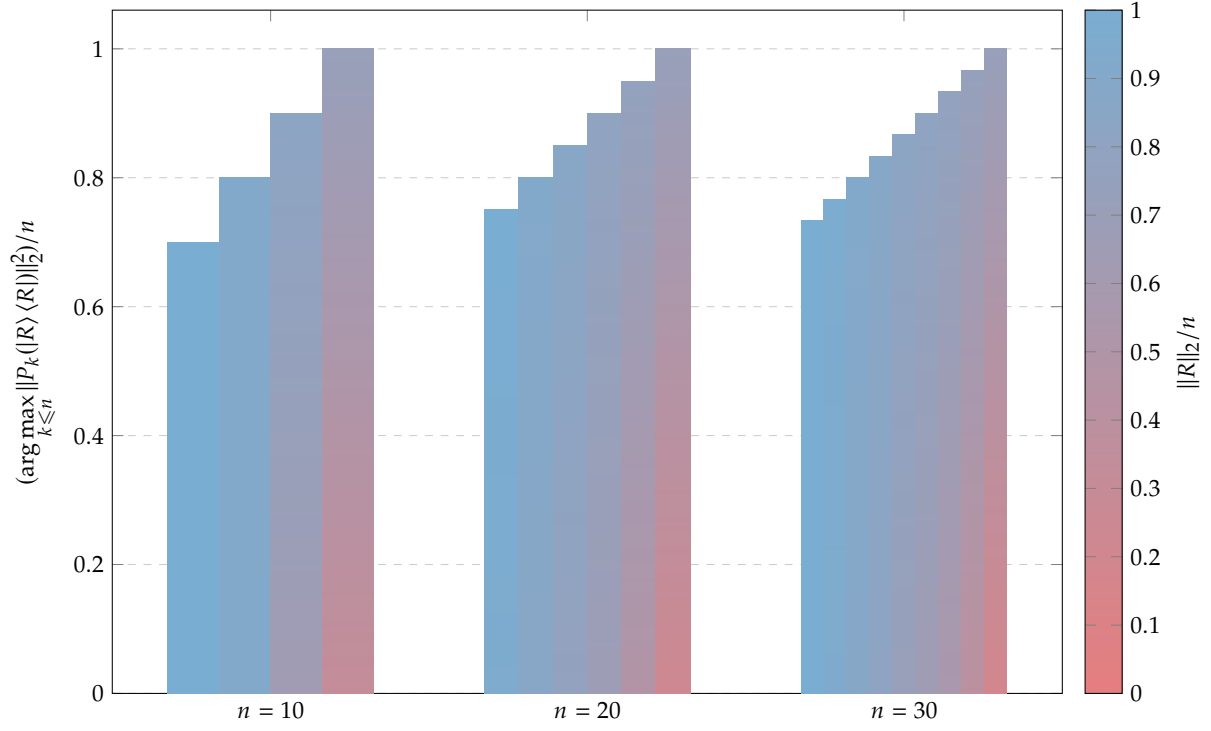}
    \caption{\justifying Irreducible components normalized index $k/n$ that maximize the Fock state purity $\|P_k(\ket{R}\bra{R})\|_2^2$ for different number of particles $n$, and $m=2n$ modes. The gradient color of the bars show every Fock state, according to their $2$-norm divided by the number of particles, that have maximum support on the same irrep.}
    \label{fig:Fock_States_Kstar_Scatter}
\end{figure}

\newpage

\newpage

% ─────────────────────────────────────────────
\section{Observable Study - Product of Number Operators}\label{app:Product_Number_Operators}
% ─────────────────────────────────────────────

    \subsection{Definition and properties:}

We consider the case of product of photon number operators:
\begin{equation}
    O = \prod_{i \in I} \hat{n}_i = \sum_{R \in \Phi_{m,n}} \left( \prod_{i \in I} R_i \right) \ket{R}\bra{R}
\end{equation}

The $2$-norm of this observable is given by:
\begin{equation}
    ||O||_2^2 = \sum_{R \in \Phi_{m,n}} \left( \prod_{i \in I} R_i \right)^2    
\end{equation}

And its infinite norm is given by:
\begin{equation}
    ||O||_\infty = \max_{R \in \Phi_{m,n}} \left( \prod_{i \in I} R_i \right)
\end{equation}

    \subsection{Irrep Decomposition:}

The $k$-application of the lowering map $\Cm^k$, introduced in \cref{eq:C_map}, on this operator is given by:
\begin{equation}
    \Cm^k(O) = \sum_{R \in \Phi_{m,n}} (\prod_{i \in I} R_j) \sum_{|b|=k, b \leq R} \left( \frac{k!}{b_1! \dots b_m!} \prod_{l=1}^m \frac{R_l!}{(R_l - b_l)!} \right) \ket{R-b}\bra{R-b}
\end{equation}

By using the change of variable $S = (S_1 \dots S_m)$, with $|S|=n-k$, and $R = S + b$, we obtain
\begin{equation}
    \Cm^k(O) = \sum_{S \in \Phi_{m,n-k}} \lambda_S \ket{S}\bra{S} \;,
\end{equation}
and $\lambda_S = \sum_{|b|=k} \left(\prod_{j \in J}(S_j + b_j)\right) \frac{k!}{b_1! \dots b_m!} \prod_{l=1}^m \frac{(S_i + b_i)!}{S_i!}$.

Because $\Cm^k(O)$ is diagonal, we have: 
\begin{equation}
    Tr[(\Cm^k(O))^2] = \sum_{S \in \Phi_{m,n-k}} \lambda_S^2
\end{equation}

\begin{figure}
    \centering
    \includegraphics[width=0.89\linewidth]{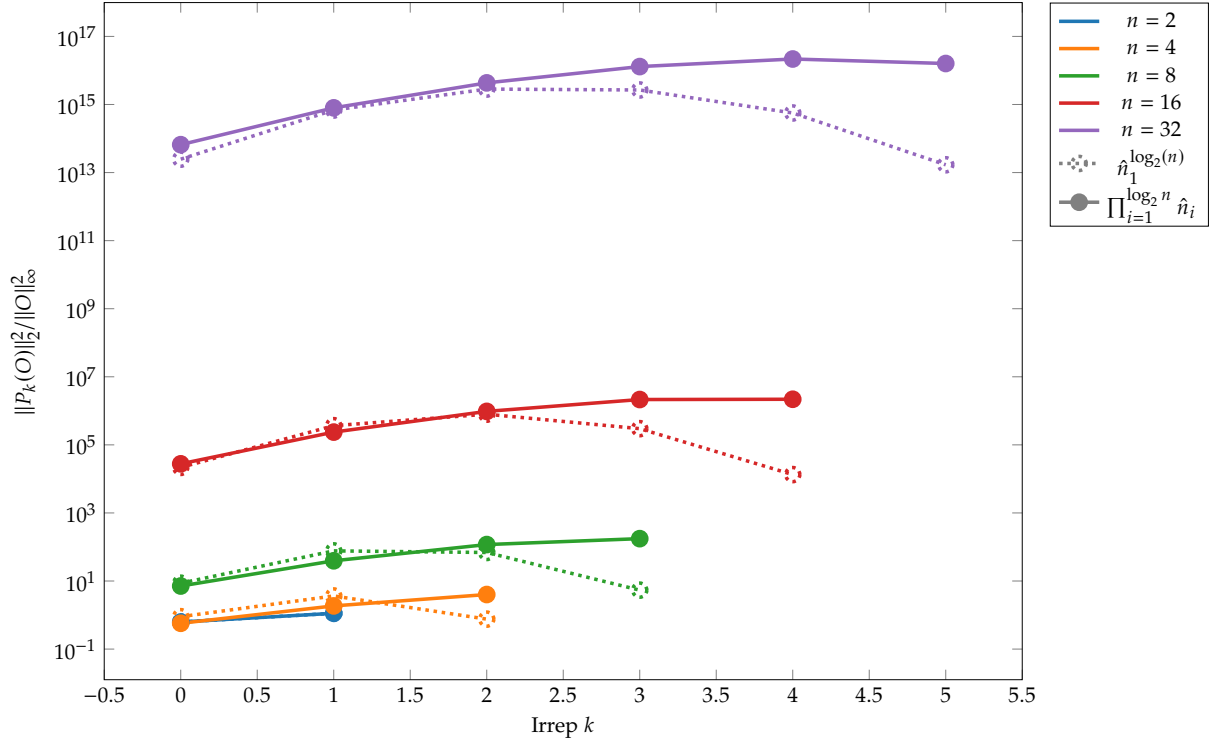}
    \caption{\justifying Purity evolution of a product of number operators. We consider the case where all the number operators are applied to the same mode, and the case where they are all applied to different modes. The number of terms in the product evolves logarithmically with $n$.}
    \label{fig:Purities_Evolution_Product_Number_Operators}
\end{figure}

Notice that, as explained by \cref{propop:degree_cutoff}, the product of $p \leq n$ number operators gives an operator of degree $p$, and thus only has support on the first $p+1$ irreps  of $W_n$.

    \subsection{Classical Simulation}

Due to the fact that a product of $p \leq n$ parity number operator is an observable of degree $p$, and thus only has support on the $p+1$ first irrep, a passive linear optics model based on such observable can always be simulate classically using \cref{simtech:exact_gsim}.

In addition, such observable can be written as a sum of monomial terms ${M}_{\bm p, \bm q}$ of degree $p$. The number of term depends on the number of distinct modes in the product. Indeed, one can write: 
\begin{equation}
     O = \prod_{i \in I} \hat{n}_i = \prod_{i \in I} \hat{a}_i \hat{a}^\dagger_i \,.
\end{equation}
Because the bosonic creation and annihilation operators satisfy the following properties
\begin{equation}
    [a_i,a_j^\dagger]= \delta_{i,j}, \qquad [a_i,a_j]=[a_i^\dagger,a_j^\dagger]=0 \, ,
\end{equation}
we have that the number of monomial terms depends on the number of different modes used in the product. When all the number operators are applied on different mode, every creation and annihilation operators commute and $O$ can be written as a monomial term of degree $p$. Therefore, one can use \cref{simtech:direct_expansion} to simulate the corresponding quantum model expectation value.

    \subsection{First Moment Studies:}

We consider the first moment of the observable expectation value, for different values of $p$.

\begin{figure}[h!]
    \centering
    \includegraphics[width=1.0\linewidth]{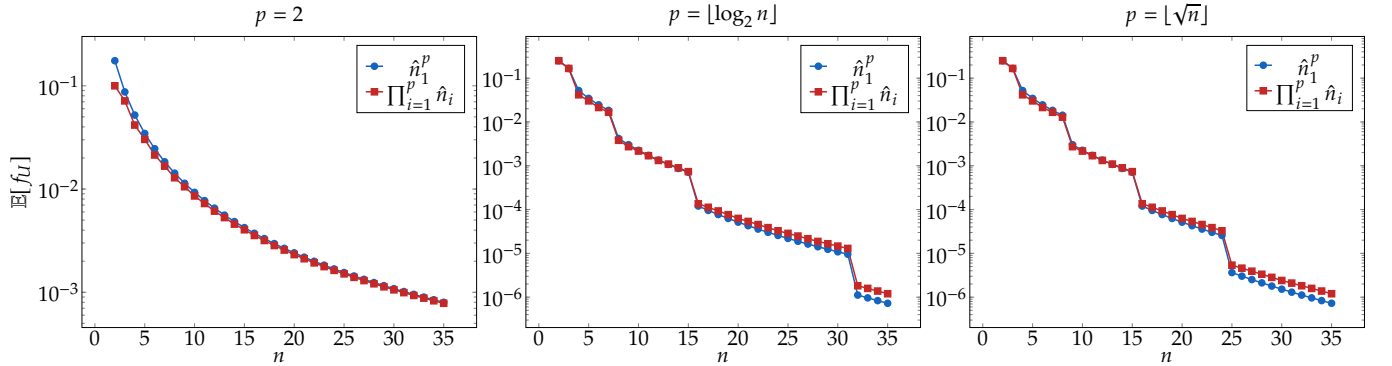}
    \caption{Study of the first moment for the expectation value of a product of photon number operators.}
    \label{fig:First_Moment_Product_Number_Operators}
\end{figure}

    \subsection{Second Moment Studies:}

To study the second moment, we propose to verify numerically the decay of the second moment given in \cref{eq:sec_m_exp}, and recalled here:
\begin{equation*}
    \e[U \sim U(m)]{f_U(\rho,O)^2} =  \sum_{k=0}^n \frac{\|P_k^{(n)}(\rho)\|_2^2 \|P_k^{(n)}(O)\|_2^2}{d_{k}^{(n)}} \,.
\end{equation*}

In order to make sure that a polynomial decay is not only due to the first irreps, which would allow to apply \cref{simtech:approx_gsim}, we propose to also study the higher terms of the second moment:
\begin{equation}
    \sum_{k=k_0}^n \frac{\|P_k^{(n)}(\rho)\|_2^2 \|P_k^{(n)}(O)\|_2^2}{d_{k}^{(n)}} \,,
\end{equation}
for different scaling of $k_0$, with $k_0=1$ corresponding to the entire second moment. In order to make the observable efficient on a quantum processor, we normalized the observable by the square of its infinite norm. We compute the evolution of this quantity of interest, and we provide the behaviour diagnosis through the regression match of each curve, with corresponding R2 value:

\begin{table*}[h!]
\centering
\begin{tabular}{|c|c|c|c|c|c|c|c|c|c|}
    \hline
    \multirow{3}{*}{\diagbox[width=1.6cm, height=1.4cm]{\small$p$}{\small$O$}}
    & \multirow{3}{*}{} 
    & \multicolumn{4}{c|}{$\hat{n}_1^p$} 
    & \multicolumn{4}{c|}{$\prod_{i=1}^p \hat{n}_i$} 
    \\ \hhline{~|~--------|}

     % --- Sub-header row: cols 1-2 empty (covered by multirow), cols 3-14 alternate ---
    & Input
    & \multicolumn{2}{c|}{$k_0=1$} 
    & \multicolumn{2}{c|}{$k_0=\lfloor\log_2 n \rfloor$}
    & \multicolumn{2}{c|}{$k_0=1$} 
    & \multicolumn{2}{c|}{$k_0=\lfloor\log_2 n \rfloor$}
    \\ \hhline{~|~--------|}
    
    % --- Sub-header row: cols 1-2 empty (covered by multirow), cols 3-14 alternate ---
    & 
    & Decay & $R^2$
    & Decay & $R^2$
    & Decay & $R^2$
    & Decay & $R^2$
    \\ \hline
    %------------------------------------------------------------------
    \multirow{4}{*}{$2$}
    & \multicolumn{1}{@{}c@{}|}{\cellcolor{tabred!60}Min}
    & \small\textcolor{blue!70!black}{\textbf{P}} $(n^{-3.90})$
    & \small\textcolor{blue!70!black}{\textbf{0.9998}}
    & \multicolumn{1}{@{}c@{}|}{\cellcolor{gray!15}\small ---}
    & \multicolumn{1}{@{}c@{}|}{\cellcolor{gray!15}\small ---}
    & \small\textcolor{blue!70!black}{\textbf{P}} $(n^{-3.73})$
    & \small\textcolor{blue!70!black}{\textbf{0.9985}}
    & \multicolumn{1}{@{}c@{}|}{\cellcolor{gray!15}\small ---}
    & \multicolumn{1}{@{}c@{}|}{\cellcolor{gray!15}\small ---}
    \\ \hhline{~|---------|} 
    & \multicolumn{1}{@{}c@{}|}{\cellcolor{tabred!60}Min}
    & \small\textcolor{red!70!black}{E}
    & \small\textcolor{red!70!black}{0.8936}
    & \multicolumn{1}{@{}c@{}|}{\cellcolor{gray!15}\small ---}
    & \multicolumn{1}{@{}c@{}|}{\cellcolor{gray!15}\small ---}
    & \small\textcolor{red!70!black}{E}
    & \small\textcolor{red!70!black}{0.9038}
    & \multicolumn{1}{@{}c@{}|}{\cellcolor{gray!15}\small ---}
    & \multicolumn{1}{@{}c@{}|}{\cellcolor{gray!15}\small ---}
    \\ \hhline{~|---------|} 
    & \multicolumn{1}{@{}c@{}|}{\cellcolor{tabblue!60}Max}
    & \small\textcolor{blue!70!black}{\textbf{P}} $(n^{-3.66})$
    & \small\textcolor{blue!70!black}{\textbf{0.9987}}
    & \multicolumn{1}{@{}c@{}|}{\cellcolor{gray!15}\small ---}
    & \multicolumn{1}{@{}c@{}|}{\cellcolor{gray!15}\small ---}
    & \small\textcolor{blue!70!black}{\textbf{P}} $(n^{-3.41})$
    & \small\textcolor{blue!70!black}{\textbf{0.9931}}
    & \multicolumn{1}{@{}c@{}|}{\cellcolor{gray!15}\small ---}
    & \multicolumn{1}{@{}c@{}|}{\cellcolor{gray!15}\small ---}
    \\ \hhline{~|---------|} 
    & \multicolumn{1}{@{}c@{}|}{\cellcolor{tabblue!60}Max}
    & \small\textcolor{red!70!black}{E}
    & \small\textcolor{red!70!black}{0.9052}
    & \multicolumn{1}{@{}c@{}|}{\cellcolor{gray!15}\small ---}
    & \multicolumn{1}{@{}c@{}|}{\cellcolor{gray!15}\small ---}
    & \small\textcolor{red!70!black}{E}
    & \small\textcolor{red!70!black}{0.9239}
    & \multicolumn{1}{@{}c@{}|}{\cellcolor{gray!15}\small ---}
    & \multicolumn{1}{@{}c@{}|}{\cellcolor{gray!15}\small ---}
    \\ \hline
    %------------------------------------------------------------------
    \multirow{4}{*}{$\lfloor\log n \rfloor$}
    & \multicolumn{1}{@{}c@{}|}{\cellcolor{tabred!60}Min}
    & \small\textcolor{blue!70!black}{P} $(n^{-9.33})$
    & \small\textcolor{blue!70!black}{0.9420}
    & \small\textcolor{blue!70!black}{P} $(n^{-16.21})$
    & \small\textcolor{blue!70!black}{0.9424}
    & \small\textcolor{blue!70!black}{P} $(n^{-8.90})$
    & \small\textcolor{blue!70!black}{0.9488}
    & \small\textcolor{blue!70!black}{P} $(n^{-13.93})$
    & \small\textcolor{blue!70!black}{0.9431}
    \\ \hhline{~|---------|} 
    & \multicolumn{1}{@{}c@{}|}{\cellcolor{tabred!60}Min}
    & \small\textcolor{red!70!black}{\textbf{E}}
    & \small\textcolor{red!70!black}{\textbf{0.9579}}
    & \small\textcolor{red!70!black}{\textbf{E}}${}$
    & \small\textcolor{red!70!black}{\textbf{0.9472}}
    & \small\textcolor{red!70!black}{\textbf{E}}
    & \small\textcolor{red!70!black}{\textbf{0.9528}}
    & \small\textcolor{red!70!black}{\textbf{E}}
    & \small\textcolor{red!70!black}{\textbf{0.9607}}
    \\ \hhline{~|---------|}
    & \multicolumn{1}{@{}c@{}|}{\cellcolor{tabblue!60}Max}
    & \small\textcolor{blue!70!black}{P} $(n^{-8.31})$
    & \small\textcolor{blue!70!black}{0.9422}
    & \small\textcolor{blue!70!black}{\textbf{P}} $(n^{-11.07})$
    & \small\textcolor{blue!70!black}{\textbf{0.9414}}
    & \small\textcolor{blue!70!black}{P} $(n^{-7.92})$
    & \small\textcolor{blue!70!black}{0.9472}
    & \small\textcolor{blue!70!black}{P} $(n^{-8.79})$
    & \small\textcolor{blue!70!black}{0.9462}
    \\ \hhline{~|---------|} 
    & \multicolumn{1}{@{}c@{}|}{\cellcolor{tabblue!60}Max}
    & \small\textcolor{red!70!black}{\textbf{E}}
    & \small\textcolor{red!70!black}{\textbf{0.9623}}
    & \small\textcolor{red!70!black}{E}
    & \small\textcolor{red!70!black}{0.9287}
    & \small\textcolor{red!70!black}{\textbf{E}}
    & \small\textcolor{red!70!black}{\textbf{0.9579}}
    & \small\textcolor{red!70!black}{\textbf{E}}${}$
    & \small\textcolor{red!70!black}{\textbf{0.9489}}
    \\ \hline
    %------------------------------------------------------------------
    \multirow{4}{*}{$\lfloor \sqrt{n} \rfloor$}
    & \multicolumn{1}{@{}c@{}|}{\cellcolor{tabred!60}Min}
    & \small\textcolor{blue!70!black}{P} $(n^{-10.30})$
    & \small\textcolor{blue!70!black}{0.9328}
    & \small\textcolor{blue!70!black}{P} $(n^{-16.39})$
    & \small\textcolor{blue!70!black}{0.9445}
    & \small\textcolor{blue!70!black}{P} $(n^{-9.81})$
    & \small\textcolor{blue!70!black}{0.9401}
    & \small\textcolor{blue!70!black}{P} $(n^{-14.39})$
    & \small\textcolor{blue!70!black}{0.9460}
    \\ \hhline{~|---------|}
    & \multicolumn{1}{@{}c@{}|}{\cellcolor{tabred!60}Min}
    & \small\textcolor{red!70!black}{\textbf{E}}
    & \small\textcolor{red!70!black}{\textbf{0.9749}}
    & \small\textcolor{red!70!black}{\textbf{E}}
    & \small\textcolor{red!70!black}{\textbf{0.9541}}
    & \small\textcolor{red!70!black}{\textbf{E}}
    & \small\textcolor{red!70!black}{\textbf{0.9711}}
    & \small\textcolor{red!70!black}{\textbf{E}}
    & \small\textcolor{red!70!black}{\textbf{0.9738}}
    \\ \hhline{~|---------|} 
    & \multicolumn{1}{@{}c@{}|}{\cellcolor{tabblue!60}Max}
    & \small\textcolor{blue!70!black}{P} $(n^{-9.12})$
    & \small\textcolor{blue!70!black}{0.9329}
    & \small\textcolor{blue!70!black}{\textbf{P}} $(n^{-11.28})$
    & \small\textcolor{blue!70!black}{\textbf{0.9466}}
    & \small\textcolor{blue!70!black}{P} $(n^{-8.69})$
    & \small\textcolor{blue!70!black}{0.9382}
    & \small\textcolor{blue!70!black}{P} $(n^{-9.24})$
    & \small\textcolor{blue!70!black}{0.9505}
    \\ \hhline{~|---------|}  
    & \multicolumn{1}{@{}c@{}|}{\cellcolor{tabblue!60}Max}
    & \small\textcolor{red!70!black}{\textbf{E}}
    & \small\textcolor{red!70!black}{\textbf{0.9772}}
    & \small\textcolor{red!70!black}{E}
    & \small\textcolor{red!70!black}{0.9410}
    & \small\textcolor{red!70!black}{\textbf{E}}
    & \small\textcolor{red!70!black}{\textbf{0.9741}}
    & \small\textcolor{red!70!black}{\textbf{E}}
    & \small\textcolor{red!70!black}{\textbf{0.9687}}
    \\ \hline
\end{tabular}
\caption{\justifying Decay regime of the normalised second moment $\mathbb{E}[f_U^2]$ for $n\in\{2,\ldots,35\}$ ($m=2n$ modes). Each cell gives two rows: power-law fit (\textcolor{blue!70!black}{P}, exponent $n^{\hat\alpha}$, log--log $R^2$) and exponential fit (\textcolor{red!70!black}{E}, semi-log $R^2$). The dominant regime (higher $R^2$) is \textbf{bold}. Min: $\ket{1,\ldots,1}$; Max: $\ket{n,0,\ldots,0}$. Cells --- indicate $p=\Theta(1)$ where a varying $k_0$ is not applicable.}
\label{tbl:Decay_Product_Number_Operator}
\end{table*}

To characterise the decay of a quantity of interest with the system size $n$, we fit the numerical data with two competing models: a power-law decay $y \sim n^\alpha$ and an exponential decay $y \sim e^{-\beta n}$. Both fits are performed by ordinary least-squares linear regression in logarithmic space so that the exponents of the coefficients of the polynomial and exponential decay models are extracted as the slope of a straight-line fit. The quality of each fit is assessed by the coefficient of determination $R^2$, computed in the same logarithmic space as the regression: $R^2 = 1 - \mathrm{S}_\mathrm{res}/\mathrm{S}_\mathrm{tot}$, where $\mathrm{S}_\mathrm{res}$ is the residual sum of squares and $\mathrm{S}_\mathrm{tot}$ is the total variance of $\log y$ around its mean. An $R^2$ close to $1$ indicates that the chosen model captures the observed decay almost perfectly on the plotted range, while a low $R^2$ signals a poor fit. Comparing the two $R^2$ values provides a data-driven criterion for distinguishing polynomial from exponential scaling: whichever model achieves the higher $R^2$ is retained as the better description of the finite-$n$ behaviour.

In \cref{tbl:Decay_Product_Number_Operator}, we only the present the scaling for $k_0=1$ and $k_0=\lfloor \log_2n \rfloor$. Considering $k_0 = \lfloor \sqrt{n} \rfloor$ would only be relevant for the case $p \geq \lfloor \sqrt{n} \rfloor$. However, we can see that the second moment is already vanishing exponentially when considering $k_0=\lfloor \log_2n \rfloor$.

\begin{figure}[h!]
    \centering
    \includegraphics[width=1.0\linewidth]{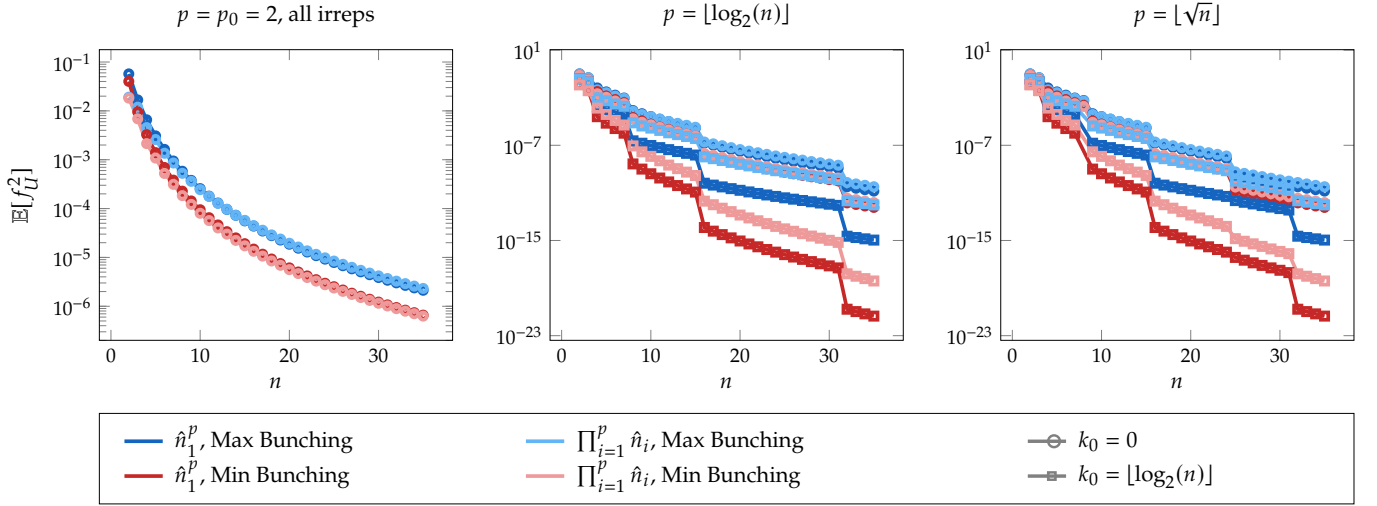}
    \caption{\justifying Study of the second moment, for different evolution of the number of terms in the product, corresponding to \cref{tbl:Decay_Product_Number_Operator}.}
    \label{fig:Product_Photon_Number_Second_Moments}
\end{figure}

\subsection{Cause of Barren Plateaus}\label{subsec:Cause_BP_Product_Number}

We notice that, in the particular case of a product of photon number operators, the presence of Barren Plateaus when $p = \Omega (\log_2 n)$ is due to the fact that the purity distribution of the input state $\{ P_k^{(n)}(\rho) \}_{k=0}^n$ is not aligned with the purity distribution of the observable $\{ P_k^{(n)}(O) \}_{k=0}^n$. This difference mainly comes from the fact that Fock states have mainly support on the last irreps (see \cref{app:Initial_State_Purities}) while a product of $p \leq n$ photon number operators only have support on the first $p+1$ irreps. In the case where the operators are applied on a large number of different modes, a necessary condition to avoid \cref{simtech:direct_expansion}, the 2-norm can decrease and thus all the purities can also decrease. In the cases chosen here $p=\lfloor \log_2 n \rfloor$ and $p = \lfloor \sqrt{n} \rfloor$, the values of:
\begin{equation}
    \sum_{k=k_0}^n \frac{\|P_k^{(n)}(O)\|_2^2}{d_{k}^{(n)}} \,
\end{equation}
increases and that, even when considering only large irreps ($k_0 = \log_2 n$) as illustrated in \cref{fig:Product_Photon_Numbers_Cause_BP}. 

\begin{figure}[h!]
    \centering
    \includegraphics[width=0.99\linewidth]{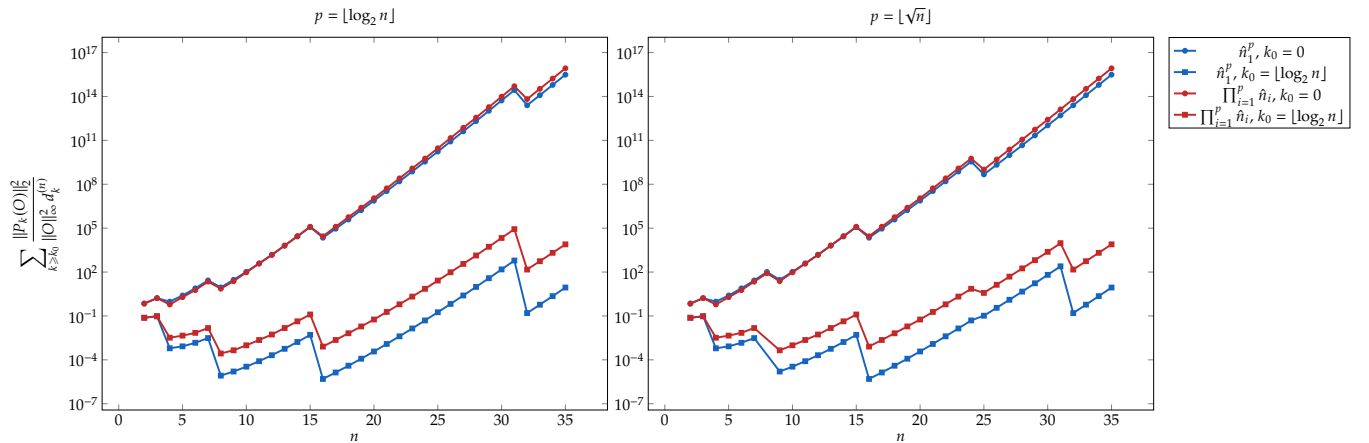}
    \caption{\justifying Evolution of the sum of the observable purities divided by the dimension of the irrep, from irrep $k_0$ to irrep $n$.}
    \label{fig:Product_Photon_Numbers_Cause_BP}
\end{figure}

\newpage
% ─────────────────────────────────────────────
\section{Observable Study - Single Support Observable}\label{app:Single_Support_Obs}
% ─────────────────────────────────────────────

    \subsection{Definition and properties:}

Based on \cref{propop:monomial_single_irrep_support}, we now that a monomial $[{M}_{\bm p, \bm q}]_d^{(n)}$ of degree $d\leq n$ acting on $n$ particles such that $p_i \neq p_j \;,\forall i,j$ has full support on the highest irrep $\lambda_d^{(n)}$. We propose the study of the following observable that respect such condition and only has support on the irrep $p$:
\begin{equation}
    \prod_{k=1}^p \frac{\hat{a}_{i_k}^\dagger \hat{a}_{j_k} + \hat{a}_{j_k}^\dagger \hat{a}_{i_k}}{\sqrt{2}} \, .
\end{equation}

Such observable can be constructed using Beam Splitters between every $p$ pairs of modes and by applying a photon number measurement on all of them. Because the observable only has support on the irrep $k$, we do not have to compute its irrep decomposition, and we can simply consider its two norm:
\begin{equation}
    Tr \left[ O O^\dagger \right] = \frac{1}{2^p} Tr \left[ \prod_{k=1}^p 2 \hat{n}_{i_k} \hat{n}_{j_k} + \hat{n}_{i_k} + \hat{n}_{j_k}  \right] \,.
\end{equation}
As a result, we have (considering that every indices in the expression of the observable operator are different):
\begin{equation}
    Tr \left[ O O^\dagger \right] = \frac{1}{2^p} \sum_{q=0}^p Tr[\hat{n}_{i_1} \dots \hat{n}_{i_{p+q}}] = \frac{1}{2^p}  \sum_{q=0}^p \binom{p}{q} \frac{(m+n-1)!}{(n-p-q)!(m+p+q-1)!}
\end{equation}

To ensure that such choice of observable does not affect shot noise, we normalize this observable by the square of its infinite norm $||O||^2_\infty$. For the observable $\prod_{k=1}^K \hat n_{l_k}$, where the indices $(l_1,\dots,l_K)$ are all distinct, the maximal eigenvalue on the fixed $n$-photon subspace is $ \lambda_{\max}(N)=(c+1)^r c^{K-r}$, where the integers $c$ and $r$ are defined by the Euclidean division of $n$ by $K$: 
$$n=Kc+r,\qquad 0\le r<K.$$

Where $(c+1)^r c^{K-r}$ is the maximal product obtained by distributing the $n$ photons as evenly as possible among the $K$ selected modes, namely: $r$ modes contain $c+1$ photons, and $K-r$ modes contain $c$ photons. Equivalently, the optimal occupation numbers of the $K$ selected modes are $$(c+1,\dots,c+1,c,\dots,c),$$ with exactly $r$ occurrences of $c+1$ and $K-r$ occurrences of $c$.

    \subsection{Classical Simulation}

The single support observable proposed here is not a simple monomial of degree $p$, and its decomposition into a sum of monomial would require to consider a large number of terms. This prevent the use of \cref{simtech:direct_expansion} for $p = \Omega(\log n)$. For constant degree $p$, the observable only has support on the irrep $\lambda_p^{(n)}$ of polynomial size, allowing the use of \cref{simtech:exact_gsim}. 

\newpage

\begin{figure}[h!]
    \centering
    \includegraphics[width=0.60\linewidth]{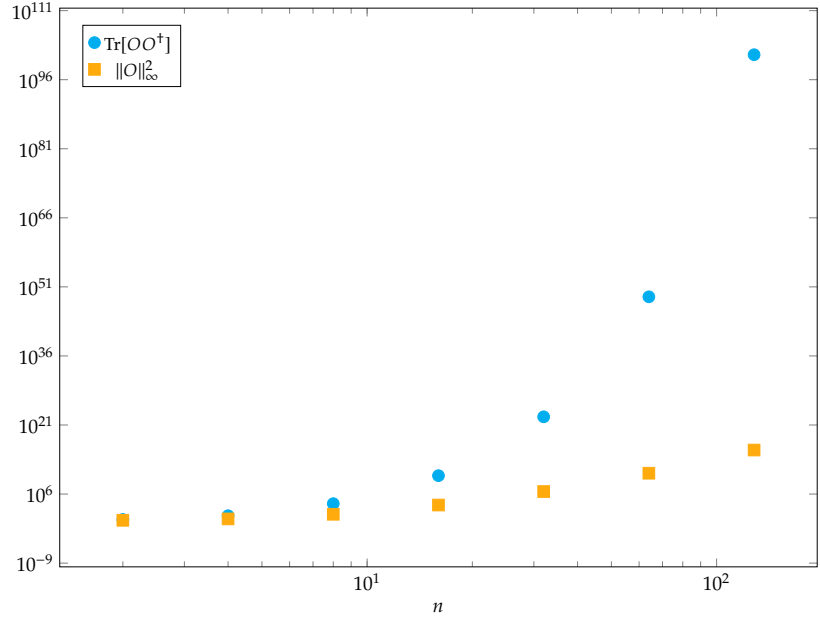}
    \caption{Evolution of the $2$-norm and infinite norm of the single irrep observable of degree $\log_2 n$, with $n$ the number of particles and $m=2n$.}
    \label{fig:Single_Irrep_Purities}
\end{figure}

    \subsection{Second Moment Studies:}

To study the second moment, we propose to verify numerically the decay of the second moment given in \cref{eq:sec_m_exp}, and recalled here:
\begin{equation*}
    \e[U \sim U(m)]{f_U(\rho,O)^2} =  \sum_{k=0}^n \frac{\|P_k^{(n)}(\rho)\|_2^2 \|P_k^{(n)}(O)\|_2^2}{d_{k}^{(n)}} \,.
\end{equation*}

\begin{figure}[b!]
    \centering
    \includegraphics[width=1.0\linewidth]{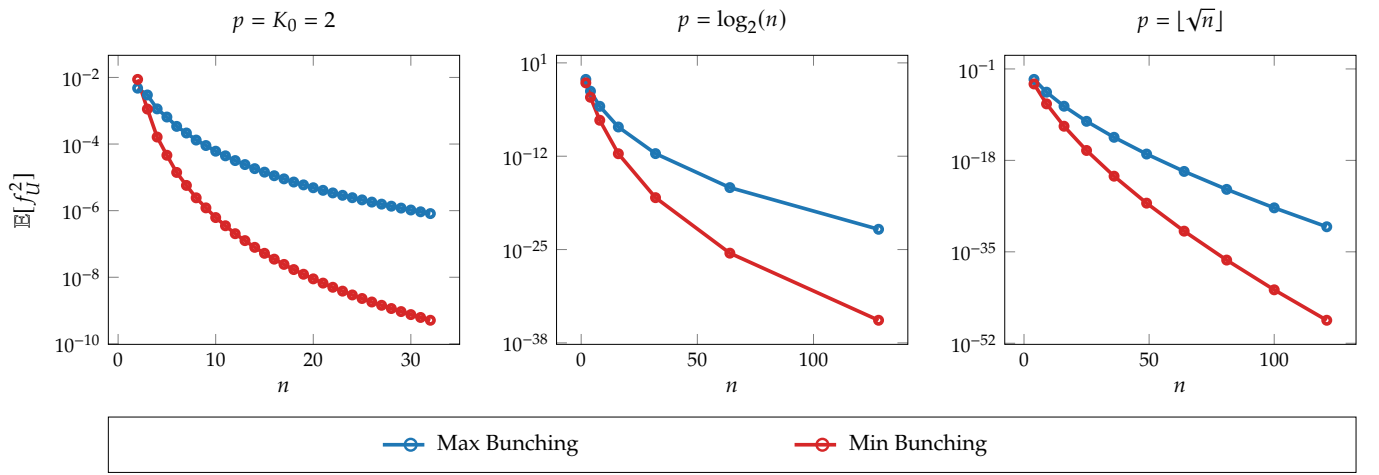}
    \caption{\justifying Study of the second moment, for different evolution of the degree, corresponding to \cref{tbl:Decay_Single_Irrep_Observable}.}
    \label{fig:Single_Irrep_Observable_Second_Moments}
\end{figure}

In the particular case of the single support observable, we don't have to compute the purity distribution of the observable, and the second moment can be written as:
\begin{equation}
    \e[U \sim U(m)]{f_U(\rho,O)^2} =  \frac{\|P_p^{(n)}(\rho)\|_2^2 ||O||^2_2}{d_{p}^{(n)}} \,.
\end{equation}

As a result, we can do larger simulation, and we don't have to consider a decomposition of the second moment to understand if Barren Plateaus can be avoided simply due to the contribution of polynomially large irreps. We compute the evolution of the second moment, and we provide the behaviour
diagnosis through the regression match of each curve, with corresponding R2 value:

\begin{table*}[h!]
\centering
\begin{tabular}{|c|c|c|c|c|c|c|}
    \hline
    \multirow{3}{*}{} 
    & \multicolumn{6}{c|}{$\prod_{k=1}^p (\hat{a}_{i_k}^\dagger \hat{a}_{j_k} + \hat{a}_{j_k}^\dagger \hat{a}_{i_k}) / \sqrt{2} $} 
    \\ \hhline{|~------|}
     % --- Sub-header row: cols 1-2 empty (covered by multirow), cols 3-14 alternate ---
    Input
    & \multicolumn{2}{c|}{$p = 2$} 
    & \multicolumn{2}{c|}{$p = \log_2 n$}
    & \multicolumn{2}{c|}{$p = \sqrt{n}$} 
    \\ \hhline{|~------|}
    
    % --- Sub-header row: cols 1-2 empty (covered by multirow), cols 3-14 alternate ---
    Bunching
    & Decay & $R^2$
    & Decay & $R^2$
    & Decay & $R^2$
    \\ \hline
    %------------------------------------------------------------------
    \multicolumn{1}{@{}c@{}|}{\cellcolor{tabred!60}Min}
    & \small\textcolor{blue!70!black}{\textbf{P}} $(n^{-6.90})$
    & \small\textcolor{blue!70!black}{\textbf{0.9998}}
    & \small\textcolor{blue!70!black}{\textbf{P}} $(n^{-11.40})$
    & \small\textcolor{blue!70!black}{\textbf{0.9558}}
    & \small\textcolor{blue!70!black}{P} $(n^{-29.66})$
    & \small\textcolor{blue!70!black}{0.9824}
    \\ \hline
    \multicolumn{1}{@{}c@{}|}{\cellcolor{tabred!60}Min}
    & \small\textcolor{red!70!black}{E}
    & \small\textcolor{red!70!black}{0.8929}
    & \small\textcolor{red!70!black}{E}
    & \small\textcolor{red!70!black}{0.9198}
    & \small\textcolor{red!70!black}{\textbf{E}}
    & \small\textcolor{red!70!black}{\textbf{0.9824}}
    \\ \hline
    \multicolumn{1}{@{}c@{}|}{\cellcolor{tabblue!60}Max}
    & \small\textcolor{blue!70!black}{\textbf{P}} $(n^{-3.38})$
    & \small\textcolor{blue!70!black}{\textbf{0.9926}}
    & \small\textcolor{blue!70!black}{\textbf{P}} $(n^{-11.40})$
    & \small\textcolor{blue!70!black}{\textbf{0.9558}}
    & \small\textcolor{blue!70!black}{P} $(n^{-18.48})$
    & \small\textcolor{blue!70!black}{0.9817}
    \\ \hline
    \multicolumn{1}{@{}c@{}|}{\cellcolor{tabblue!60}Max}
    & \small\textcolor{red!70!black}{E}
    & \small\textcolor{red!70!black}{0.9288}
    & \small\textcolor{red!70!black}{E}
    & \small\textcolor{red!70!black}{0.9135}
    & \small\textcolor{red!70!black}{\textbf{E}}
    & \small\textcolor{red!70!black}{\textbf{0.9824}}
    \\ \hline
\end{tabular}
\caption{\justifying Decay regime of the normalised second moment $\mathbb{E}[f_U^2]$ for $n\in\{2,\ldots,35\}$ for the case $p=2$, and $n \leq 130$ with $\log_2 n \in \mathbb{N}$ for $p= \log_2 n$, and  $n \leq 130$ with $\sqrt{n} \in \mathbb{N}$ for $p= \sqrt{n} $ ($m=2n$ modes). The dominant regime (higher $R^2$) is \textbf{bold}. Min: $\ket{1,\ldots,1}$; Max: $\ket{n,0,\ldots,0}$.}
\label{tbl:Decay_Single_Irrep_Observable}
\end{table*}

To characterise the decay of a quantity of interest with the system size $n$, we fit the numerical data with two competing models: a power-law decay $y \sim n^\alpha$ and an exponential decay $y \sim e^{-\beta n}$. Both fits are performed by ordinary least-squares linear regression in logarithmic space so that the exponents of the coefficients of the polynomial and exponential decay models are extracted as the slope of a straight-line fit. The quality of each fit is assessed by the coefficient of determination $R^2$, computed in the same logarithmic space as the regression: $R^2 = 1 - \mathrm{S}_\mathrm{res}/\mathrm{S}_\mathrm{tot}$, where $\mathrm{S}_\mathrm{res}$ is the residual sum of squares and $\mathrm{S}_\mathrm{tot}$ is the total variance of $\log y$ around its mean. An $R^2$ close to $1$ indicates that the chosen model captures the observed decay almost perfectly on the plotted range, while a low $R^2$ signals a poor fit. Comparing the two $R^2$ values provides a data-driven criterion for distinguishing polynomial from exponential scaling: whichever model achieves the higher $R^2$ is retained as the better description of the finite-$n$ behaviour.

    \subsection{Cause of Barren Plateaus}\label{subsec:Cause_BP_Single_Support_Obs}

We notice that, as for the case of number operator products in \cref{app:Product_Number_Operators}, the presence of Barren Plateaus when $p = \Omega (\log_2 n)$ is due to the fact that the purity distribution of the input state $\{ P_k^{(n)}(\rho) \}_{k=0}^n$ is vanishing exponentially and is too small with respect to the $2$-norm of the observable . This difference mainly comes from the fact that Fock states have mainly support on the last irreps (see \cref{app:Initial_State_Purities}) while the single irrep observable only has support on the irrep $p$. As explained previously, this observable can be made with $p$ beamsplitters and $2p$ number operators applied on different modes. As a result, increasing too much $p$ will decrease the $2$-norm of the observable, and even if the purities of the input state and observable could match, the later would not be able to compensate the irrep dimension.

Notice that, studying the value of:
\begin{equation}\label{eq:Norm2_dim_Single_Irrep_Obs}
    \sum_{k=0}^n \frac{\|P_k^{(n)}(O)\|_2^2}{d_{k}^{(n)}} = \frac{||O||^2_2}{d_{p}^{(n)}} \, .
\end{equation}
This quantity can increase even when considering an evolution of $p$ with $n$ that makes the irrep dimension exponential with respect to the system size. In \cref{fig:Single_Irrep_Observable_Cause_BP}, we give several examples of linear evolution of $p$ with respect to $n$. By studying the evolution of \cref{eq:Norm2_dim_Single_Irrep_Obs} for different choices of $\alpha$ with $p=n/\alpha$, we notice an inflexion point showing that the $2$-norm of the normalised observable can only compensate the irrep dimension for small values of $\alpha$. Even if this corresponds to cases where the irreps are exponential with respect to the system size, it does not match the irreps on which Fock states are concentrated.

\begin{figure}
    \centering
    \includegraphics[width=0.5\linewidth]{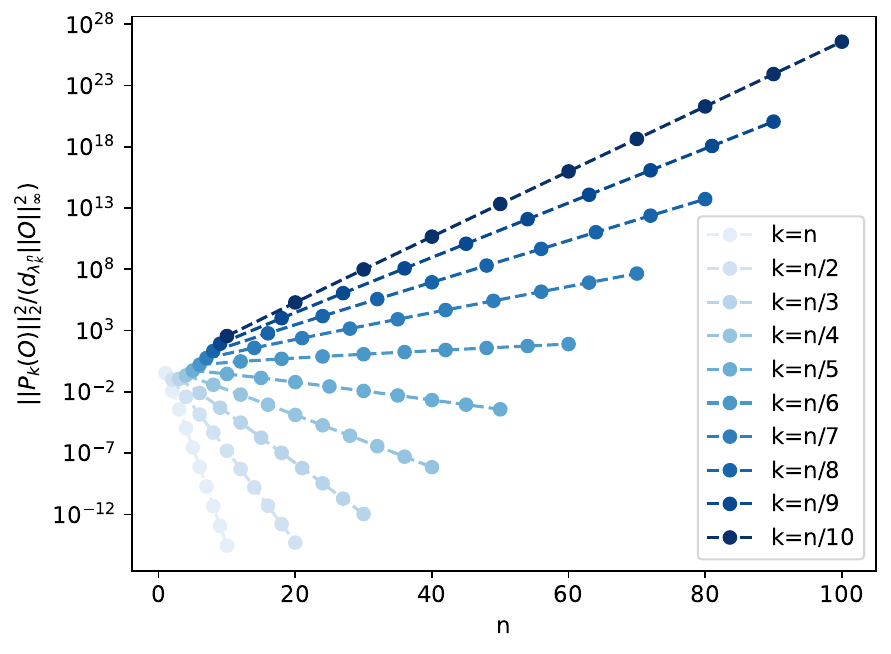}
    \caption{\justifying Evolution of \cref{eq:Norm2_dim_Single_Irrep_Obs} when considering linear evolutions of $p$ with $n$.}
    \label{fig:Single_Irrep_Observable_Cause_BP}
\end{figure}
\newpage
% ─────────────────────────────────────────────
\section{Observable Study - Parity and Number-phase Operators}\label{app:Parity_Operators}
% ─────────────────────────────────────────────

    \subsection{Definition and properties:}

In this Section, we consider a very general set of parity measurement. We call $g_p(\{\hat{n}_i \}_{i \in I})$, a polynomial function of number operators of degree $p \leq n$, with $n$ the number of particles. In this Section, we focus on observables of the form:
\begin{equation}\label{eq:Poly_Parity_Operator}
    O = (-1)^{g_d(\{\hat{n}_i \}_{i \in I})} = \sum_{R \in \Phi_m^n} (-1)^{g_d(\{R_i \}_{i \in I})} \ket{R}\bra{R}\;.
\end{equation}

All eigenvalues are bounded by $1$, so we have:
\begin{equation}
    \|O||_\infty \leq 1 \quad \text{and} \quad \|O\|_2^2 \leq |\Phi_m^n| \;.
\end{equation}

In this work we consider two cases. The first one:
\begin{equation}\label{eq:Parity_Measurement_Observable}
    O = \prod_{i \in I} (-1)^{\hat{n}_i} \,,
\end{equation}
corresponds to the parity measurement observable used in bosonic generative modelling \cite{kurkin_universality_2026, kolarovszki_generative_2026, gottlieb_efficient_2026}. The eigenvalues are exactly $\pm 1$, and the $2$-norm is given by $|\Phi_m^n|$.

The second case we consider, in order to avoid de-quantization techniques such as presented in the later part of this appendix, that we call \emph{number-phase observable}, is given by:
\begin{equation}
    O = \prod_{i \in I} (-1)^{\alpha_i \hat{n}_i^2} \,,
\end{equation}
with $\alpha_i = 1/r_i$, and $r_i$ the $i^\text{th}$ prime number. This observable physically corresponds a layer of a Kerr gates, and can be achieved by a post-processing of photon- number measurement outcomes applied on the subset of modes $I$. 

    \subsection{Irrep Decomposition:}

Applying the change of variable $S = R - b$ ($|b|=k$, $b \leq R$), the $k$-fold  application of the lowering map $\Cm$, introduced in \cref{eq:C_map}, on an operator of the form of \cref{eq:Poly_Parity_Operator} is given by:
\begin{equation}
    \Cm^k(O) = \sum_{S \in \Phi_{m}^{n-k}} \lambda_S \ket{S}\bra{S},\qquad
  \lambda_S = \sum_{\substack{b \in \mathbb{Z}_{\geq 0}^m\\|b|=k}} (-1)^{g_d(\{S_i + b_i \}_{i \in I})}\,\frac{k!}{\prod_l b_l!}\prod_l\frac{(S_l+b_l)!}{S_l!}.
\end{equation}

Since $\Cm^k(O)$ is diagonal: $\mathrm{Tr}[(\Cm^k(O))^2] = \sum_{S \in \Phi_{m}^{n-k}} |\lambda_S|^2$.

\begin{figure}
    \centering
    \includegraphics[width=0.99\linewidth]{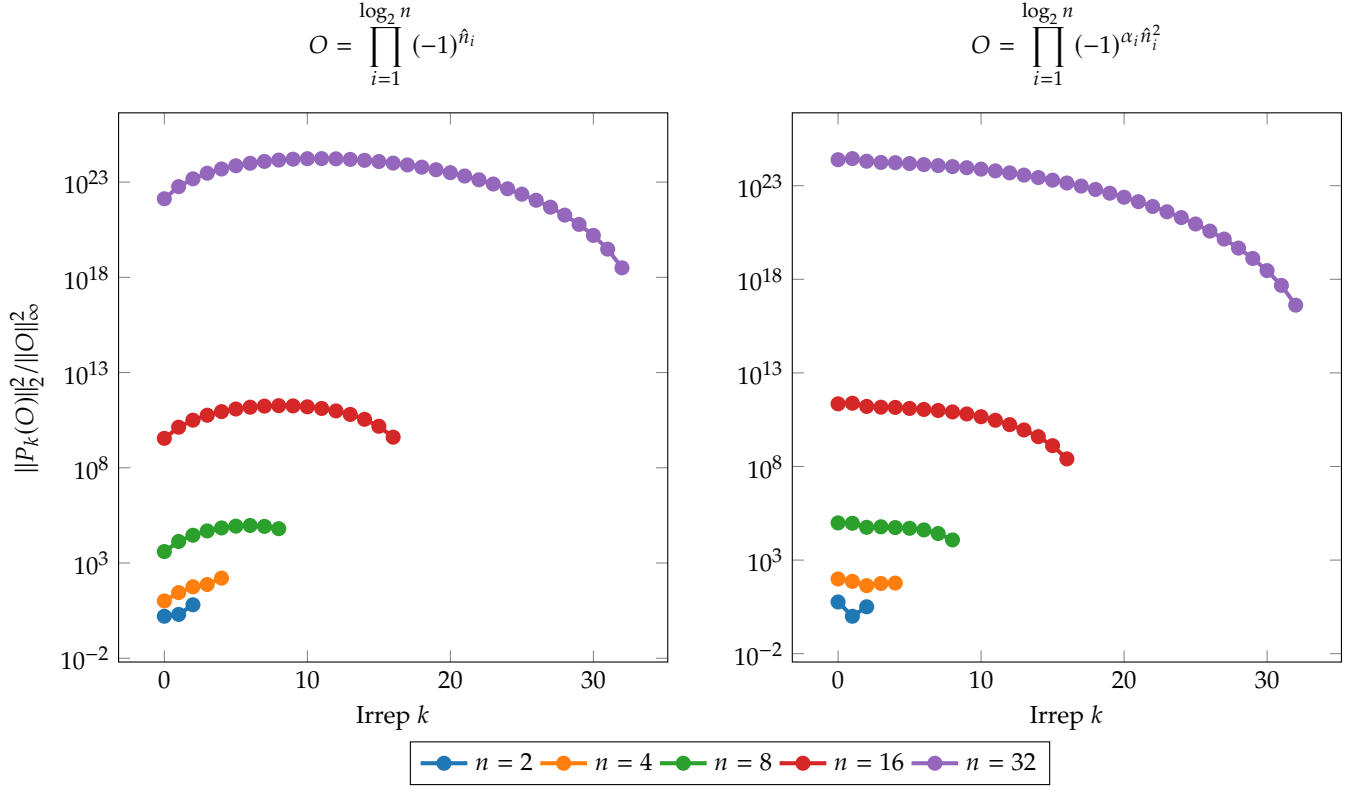}
    \caption{Irrep purities profile of linear and quadratic number-phase observables normalized by their infinity norm, for different values of photons $n=2,4,8,16,32$ in $m=2n$ modes.}
    \label{fig:Parity_Operators_Purities}
\end{figure}

    \subsection{Classical Simulation}\label{subsec:Classical_Sim_Parity_Op}

        \subsubsection*{Kerr gate decomposition}

The particular case of parity measurement observable described by \cref{eq:Parity_Measurement_Observable} can be efficiently simulated classically. As highlighted in \cite{kurkin_universality_2026}, the expectation value physically corresponds to that of a layer phase-shifter of the form $\bra{\phi}\varphi_{n}(U)^\dagger \, \prod_{j=1}^p e^{\phi_j \hat{n}_j} \, \varphi_{n}(U)\ket{\psi}$, with $\varphi_{n}(U)^\dagger \, \prod_{j=1}^p e^{\phi_j \hat{n}_j} \, \varphi_{n}(U)$ a unitary matrix. Such quantity can be approximated by \cref{simtech:Expected_value_LO}.

In general, for any polynomial $g_d(\{\hat{n}_i \}_{i \in I})$, we can decompose the corresponding parity operator given by \cref{eq:Poly_Parity_Operator} using a discrete Fourier transform. We extent the decomposition used in \cite{upreti_exponentially-improved_2026, jin-wei_superposition_1996}: 

\begin{lemma}[Inspired by \cite{upreti_exponentially-improved_2026}]\label{lem:Decomposition_Phase_Shifter}
   Consider a product of $p$ number operator $M_p = \prod_{i=1}^p \hat n_i^{s_i}$ of degree $p\leq n$. We have that:
    \begin{equation}\label{eq:poly_phase_decomp}
        \exp\!\left( i\pi\,\frac{a}{b} \prod_{i=1}^{p} \hat{n}_i^{\,s_i} \right)
        \;=\; \sum_{j_1,\dots,j_p=0}^{B-1} g_{j_1\cdots j_p}\,
              \prod_{i=1}^{p} \exp\!\left( i\,\frac{2\pi j_i}{B}\,\hat{n}_i \right),
    \end{equation}
    with coefficients given by the $p$-dimensional inverse discrete Fourier transform:
    \begin{equation}\label{eq:poly_phase_coeffs}
        g_{j_1\cdots j_p}
        \;=\; \frac{1}{B^{\,p}} \sum_{n_1,\dots,n_p=0}^{B-1}
              \exp\!\left( i\pi\,\frac{a}{b} \prod_{i=1}^{p} n_i^{\,s_i} \right)
              \prod_{i=1}^{p} \exp\!\left( -i\,\frac{2\pi j_i}{B}\,n_i \right),
    \end{equation}
    where $B=b$ when $\tfrac{a}{b}\prod_i n_i^{s_i}$ is invariant modulo $2$ under every shift $n_i\mapsto n_i+b$, and $B=2b$ otherwise. 
\end{lemma}

More specifically, in the numerical simulations we consider the observable constructed from the unitary operator 
corresponding to a product of quadratic number-phase shifts,
\begin{equation}\label{eq:quadratic_obs}
    U_{\boldsymbol{r}}
    =
    \prod_{i=1}^{p}
    \exp\!\left(
        i\frac{\pi}{r_i}\hat n_i^2
    \right),
\end{equation}
where \(r_i > 2\) is a prime number.

By adapting the results from \cref{lem:Decomposition_Phase_Shifter} 
to this this operator, we get

\begin{equation}\label{eq:unitary_fourier_decomp}
   U_{\bm r}
    =
    \sum_{\bm q\in\mathcal Q_1\times\cdots\times\mathcal Q_p}
    \left( \prod_{i=1}^{p}c_{i,q_i} \right)\,
    \exp\!\left(
         i \pi
        \sum_{i=1}^{p}
        \frac{q_i}{r_i}\hat n_i
    \right),
\end{equation}
where
\begin{equation}
     \mathcal Q_i = \{1,3,\ldots,2r_i-1\},
 \end{equation}
 and
\begin{equation} \label{eq:prime_product_fourier_coefficients}
    c_{i,q}
    =
    \frac{1}{2r_i}
    \sum_{s=0}^{B_i-1}
    \exp\!\left(     \frac{i\pi}{r_i}s^2 -\frac{ i \pi }{r_i}qs \right).
\end{equation}
The decomposition in \cref{eq:unitary_fourier_decomp} contains exactly
\begin{equation}
    N_{\mathrm F}=  \prod_{i=1}^{p}r_i = \exp\!\left[(1+o(1))p\log p\right],
\end{equation}
non zero Fourier coefficients, and their\(\ell_1\)-norm is
\begin{equation}\label{eq:l1_norm}
    \sum_{\bm q}|c_{\bm q}|    =\sqrt{\prod_{i=1}^{p}r_i}   = \sqrt{N_{\mathrm F}}.
\end{equation}

While each linear phase shifter can be approximation using \cref{simtech:Transition_Approx_LO}, the number of terms in the expansion scales exponentially in $n$ for $p$ growing with $n$. Instead of computing each term in the Fourier expansion individually, one may consider importance sampling. However, one can easily show that, in this particular case, the number of samples needed will be lower bounded by the \(\ell_1\)-norm of the Fourier coefficient vector, which also grows exponentially, as derived in \cref{eq:l1_norm}. In addition, since the Fourier coefficients norm is the same, truncation methods based on Fourier decay are also ruled out.

        \subsubsection*{Max bunching case}

The maximally bunched input $\rho_n = |n,0,\dots,0\rangle\!\langle n,0,\dots,0|$ deserves a separate comment, because the expectation value of the quadratic number-phase observable can be estimated classically to any inverse-polynomial additive accuracy in polynomial time, for a worst-case interferometer and without any large-irrep truncation. The reason is that a passive linear-optical interferometer $V\in U(m)$ acts on this input as
\begin{equation}
  \varphi_n(V)\,|n,0,\dots,0\rangle
  = \frac{1}{\sqrt{n!}}\Big(\sum_{j=1}^{m}V_{j1}a_j^{\dagger}\Big)^{\!n}|0\rangle ,
\end{equation}
that is, as a single power of a single linear combination of creation operators. Expanding the power multinomially yields exactly one term per measurement outcome, so no interference between distinct many-body paths survives and no permanent appears. The photon-number distribution follows exactly the one of the multinomial distribution $S\sim\mathrm{Mult}(n;q_1,\dots,q_m)$, the distribution of counts from $n$ independent trials, where each trial lands in one of $m$ categories with probabilities $q_1, \dots, q_m$, and
with $q_j=|V_{j1}|^{2}$, which is to say the $n$ photons populate the output modes independently. \\

Because the observable is diagonal in the Fock basis and bounded, its expectation value is an ordinary classical expectation over this distribution, and can be obtained in two ways. The first is sampling: a multinomial distribution can be sampled from classically efficiently, its marginal on the $p$ measured modes is again multinomial once the remaining modes are lumped into a single category, and the observable takes values in $[-1,1]$, so Hoeffding's inequality gives additive accuracy $\epsilon$ with probability $1-\delta$ with $O(\epsilon^{-2}\log(1/\delta))$ samples, independently of $n$, $m$ and $p$. The second is deterministic: writing the observable as $h(s)=\mathrm{Re}\prod_{i=1}^{p}e^{\mathrm{i}\pi s_i^{2}/r_i}$ exhibits it as a product over modes, so its expectation can be accumulated mode by mode by a dynamic program whose running state is merely the number of photons placed so far, at a cost of $O(pn^{2})$ arithmetic operations. \\

The second route that generalizes for states with fixed number of occupied modes. For any input with $k$ occupied modes the same mode-by-mode accumulation still applies, provided the running state records how many photons each occupied input mode has emitted so far, separately for the amplitude and for its conjugate; its size is then at most $\prod_{i=1}^{k}(n_i+1)^{2}$, with $n_i$ the input occupations. This is polynomial in whenever $k$ is constant. It is also where the argument stops: the minimally bunched input $|1,\dots,1,0,\dots,0\rangle$ has indeed $k=n$ occupied modes, so keeping track of all coefficients requires $O(2^{n})$ space, and its output amplitudes are proportional to matrix permanents. A polynomially decaying high-irrep tail is therefore necessary but not sufficient for classical intractability, and inputs with a constant number of occupied modes cannot populate the regime in which barren plateaus are absent and no efficient classical simulation is known.

    \subsection{Second Moment Studies:}

To study the second moment, we propose to verify numerically the decay of the second moment given in \cref{eq:sec_m_exp}, and recalled here:
\begin{equation*}
    \e[U \sim U(m)]{f_U(\rho,O)^2} =  \sum_{k=0}^n \frac{\|P_k^{(n)}(\rho)\|_2^2 \|P_k^{(n)}(O)\|_2^2}{d_{k}^{(n)}} \,.
\end{equation*}

In order to make sure that a polynomial decay is not only due to the first irreps, which would allow to apply \cref{simtech:approx_gsim}, we propose to also study separately the higher terms of the second moment defined in \cref{eq:irrep_simulation_tail}:
\begin{equation}
    \Delta_{k_0}(\rho, O) = \sum_{k=k_0}^n \frac{\|P_k^{(n)}(\rho)\|_2^2 \|P_k^{(n)}(O)\|_2^2}{d_{k}^{(n)}} \,,
\end{equation}
for different scaling of $k_0$, with $k_0=1$ corresponding to the variance of the model. In order to make the observable efficient to estimate on a quantum processor, we normalize it by the square of its infinity norm. We compute the evolution $\Delta_{k_0}(\rho, O)$ for varying number of photons and plot the results in \cref{fig:Second_Moments_Sum_Parity_Operators}, and we provide the behaviour diagnosis through the regression match of each curve, with corresponding R2 value in \cref{tbl:Decay_Parity_Operators}.

\newpage

\begin{table*}[h!]
\centering
\begin{tabular}{|c|c|c|c|c|c|c|c|c|c|c|c|c|c|}
    \hline
    \multirow{3}{*}{\diagbox[width=1.6cm, height=1.4cm]{\small$p$}{\small$O$}}
    & \multirow{3}{*}{} 
    & \multicolumn{6}{c|}{$\displaystyle\prod_{i=1}^p (-1)^{\hat{n}_i}$} 
    & \multicolumn{6}{c|}{$\displaystyle\prod_{i=1}^p (-1)^{\alpha_i \hat{n}_i^2}$} 
    \\ \hhline{~|~------------|}

    & Input
    & \multicolumn{2}{c|}{$k_0=1$} 
    & \multicolumn{2}{c|}{$k_0=\lfloor\log_2 n \rfloor$}
    & \multicolumn{2}{c|}{$k_0=\lfloor \sqrt{n} \rfloor$}
    & \multicolumn{2}{c|}{$k_0=1$} 
    & \multicolumn{2}{c|}{$k_0=\lfloor\log_2 n \rfloor$}
    & \multicolumn{2}{c|}{$k_0=\lfloor \sqrt{n} \rfloor$}
    \\ \hhline{~|~------------|}
    
    & Bunching
    & Decay & $R^2$
    & Decay & $R^2$
    & Decay & $R^2$
    & Decay & $R^2$
    & Decay & $R^2$
    & Decay & $R^2$
    \\ \hline
    %------------------------------------------------------------------
    \multirow{4}{*}{$2$}
    & \multicolumn{1}{@{}c@{}|}{\cellcolor{tabred!60}Min}
    & \small\textcolor{blue!70!black}{\textbf{P}} $(n^{-1.49})$
    & \small\textcolor{blue!70!black}{$\textbf{0.9359}$}
    & \small\textcolor{blue!70!black}{\textbf{P}} $(n^{-7.23})$
    & \small\textcolor{blue!70!black}{$\textbf{0.9557}$}
    & \small\textcolor{blue!70!black}{P} $(n^{-8.01})$
    & \small\textcolor{blue!70!black}{$0.9486$}
    & \small\textcolor{blue!70!black}{\textbf{P}} $(n^{-1.33})$
    & \small\textcolor{blue!70!black}{$\textbf{0.9735}$}
    & \small\textcolor{blue!70!black}{\textbf{P}} $(n^{-7.28})$
    & \small\textcolor{blue!70!black}{$\textbf{0.9537}$}
    & \small\textcolor{blue!70!black}{P} $(n^{-8.12})$
    & \small\textcolor{blue!70!black}{$0.9468$}
    \\ \hhline{|~|-------------|} 
    & \multicolumn{1}{@{}c@{}|}{\cellcolor{tabred!60}Min}
    & \small\textcolor{red!70!black}{E}
    & \small\textcolor{red!70!black}{$0.7064$}
    & \small\textcolor{red!70!black}{E}
    & \small\textcolor{red!70!black}{$0.9426$}
    & \small\textcolor{red!70!black}{E}
    & \small\textcolor{red!70!black}{$0.9654$}
    & \small\textcolor{red!70!black}{E}
    & \small\textcolor{red!70!black}{$0.7806$}
    & \small\textcolor{red!70!black}{E}
    & \small\textcolor{red!70!black}{$0.9405$}
    & \small\textcolor{red!70!black}{E}
    & \small\textcolor{red!70!black}{$0.9649$}
    \\ \hhline{|~|-------------|} 
    & \multicolumn{1}{@{}c@{}|}{\cellcolor{tabblue!60}Max}
    & \small\textcolor{blue!70!black}{\textbf{P}} $(n^{0.03})$
    & \small\textcolor{blue!70!black}{$\textbf{0.8178}$}
    & \small\textcolor{blue!70!black}{\textbf{P}} $(n^{-1.64})$
    & \small\textcolor{blue!70!black}{$\textbf{0.8881}$}
    & \small\textcolor{blue!70!black}{\textbf{P}} $(n^{-2.02})$
    & \small\textcolor{blue!70!black}{$\textbf{0.9101}$}
    & \small\textcolor{blue!70!black}{\textbf{P}} $(n^{-0.08})$
    & \small\textcolor{blue!70!black}{$\textbf{0.7308}$}
    & \small\textcolor{blue!70!black}{\textbf{P}} $(n^{-1.90})$
    & \small\textcolor{blue!70!black}{$\textbf{0.8911}$}
    & \small\textcolor{blue!70!black}{\textbf{P}} $(n^{-2.34})$
    & \small\textcolor{blue!70!black}{$\textbf{0.9274}$}
    \\ \hhline{|~|-------------|} 
    & \multicolumn{1}{@{}c@{}|}{\cellcolor{tabblue!60}Max}
    & \small\textcolor{red!70!black}{E}
    & \small\textcolor{red!70!black}{$0.6072$}
    & \small\textcolor{red!70!black}{E}
    & \small\textcolor{red!70!black}{$0.8221$}
    & \small\textcolor{red!70!black}{E}
    & \small\textcolor{red!70!black}{$0.9078$}
    & \small\textcolor{red!70!black}{E}
    & \small\textcolor{red!70!black}{$0.4577$}
    & \small\textcolor{red!70!black}{E}
    & \small\textcolor{red!70!black}{$0.7799$}
    & \small\textcolor{red!70!black}{E}
    & \small\textcolor{red!70!black}{$0.8833$}
    \\ \hline
    %%
    %------------------------------------------------------------------
    \multirow{4}{*}{$\lfloor \log_2 n \rfloor$}
    & \multicolumn{1}{@{}c@{}|}{\cellcolor{tabred!60}Min}
    & \small\textcolor{blue!70!black}{\textbf{P}} $(n^{-2.80})$
    & \small\textcolor{blue!70!black}{$\textbf{0.9685}$}
    & \small\textcolor{blue!70!black}{\textbf{P}} $(n^{-7.66})$
    & \small\textcolor{blue!70!black}{$\textbf{0.9547}$}
    & \small\textcolor{blue!70!black}{P} $(n^{-8.35})$
    & \small\textcolor{blue!70!black}{$0.9546$}
    & \small\textcolor{blue!70!black}{\textbf{P}} $(n^{-1.32})$
    & \small\textcolor{blue!70!black}{$\textbf{0.9940}$}
    & \small\textcolor{blue!70!black}{\textbf{P}} $(n^{-7.14})$
    & \small\textcolor{blue!70!black}{$\textbf{0.9514}$}
    & \small\textcolor{blue!70!black}{P} $(n^{-7.94})$
    & \small\textcolor{blue!70!black}{$0.9430$}
    \\ \hhline{|~|-------------|} 
    & \multicolumn{1}{@{}c@{}|}{\cellcolor{tabred!60}Min}
    & \small\textcolor{red!70!black}{E}
    & \small\textcolor{red!70!black}{$0.8465$}
    & \small\textcolor{red!70!black}{E}
    & \small\textcolor{red!70!black}{$0.9403$}
    & \small\textcolor{red!70!black}{E}
    & \small\textcolor{red!70!black}{$0.9656$}
    & \small\textcolor{red!70!black}{E}
    & \small\textcolor{red!70!black}{$0.8575$}
    & \small\textcolor{red!70!black}{E}
    & \small\textcolor{red!70!black}{$0.9458$}
    & \small\textcolor{red!70!black}{E}
    & \small\textcolor{red!70!black}{$0.9673$}
    \\ \hhline{|~|-------------|} 
    & \multicolumn{1}{@{}c@{}|}{\cellcolor{tabblue!60}Max}
    & \small\textcolor{blue!70!black}{\textbf{P}} $(n^{-1.22})$
    & \small\textcolor{blue!70!black}{$\textbf{0.8837}$}
    & \small\textcolor{blue!70!black}{\textbf{P}} $(n^{-2.15})$
    & \small\textcolor{blue!70!black}{$\textbf{0.9052}$}
    & \small\textcolor{blue!70!black}{\textbf{P}} $(n^{-2.45})$
    & \small\textcolor{blue!70!black}{$\textbf{0.9513}$}
    & \small\textcolor{blue!70!black}{P} $(n^{-0.02})$
    & \small\textcolor{blue!70!black}{$0.0113$}
    & \small\textcolor{blue!70!black}{\textbf{P}} $(n^{-1.72})$
    & \small\textcolor{blue!70!black}{$\textbf{0.8958}$}
    & \small\textcolor{blue!70!black}{\textbf{P}} $(n^{-2.11})$
    & \small\textcolor{blue!70!black}{$\textbf{0.9180}$}
    \\ \hhline{|~|-------------|} 
    & \multicolumn{1}{@{}c@{}|}{\cellcolor{tabblue!60}Max}
    & \small\textcolor{red!70!black}{E}
    & \small\textcolor{red!70!black}{$0.8319$}
    & \small\textcolor{red!70!black}{E}
    & \small\textcolor{red!70!black}{$0.8036$}
    & \small\textcolor{red!70!black}{E}
    & \small\textcolor{red!70!black}{$0.8917$}
    & \small\textcolor{red!70!black}{E}
    & \small\textcolor{red!70!black}{$0.1059$}
    & \small\textcolor{red!70!black}{E}
    & \small\textcolor{red!70!black}{$0.8078$}
    & \small\textcolor{red!70!black}{E}
    & \small\textcolor{red!70!black}{$0.8963$}
    \\ \hline
    %------------------------------------------------------------------
    \multirow{4}{*}{$\lfloor \sqrt n \rfloor$}
    & \multicolumn{1}{@{}c@{}|}{\cellcolor{tabred!60}Min}
    & \small\textcolor{blue!70!black}{\textbf{P}} $(n^{-3.15})$
    & \small\textcolor{blue!70!black}{$\textbf{0.9778}$}
    & \small\textcolor{blue!70!black}{\textbf{P}} $(n^{-7.84})$
    & \small\textcolor{blue!70!black}{$\textbf{0.9601}$}
    & \small\textcolor{blue!70!black}{P} $(n^{-8.51})$
    & \small\textcolor{blue!70!black}{$0.9476$}
    & \small\textcolor{blue!70!black}{P} $(n^{-1.35})$
    & \small\textcolor{blue!70!black}{$0.9950$}
    & \small\textcolor{blue!70!black}{\textbf{P}} $(n^{-7.14})$
    & \small\textcolor{blue!70!black}{$\textbf{0.9513}$}
    & \small\textcolor{blue!70!black}{P} $(n^{-7.93})$
    & \small\textcolor{blue!70!black}{$0.9438$}
    \\ \hhline{|~|-------------|} 
    & \multicolumn{1}{@{}c@{}|}{\cellcolor{tabred!60}Min}
    & \small\textcolor{red!70!black}{E}
    & \small\textcolor{red!70!black}{$0.8987$}
    & \small\textcolor{red!70!black}{E}
    & \small\textcolor{red!70!black}{$0.9525$}
    & \small\textcolor{red!70!black}{E}
    & \small\textcolor{red!70!black}{$0.9640$}
    & \small\textcolor{red!70!black}{E}
    & \small\textcolor{red!70!black}{$0.8650$}
    & \small\textcolor{red!70!black}{E}
    & \small\textcolor{red!70!black}{$0.9456$}
    & \small\textcolor{red!70!black}{E}
    & \small\textcolor{red!70!black}{$0.9676$}
    \\ \hhline{|~|-------------|} 
    & \multicolumn{1}{@{}c@{}|}{\cellcolor{tabblue!60}Max}
    & \small\textcolor{blue!70!black}{P} $(n^{-1.51})$
    & \small\textcolor{blue!70!black}{$0.9042$}
    & \small\textcolor{blue!70!black}{\textbf{P}} $(n^{-2.32})$
    & \small\textcolor{blue!70!black}{$\textbf{0.9478}$}
    & \small\textcolor{blue!70!black}{\textbf{P}} $(n^{-2.59})$
    & \small\textcolor{blue!70!black}{$\textbf{0.9314}$}
    & \small\textcolor{blue!70!black}{P} $(n^{-0.03})$
    & \small\textcolor{blue!70!black}{$0.0373$}
    & \small\textcolor{blue!70!black}{\textbf{P}} $(n^{-1.71})$
    & \small\textcolor{blue!70!black}{$\textbf{0.8942}$}
    & \small\textcolor{blue!70!black}{\textbf{P}} $(n^{-2.10})$
    & \small\textcolor{blue!70!black}{$\textbf{0.9218}$}
    \\ \hhline{|~|-------------|} 
    & \multicolumn{1}{@{}c@{}|}{\cellcolor{tabblue!60}Max}
    & \small\textcolor{red!70!black}{E}
    & \small\textcolor{red!70!black}{$0.9146$}
    & \small\textcolor{red!70!black}{E}
    & \small\textcolor{red!70!black}{$0.8687$}
    & \small\textcolor{red!70!black}{E}
    & \small\textcolor{red!70!black}{$0.8927$}
    & \small\textcolor{red!70!black}{E}
    & \small\textcolor{red!70!black}{$0.1711$}
    & \small\textcolor{red!70!black}{E}
    & \small\textcolor{red!70!black}{$0.8059$}
    & \small\textcolor{red!70!black}{E}
    & \small\textcolor{red!70!black}{$0.8979$}
    \\ \hline
\end{tabular}
\caption{\justifying Decay regime of the normalised second moment $\mathbb{E}[f_U^2]$ for $n\in\{2,\ldots,35\}$ ($m=2n$ modes). Each cell gives two rows: power-law fit (\textcolor{blue!70!black}{P}, exponent $n^{\hat\alpha}$, log--log $R^2$) and exponential fit (\textcolor{red!70!black}{E}, semi-log $R^2$). When a clear dominant regime can be highlighted (higher $R^2$), the text is \textbf{bold}. Min: $\ket{1,\ldots,1}$; Max: $\ket{n,0,\ldots,0}$. }
\label{tbl:Decay_Parity_Operators}
\end{table*}

\begin{figure}[h!]
    \centering
    \includegraphics[width=1.0\linewidth]{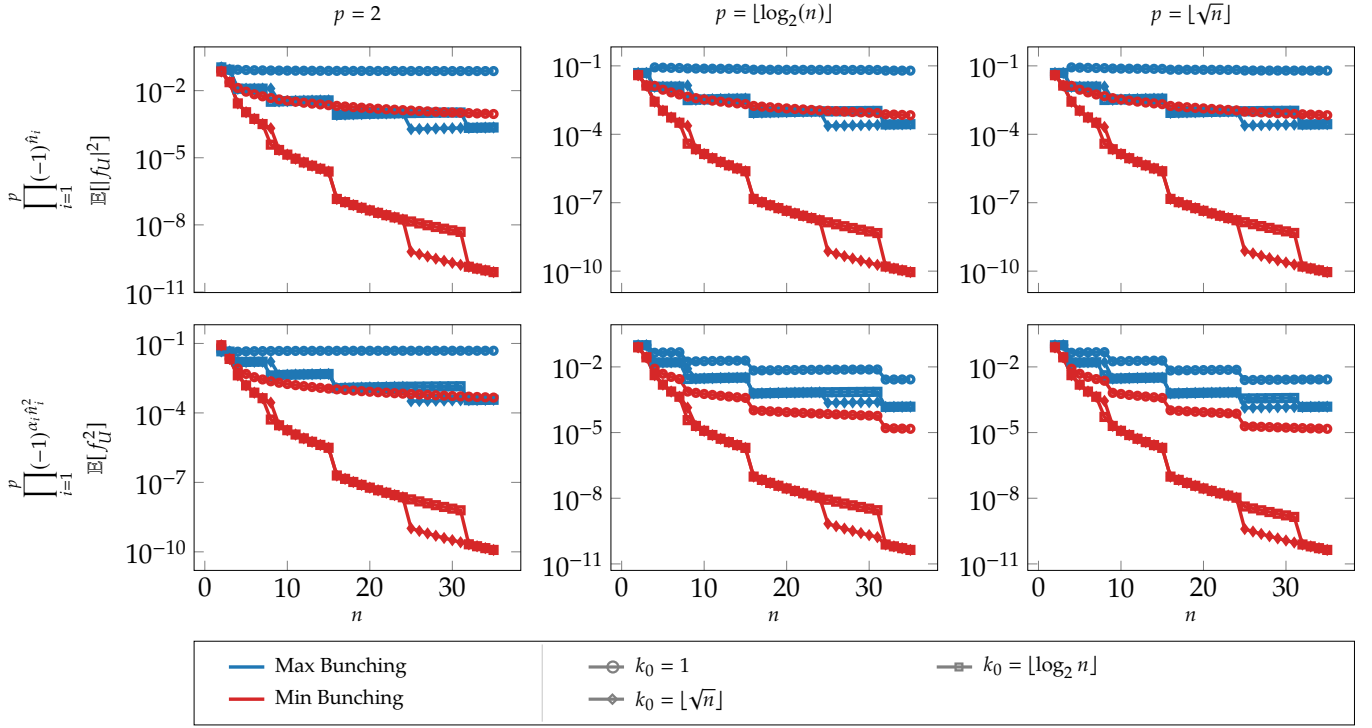}
    \caption{\justifying Study of the second moment of the parity operator $\prod_{i=1}^p (-1)^{\hat n_i}$, and of the polynomial parity operator $\prod_{i=1}^p (-1)^{\alpha_i \hat n_i^2}$, with $\alpha_i= 1/r_i$ and $r_i$ the i-th prime number for different evolution of $p$, corresponding to \cref{tbl:Decay_Parity_Operators}.}
    \label{fig:Second_Moments_Sum_Parity_Operators}
\end{figure}

In \cref{tbl:Decay_Parity_Operators}, the truncated signals do not vanish exponentially for the large truncation orders $k_0=\lfloor\log_2 n\rfloor$ and $k_0=\lfloor\sqrt{n}\rfloor$. They do, however, exhibit a faster polynomial decay than the untruncated signal. Consequently, one can use a surrogate model as defined in \cref{simtech:approx_gsim}, retaining only the irreps of polynomial size, at the cost of a polynomially small error given by the tail of the remaining irreps in \cref{eq:irrep_simulation_tail}.

The polynomial decay of the last irrep contribution depends on the input state. As explained in \cref{subsec:Classical_Sim_Parity_Op}, the case of maximal bunching, where all particles occupy the same mode, can be simulated classically. In \cref{tab:sec5_decay_scaling} we report the decay of $\Delta_{k_0}$ for two scaling of the input state: a fixed value of $\|R\|_2^2/n$, and the case where half the particles are in a single mode while the rest occupy distinct modes.

\begin{table*}[h!]
\centering
\begin{tabular}{|c|c|c|c|c|c|c|}
    \hline
    \multirow{2}{*}{Input}
    & \multicolumn{2}{c|}{$k_0=1$}
    & \multicolumn{2}{c|}{$k_0=\lfloor\log_2 n\rfloor$}
    & \multicolumn{2}{c|}{$k_0=\lfloor\sqrt{n}\rfloor$}
    \\ \cline{2-7}
    & Decay & $R^2$
    & Decay & $R^2$
    & Decay & $R^2$
    \\ \hline
    %------------------------------------------------------------------
    \multicolumn{1}{|@{}c@{}|}{\cellcolor{inpA}$\|R\|_2/n$}
    & \small\textcolor{blue!70!black}{\textbf{P} $(n^{-1.35})$}
    & \small\textcolor{blue!70!black}{\textbf{0.995}}
    & \small\textcolor{blue!70!black}{\textbf{P} $(n^{-7.24})$}
    & \small\textcolor{blue!70!black}{\textbf{0.950}}
    & \small\textcolor{blue!70!black}{P $(n^{-8.11})$}
    & \small\textcolor{blue!70!black}{0.937}
    \\ \arrayrulecolor{inpA}\cline{1-1}\arrayrulecolor{black}\cline{2-7}
    \multicolumn{1}{|@{}c@{}|}{\cellcolor{inpA}$\approx 0.275$}
    & \small\textcolor{red!70!black}{E}
    & \small\textcolor{red!70!black}{0.865}
    & \small\textcolor{red!70!black}{E}
    & \small\textcolor{red!70!black}{0.949}
    & \small\textcolor{red!70!black}{\textbf{E}}
    & \small\textcolor{red!70!black}{\textbf{0.970}}
    \\ \hline
    %------------------------------------------------------------------
    \multicolumn{1}{|@{}c@{}|}{\cellcolor{inpB}$\|R\|_2/n$}
    & \small\textcolor{blue!70!black}{\textbf{P} $(n^{-0.78})$}
    & \small\textcolor{blue!70!black}{\textbf{0.775}}
    & \small\textcolor{blue!70!black}{\textbf{P} $(n^{-5.90})$}
    & \small\textcolor{blue!70!black}{\textbf{0.965}}
    & \small\textcolor{blue!70!black}{\textbf{P} $(n^{-6.71})$}
    & \small\textcolor{blue!70!black}{\textbf{0.965}}
    \\ \arrayrulecolor{inpB}\cline{1-1}\arrayrulecolor{black}\cline{2-7}
    \multicolumn{1}{|@{}c@{}|}{\cellcolor{inpB}$\approx 0.310$}
    & \small\textcolor{red!70!black}{E}
    & \small\textcolor{red!70!black}{0.455}
    & \small\textcolor{red!70!black}{E}
    & \small\textcolor{red!70!black}{0.882}
    & \small\textcolor{red!70!black}{E}
    & \small\textcolor{red!70!black}{0.930}
    \\ \hline
    %------------------------------------------------------------------
    \multicolumn{1}{|@{}c@{}|}{\cellcolor{inpC}$\|R\|_2/n$}
    & \small\textcolor{blue!70!black}{\textbf{P} $(n^{-0.40})$}
    & \small\textcolor{blue!70!black}{\textbf{0.554}}
    & \small\textcolor{blue!70!black}{\textbf{P} $(n^{-4.87})$}
    & \small\textcolor{blue!70!black}{\textbf{0.948}}
    & \small\textcolor{blue!70!black}{\textbf{P} $(n^{-5.62})$}
    & \small\textcolor{blue!70!black}{\textbf{0.963}}
    \\ \arrayrulecolor{inpC}\cline{1-1}\arrayrulecolor{black}\cline{2-7}
    \multicolumn{1}{|@{}c@{}|}{\cellcolor{inpC}$\approx 0.378$}
    & \small\textcolor{red!70!black}{E}
    & \small\textcolor{red!70!black}{0.269}
    & \small\textcolor{red!70!black}{E}
    & \small\textcolor{red!70!black}{0.866}
    & \small\textcolor{red!70!black}{E}
    & \small\textcolor{red!70!black}{0.933}
    \\ \hline
    %------------------------------------------------------------------
    \multicolumn{1}{|@{}c@{}|}{\cellcolor{inpD}$\|R\|_2/n$}
    & \small\textcolor{blue!70!black}{\textbf{P} $(n^{-0.23})$}
    & \small\textcolor{blue!70!black}{\textbf{0.470}}
    & \small\textcolor{blue!70!black}{\textbf{P} $(n^{-4.12})$}
    & \small\textcolor{blue!70!black}{\textbf{0.948}}
    & \small\textcolor{blue!70!black}{\textbf{P} $(n^{-4.82})$}
    & \small\textcolor{blue!70!black}{\textbf{0.953}}
    \\ \arrayrulecolor{inpD}\cline{1-1}\arrayrulecolor{black}\cline{2-7}
    \multicolumn{1}{|@{}c@{}|}{\cellcolor{inpD}$\approx 0.457$}
    & \small\textcolor{red!70!black}{E}
    & \small\textcolor{red!70!black}{0.263}
    & \small\textcolor{red!70!black}{E}
    & \small\textcolor{red!70!black}{0.873}
    & \small\textcolor{red!70!black}{E}
    & \small\textcolor{red!70!black}{0.935}
    \\ \hline
    %------------------------------------------------------------------
    \multicolumn{1}{|@{}c@{}|}{\cellcolor{inpE}$\|R\|_2/n$}
    & \small\textcolor{blue!70!black}{\textbf{P} $(n^{-0.14})$}
    & \small\textcolor{blue!70!black}{\textbf{0.356}}
    & \small\textcolor{blue!70!black}{\textbf{P} $(n^{-3.49})$}
    & \small\textcolor{blue!70!black}{\textbf{0.948}}
    & \small\textcolor{blue!70!black}{\textbf{P} $(n^{-4.15})$}
    & \small\textcolor{blue!70!black}{\textbf{0.927}}
    \\ \arrayrulecolor{inpE}\cline{1-1}\arrayrulecolor{black}\cline{2-7}
    \multicolumn{1}{|@{}c@{}|}{\cellcolor{inpE}$\approx 0.542$}
    & \small\textcolor{red!70!black}{E}
    & \small\textcolor{red!70!black}{0.280}
    & \small\textcolor{red!70!black}{E}
    & \small\textcolor{red!70!black}{0.870}
    & \small\textcolor{red!70!black}{E}
    & \small\textcolor{red!70!black}{0.910}
    \\ \hline
    %------------------------------------------------------------------
    \multicolumn{1}{|@{}c@{}|}{\cellcolor{inpF}$\|R\|_2/n$}
    & \small\textcolor{blue!70!black}{P $(n^{-0.11})$}
    & \small\textcolor{blue!70!black}{0.245}
    & \small\textcolor{blue!70!black}{\textbf{P} $(n^{-2.99})$}
    & \small\textcolor{blue!70!black}{\textbf{0.927}}
    & \small\textcolor{blue!70!black}{\textbf{P} $(n^{-3.58})$}
    & \small\textcolor{blue!70!black}{\textbf{0.919}}
    \\ \arrayrulecolor{inpF}\cline{1-1}\arrayrulecolor{black}\cline{2-7}
    \multicolumn{1}{|@{}c@{}|}{\cellcolor{inpF}$\approx 0.629$}
    & \small\textcolor{red!70!black}{\textbf{E}}
    & \small\textcolor{red!70!black}{\textbf{0.289}}
    & \small\textcolor{red!70!black}{E}
    & \small\textcolor{red!70!black}{0.865}
    & \small\textcolor{red!70!black}{E}
    & \small\textcolor{red!70!black}{0.917}
    \\ \hline
    %------------------------------------------------------------------
    \multicolumn{1}{|@{}c@{}|}{\cellcolor{inpG}$\|R\|_2/n$}
    & \small\textcolor{blue!70!black}{P $(n^{-0.12})$}
    & \small\textcolor{blue!70!black}{0.502}
    & \small\textcolor{blue!70!black}{\textbf{P} $(n^{-2.60})$}
    & \small\textcolor{blue!70!black}{\textbf{0.920}}
    & \small\textcolor{blue!70!black}{\textbf{P} $(n^{-3.12})$}
    & \small\textcolor{blue!70!black}{\textbf{0.927}}
    \\ \arrayrulecolor{inpG}\cline{1-1}\arrayrulecolor{black}\cline{2-7}
    \multicolumn{1}{|@{}c@{}|}{\cellcolor{inpG}$\approx 0.722$}
    & \small\textcolor{red!70!black}{\textbf{E}}
    & \small\textcolor{red!70!black}{\textbf{0.631}}
    & \small\textcolor{red!70!black}{E}
    & \small\textcolor{red!70!black}{0.840}
    & \small\textcolor{red!70!black}{E}
    & \small\textcolor{red!70!black}{0.910}
    \\ \hline
    %------------------------------------------------------------------
    \multicolumn{1}{|@{}c@{}|}{\cellcolor{inpH}$\|R\|_2/n$}
    & \small\textcolor{blue!70!black}{P $(n^{-0.05})$}
    & \small\textcolor{blue!70!black}{0.092}
    & \small\textcolor{blue!70!black}{\textbf{P} $(n^{-2.19})$}
    & \small\textcolor{blue!70!black}{\textbf{0.901}}
    & \small\textcolor{blue!70!black}{\textbf{P} $(n^{-2.68})$}
    & \small\textcolor{blue!70!black}{\textbf{0.899}}
    \\ \arrayrulecolor{inpH}\cline{1-1}\arrayrulecolor{black}\cline{2-7}
    \multicolumn{1}{|@{}c@{}|}{\cellcolor{inpH}$\approx 0.807$}
    & \small\textcolor{red!70!black}{\textbf{E}}
    & \small\textcolor{red!70!black}{\textbf{0.222}}
    & \small\textcolor{red!70!black}{E}
    & \small\textcolor{red!70!black}{0.819}
    & \small\textcolor{red!70!black}{E}
    & \small\textcolor{red!70!black}{0.885}
    \\ \hline
    %------------------------------------------------------------------
    \multicolumn{1}{|@{}c@{}|}{\cellcolor{inpI}$\|R\|_2/n$}
    & \small\textcolor{blue!70!black}{P $(n^{-0.08})$}
    & \small\textcolor{blue!70!black}{0.231}
    & \small\textcolor{blue!70!black}{\textbf{P} $(n^{-2.02})$}
    & \small\textcolor{blue!70!black}{\textbf{0.890}}
    & \small\textcolor{blue!70!black}{\textbf{P} $(n^{-2.48})$}
    & \small\textcolor{blue!70!black}{\textbf{0.922}}
    \\ \arrayrulecolor{inpI}\cline{1-1}\arrayrulecolor{black}\cline{2-7}
    \multicolumn{1}{|@{}c@{}|}{\cellcolor{inpI}$\approx 0.911$}
    & \small\textcolor{red!70!black}{\textbf{E}}
    & \small\textcolor{red!70!black}{\textbf{0.360}}
    & \small\textcolor{red!70!black}{E}
    & \small\textcolor{red!70!black}{0.797}
    & \small\textcolor{red!70!black}{E}
    & \small\textcolor{red!70!black}{0.900}
    \\ \hline
    %------------------------------------------------------------------
    \multicolumn{1}{|@{}c@{}|}{\cellcolor{inpJ}$\|R\|_2/n$}
    & \small\textcolor{blue!70!black}{P $(n^{-0.04})$}
    & \small\textcolor{blue!70!black}{0.050}
    & \small\textcolor{blue!70!black}{\textbf{P} $(n^{-1.73})$}
    & \small\textcolor{blue!70!black}{\textbf{0.898}}
    & \small\textcolor{blue!70!black}{\textbf{P} $(n^{-2.16})$}
    & \small\textcolor{blue!70!black}{\textbf{0.914}}
    \\ \arrayrulecolor{inpJ}\cline{1-1}\arrayrulecolor{black}\cline{2-7}
    \multicolumn{1}{|@{}c@{}|}{\cellcolor{inpJ}$=1$}
    & \small\textcolor{red!70!black}{\textbf{E}}
    & \small\textcolor{red!70!black}{\textbf{0.198}}
    & \small\textcolor{red!70!black}{E}
    & \small\textcolor{red!70!black}{0.816}
    & \small\textcolor{red!70!black}{E}
    & \small\textcolor{red!70!black}{0.906}
    \\ \hline\hline
    %------------------------------------------------------------------
    % Specific example: half-one family (ket split across the two lines)
    \multicolumn{1}{|@{}c@{}|}{\cellcolor{EmphasisColorPurple}$\bigl|\lfloor n/2\rfloor, \underbrace{1,\ldots,1}, 0, \dots\bigr\rangle$}
    & \small\textcolor{blue!70!black}{\textbf{P} $(n^{-0.26})$}
    & \small\textcolor{blue!70!black}{\textbf{0.654}}
    & \small\textcolor{blue!70!black}{\textbf{P} $(n^{-4.07})$}
    & \small\textcolor{blue!70!black}{\textbf{0.949}}
    & \small\textcolor{blue!70!black}{\textbf{P} $(n^{-4.79})$}
    & \small\textcolor{blue!70!black}{\textbf{0.945}}
    \\ \arrayrulecolor{EmphasisColorPurple}\cline{1-1}\arrayrulecolor{black}\cline{2-7}
    \multicolumn{1}{|@{}c@{}|}{\cellcolor{EmphasisColorPurple}$n-\lfloor n/2\rfloor$}
    & \small\textcolor{red!70!black}{E}
    & \small\textcolor{red!70!black}{0.518}
    & \small\textcolor{red!70!black}{E}
    & \small\textcolor{red!70!black}{0.879}
    & \small\textcolor{red!70!black}{E}
    & \small\textcolor{red!70!black}{0.933}
    \\ \hline
\end{tabular}
\caption{\justifying Decay of $\Delta_{k_0}(\rho,O)$ with photon number $n$ for the number-phase observable with $p =\lfloor\sqrt{n}\rfloor$ and $m=2n$ modes, across input Fock states grouped by their bunching parameter $\|R\|_2/n$ (from minimum bunching, top, to maximum bunching, $\|R\|_2/n=1$, bottom). For each input family and each threshold $k_0$, a power-law fit $\sim n^\alpha$ (\textcolor{blue!70!black}{P}, blue) and an exponential fit $\sim e^{-\beta n}$ (\textcolor{red!70!black}{E}, red) are compared; the better fit (higher $R^2$, in log space) is shown in bold. The double rule separates the fixed-$\|R\|_2/n$ families (nearest-partition approach) from the specific example $|\lfloor n/2\rfloor,1,\ldots,1,0, \dots \rangle$, in which half the photons are bunched in a single mode.}
\label{tab:sec5_decay_scaling}
\end{table*}

Using the results from \cref{tab:sec5_decay_scaling}, we see that the contribution of the last irrep always appears to be polynomially smaller than that of the polynomially-large irreps, although the degree of the polynomial decay depends on the input and can be rather small. We leave open the question of the impact of such a polynomial error on a surrogate model based on \cref{simtech:approx_gsim}.

\begin{figure}[h!]
    \centering
    \includegraphics[width=1.0\linewidth]{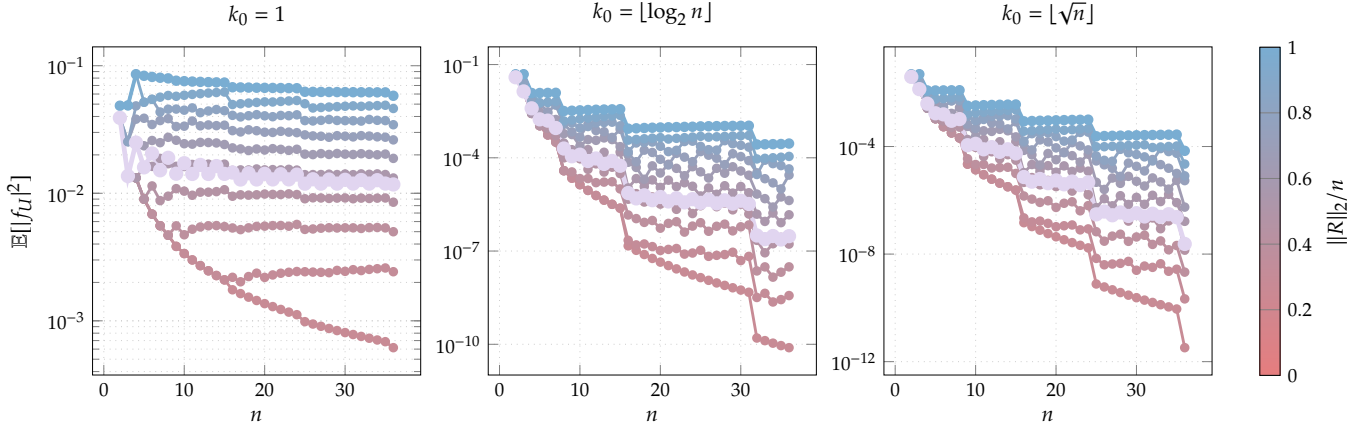}
    \caption{\justifying \justifying Decay of $\Delta_{k_0}(\rho,O)$ with photon number $n$ for the number-phase observable with $p =\lfloor\sqrt{n}\rfloor$ and $m=2n$ modes corresponding to \cref{tab:sec5_decay_scaling}.}
    \label{fig:Second_Moments_all_fock_states}
\end{figure}
\newpage
\section{On the plausibility of non-negligible signal from large irreps}\label{app:pos}

Let \(m=cn\), where \(c>1\) is fixed and \(cn\in\mathbb N\), and denote
\begin{equation}
    |\Phi_n^m|
    =
    \binom{m+n-1}{n},
    \qquad
    d_k^{(n)}
    =
    \frac{2k+m-1}{m-1}
    \binom{k+m-2}{k}^{\!2}.
\end{equation}
The contribution of the \(k\)-th irrep to the Haar variance is
\begin{equation}
    \Gamma_k(\rho,O)
    =
    \frac{\bigl\|P_k^{(n)}(\rho)\bigr\|_2^2\bigl\|P_k^{(n)}(O)\bigr\|_2^2}{d_k^{(n)}}.
    \label{eq:single_irrep_variance_contribution}
\end{equation}
Since \(P_k^{(n)}\) is Hilbert--Schmidt orthogonal and
\(\|O\|_\infty=1\),
\begin{equation}
    \bigl\|P_k^{(n)}(O)\bigr\|_2^2
    \leq
    \|O\|_2^2
    \leq
    |\Phi_n^m|.
\end{equation}
Moreover, \(\bigl\|P_k^{(n)}(\rho)\bigr\|_2^2\leq\Tr(\rho^2)\leq1\), and hence
\begin{equation}
    \Gamma_k(\rho,O)
    \leq
    \bigl\|P_k^{(n)}(\rho)\bigr\|_2^2\frac{|\Phi_n^m|}{d_k^{(n)}}
    \leq
    \frac{|\Phi_n^m|}{d_k^{(n)}}.
    \label{eq:dimensional_variance_bound}
\end{equation}

For \(k=\beta n\), with \(\beta\in(0,1]\) fixed, Stirling's formula gives
\begin{align}
    \log |\Phi_n^m|
    &=
    n s_c-\frac12\log n+\mathcal O_c(1),
    \\
    \log d_{\beta n}^{(n)}
    &=
    2n h_c(\beta)-\log n+\mathcal O_{c,\beta}(1),
\end{align}
where
\begin{align}
    s_c
    &=
    (c+1)\log(c+1)-c\log c,
    \\
    h_c(\beta)
    &=
    (c+\beta)\log(c+\beta)
      -c\log c-\beta\log\beta.
\end{align}
Because
\begin{equation}
    h_c'(\beta)
    =
    \log\!\left(1+\frac{c}{\beta}\right)>0,
\end{equation}
there exists a unique constant \(\alpha_c\in(0,1)\) satisfying
\begin{equation}
    2h_c(\alpha_c)=s_c.
    \label{eq:critical_irrep_fraction}
\end{equation}
Consequently, for every fixed \(\beta>\alpha_c\),
\begin{equation}
    \frac{|\Phi_n^m|}{d_{\beta n}^{(n)}}
    =
    \exp\!\left[
        -n\gamma_c(\beta)+\mathcal O(\log n)
    \right],
    \qquad
    \gamma_c(\beta)
    :=
    2h_c(\beta)-s_c>0.
\end{equation}
Combining this with \eqref{eq:dimensional_variance_bound} yields
\begin{equation}
    \boxed{
    \Gamma_{\beta n}(\rho,O)
    \leq
    \exp\!\left[
        -n\gamma_c(\beta)+\mathcal O(\log n)
    \right].
    }
\end{equation}
Thus, independently of the input state and of the structure of the
bounded observable, every irrep with
\(\lim_{n\to\infty}k_n/n>\alpha_c\) has an exponentially-vanishing
contribution to the variance. More uniformly, for every \(\varepsilon>0\),
\begin{equation}
    \sup_{k\geq(\alpha_c+\varepsilon)n}
    \Gamma_k(\rho,O)
    \leq
    \exp\!\left[-\gamma_{c,\varepsilon}n+\mathcal O(\log n)\right],
\end{equation}
where
\begin{equation}
    \gamma_{c,\varepsilon}
    =2h_c(\alpha_c+\varepsilon)-s_c>0.
\end{equation}
No such exponential conclusion follows at the boundary
\(k=\alpha_cn+\mathcal O(\log n)\), where \(|\Phi_n^m|/d_k^{(n)}\) is only
polynomial in \(n\).

Consider instead
\begin{equation}
    k_n=\left\lceil\kappa\sqrt n\right\rceil,
    \qquad \kappa>0.
\end{equation}
Stirling's formula in the regime \(k_n=o(m)\) gives
\begin{equation}
    \log d_{k_n}^{(n)}
    =
    \kappa\sqrt n\log n
    +
    2\kappa\sqrt n
    \log\!\left(\frac{ec}{\kappa}\right)
    +
    \mathcal O_{c,\kappa}(\log n).
\end{equation}
Therefore,
\begin{equation}
    \boxed{
    \frac{d_{k_n}^{(n)}}{|\Phi_n^m|}
    =
    \exp\!\left[
        -s_cn
        +\kappa\sqrt n\log n
        +\mathcal O_{c,\kappa}(\sqrt n)
    \right]
    =
    e^{-s_cn+o(n)}.
    }
    \label{eq:sqrtn_dimension_ratio}
\end{equation}
Hence the dimensional bound does not force the
\(k_n=\Theta(\sqrt n)\) contribution to vanish exponentially.

More precisely, if one requires
\begin{equation}
    \Gamma_{k_n}(\rho,O)\geq n^{-r},
\end{equation}
for some fixed \(r>0\), then
\eqref{eq:single_irrep_variance_contribution} implies
\begin{equation}
    \bigl\|P_{k_n}^{(n)}(\rho)\bigr\|_2^2
    \geq
    n^{-r}\frac{d_{k_n}^{(n)}}{\bigl\|P_{k_n}^{(n)}(O)\bigr\|_2^2}.
    \label{eq:exact_required_state_purity}
\end{equation}
Consequently, every such construction must satisfy
\begin{equation}
    \boxed{
    \bigl\|P_{k_n}^{(n)}(\rho)\bigr\|_2^2
    \geq
    n^{-r}\frac{d_{k_n}^{(n)}}{|\Phi_n^m|}
    =
    \exp\!\left[
        -s_cn
        +\kappa\sqrt n\log n
        +\mathcal O_{c,\kappa}(\sqrt n)
    \right].
    }
    \label{eq:necessary_sqrtn_state_purity}
\end{equation}
This necessary lower bound is itself exponentially-small. Thus an inverse-polynomial variance contribution at \(k=\Theta(\sqrt n)\) is compatible with an exponentially small state irrep purity: it can, in principle, be compensated by a sufficiently large ratio \(\bigl\|P_k^{(n)}(O)\bigr\|_2^2/d_k^{(n)}\). \cref{eq:necessary_sqrtn_state_purity} is only a necessary norm-budget condition; realizing it requires an observable whose \(k\)-th irrep component nearly saturates the available
Hilbert--Schmidt weight.

\end{document}